\documentclass[11pt,english]{article}
\usepackage[verbose=true,letterpaper]{geometry}
\AtBeginDocument{
	\newgeometry{
		textheight=9in,
		textwidth=6.5in,
		top=1in,
		headheight=14pt,
		headsep=25pt,
		footskip=30pt
	}
}
\usepackage{float}
\usepackage[parfill]{parskip}

\usepackage{multirow}
\usepackage{fullpage}
\usepackage{amsfonts}
\usepackage{graphicx}
\usepackage{mathrsfs}
\usepackage{amsmath}
\usepackage{amssymb}
\usepackage{amsthm}
\numberwithin{equation}{section}
\usepackage{tabularx}
\usepackage[title]{appendix}
\usepackage{enumerate}
\usepackage{adjustbox}
\usepackage{booktabs}
\usepackage{hyperref}
\hypersetup{
	colorlinks=true,
	linkcolor=red,
	filecolor=magenta,      
	urlcolor=cyan,
	citecolor = blue
}
\usepackage{dsfont}
\usepackage[boxed,ruled,commentsnumbered]{algorithm2e}
\newtheorem{theorem}{Theorem}[section]
\newtheorem{lemma}{Lemma}[section]

\newtheorem{proposition}{Proposition}[section]
\newtheorem{remark}{Remark}[section]
\newtheorem{corollary}{Corollary}[section]

\usepackage[round]{natbib}
\usepackage{authblk}
\usepackage{caption}
\usepackage{algorithmic}

\newcommand{\btheta}{\boldsymbol{\theta}}
\newcommand{\bbeta}{\boldsymbol{\beta}}
\newcommand{\bdelta}{\boldsymbol{\delta}}

\newcommand{\hDelta}{\widehat{\Delta}}
\newcommand{\bEta}{\boldsymbol{\eta}}

\newcommand{\bmu}{\boldsymbol{\mu}}
\newcommand{\hmu}{\widehat{\boldsymbol{\mu}}}
\newcommand{\tmu}{\widetilde{\boldsymbol{\mu}}}

\newcommand{\bgamma}{\boldsymbol{\gamma}}
\newcommand{\hgamma}{\widehat{\boldsymbol{\gamma}}}
\newcommand{\tgamma}{\widetilde{\boldsymbol{\gamma}}}
\newcommand{\cgamma}{\check{\boldsymbol{\gamma}}}

\newcommand{\balpha}{\boldsymbol{\alpha}}
\newcommand{\halpha}{\widehat{\boldsymbol{\alpha}}}
\newcommand{\talpha}{\widetilde{\boldsymbol{\alpha}}}

\newcommand{\bxi}{\boldsymbol{\xi}}
\newcommand{\txi}{\widetilde{\bxi}}
\newcommand{\bnu}{\boldsymbol{\nu}}
\newcommand{\tnu}{\widetilde{\boldsymbol{\nu}}}

\newcommand{\bX}{\boldsymbol{X}}
\newcommand{\bA}{\mathbf{A}}
\newcommand{\bQ}{\mathbf{Q}}

\newcommand{\bI}{\boldsymbol{\mathrm{I}}}
\newcommand{\dI}{\boldsymbol{\mathds{I}}}
\newcommand{\bu}{\boldsymbol{u}}
\newcommand{\bb}{\boldsymbol{b}}
\newcommand{\bB}{\boldsymbol{B}}
\newcommand{\tB}{\widetilde{\boldsymbol{B}}}
\newcommand{\bC}{\mathbf{C}}
\newcommand{\bD}{\mathbf{D}}

\newcommand{\bZ}{\boldsymbol{Z}}
\newcommand{\bU}{\boldsymbol{U}}
\newcommand{\bG}{\mathbf{G}}

\def \bdelta{\boldsymbol{\delta}}

\newcommand{\hbtheta}{\widehat{\boldsymbol{\theta}}}
\newcommand{\tbtheta}{\widetilde{\boldsymbol{\theta}}}

\newcommand{\R}{\mathbb{R}^p}
\newcommand{\E}{\mathbb{E}}
\newcommand{\Prob}{\mathbb{P}}

\newcommand{\bM}{\mathbf{M}}
\newcommand{\hbM}{\widehat{\mathbf{M}}}
\def \tbM{\widetilde{\mathbf{M}}}
\newcommand{\bSigma}{\boldsymbol{\Sigma}}

\newcommand{\hbSigma}{\widehat{\boldsymbol{\Sigma}}}
\newcommand{\Op}{O_{\mathbb{P}}}

\def \bx{\boldsymbol{x}}

\newcommand{\infnorm}[1]{|#1|_{\infty}}
\newcommand{\LRinfnorm}[1]{\left|#1\right|_{\infty}}

\newcommand{\twonorm}[1]{\|#1\|_{2}}
\newcommand{\LRtwonorm}[1]{\left\|#1\right\|_{2}}
\newcommand{\onenorm}[1]{\|#1\|_{1}}

\newcommand{\abs}[1]{|#1|}
\newcommand{\LRabs}[1]{\left|#1\right|}

\newcommand{\trans}[1]{#1^{\top}}

\newcommand{\cond}[1]{$(\mathbf{C}#1)$}

\newcommand{\LRs}[1]{\left(#1\right)}
\newcommand{\LRm}[1]{\left[#1\right]}
\newcommand{\LRl}[1]{\left\{#1\right\}}

\usepackage{authblk}

\title{\bf Varying Coefficient Linear Discriminant Analysis for Dynamic Data}
\author[1]{Yajie Bao\thanks{Email: \texttt{baoyajie2019stat@sjtu.edu.cn}}}
\author[1]{Yuyang Liu\thanks{Email: \texttt{d0408x@sjtu.edu.cn}}}

\affil[1]{School of Mathematical Sciences\\Shanghai Jiao Tong University}

\begin{document}
\maketitle
\begin{abstract}
	Linear discriminant analysis (LDA) is an important classification tool in statistics and machine learning. This paper investigates the varying coefficient LDA model for dynamic data, with Bayes' discriminant direction being a function of some exposure variable to address the heterogeneity. We propose a new least-square estimation method based on the B-spline approximation. The data-driven discriminant procedure is more computationally efficient than the dynamic linear programming rule \citep{jiang2020dynamic}. We also establish the convergence rates for the corresponding estimation error bound and the excess misclassification risk. The estimation error in $L_2$ distance is optimal for the low-dimensional regime and is near optimal for the high-dimensional regime. Numerical experiments on synthetic data and real data both corroborate the superiority of our proposed classification method.
\end{abstract}

\section{Introduction}
Classification is one of the most essential topics in statistics and machine learning, and widely applied in many scientific and industrial fields. Consider a pair of random variables $(\bX, Y)$, where $\bX \in \R$ is the covariate and $Y\in \{0,1\}$ is the label. If $Y=1$, the covariate $\bX$ follows the $p$-dimensional multivariate normal distribution $\mathcal{N}(\bmu_1, \bSigma)$, otherwise $\bX$ is distributed as $\mathcal{N}(\bmu_2, \bSigma)$. We assume the prior probabilities of two classes are equal, that is $\Prob(Y = 1) = \Prob(Y=0) = 1/2$. For a new random covariate $\bX_{\text{new}}$, we aim to predict its unknown label $Y_{\text{new}}$ according to some discriminant rule. If we know the parameters $\bmu_{1}$, $\bmu_{2}$ and $\bSigma$ in advance, let $\bmu=(\bmu_1+\bmu_2)/2$, then the well-known Bayes' linear discriminant rule is given by
\begin{equation}\label{1.1}
	 \psi(\bX_{\text{new}}) = \dI\left((\bX_{\text{new}}-\bmu)^{\top} \bSigma^{-1}(\bmu_1-\bmu_2) \geq 0\right)
\end{equation}
where $\bSigma^{-1}(\bmu_1-\bmu_2)$ is called Bayes' discriminant direction. In real data analysis, a data-driven Bayes' classification rule is given by plugging sample means $\hmu_{1}$, $\hmu_{2}$ and pooled sample covariance matrix $\hbSigma$ in \eqref{1.1}, which is asymptotically optimal when the dimensionality $p$ is fixed \citep{anderson1958introduction}. 

Driven by contemporary measurement technologies, high-dimensional data sets have been broadly collected in classification problems. Classical LDA has been proved to perform poorly (no better than random guessing) in the high-dimensional setting, especially when the dimension is much larger than the sample size \citep{bickel2004some}. To address high-dimensional issue, the sparsity assumption is introduced to LDA. Several proposed methods assumed that both $\bSigma^{-1}$ and $\bmu_{1} - \bmu_{2}$ have sparse structures. For example, \cite{shao2011sparse} used the thresholding procedure to estimate $\bSigma^{-1}$ and $\bmu_{1} - \bmu_{2}$ separately, then constructed a plug-in sparse Bayes' linear discriminant rule. Similar regularized methods can also be found in \cite{guo2007regularized, wu2009sparse, witten2009covariance}, etc. In addition, some works only assumed Bayes' linear discriminant direction $\bSigma^{-1}(\bmu_{1} - \bmu_{2})$ is sparse. \cite{cai2011direct} proposed the linear programming discriminant (LPD) rule by directly estimating the product $\bSigma^{-1}(\bmu_{1} - \bmu_{2})$ through constrained $\ell_1$ minimization. Recently, \cite{tony2019high} proposed an adaptive LPD procedure that achieved the minimax optimal convergence rate of estimation error and excess misclassification risk in high-dimensional case. \cite{mai2012direct} estimated the sparse discriminant direction via a sparse penalized least squares formulation. \cite{MAI2015175} studied high-dimensional sparse semiparametric discriminant analysis and relaxed the Gaussian assumption. For multiclass problem, \cite{mai2019multiclass} proposed a sparse discriminant procedure by estimating all discriminant directions simultaneously.

Heterogeneous data is widespread in many modern scientific fields, such as finance, biology, and astronomy \citep{fan2014challenges}. The prevalent statistical approach to address the heterogeneity is imposing the dynamic or varying coefficient assumption, where the population means and covariance matrix may vary with some observable exposure variable. In specific, \cite{chen2019new,chen2016dynamic,wang2021nonparametric} investigated the dynamic covariance model in the high-dimensional regime. Under the dynamic setting, Bayes' discriminant direction is a function of the exposure variable. Consequently, classical plug-in Bayes' discriminant rule will deteriorate in analyzing non-static data, and thus leads to unsatisfactory performance. To address the dynamic data, \cite{jiang2020dynamic} proposed the dynamic linear programming discriminant (DLPD) rule by assuming $\bmu_1$, $\bmu_2$ and $\bSigma$ are functions of some $q$-dimensional random covariate $\boldsymbol{U}$. To estimate the sparse Fisher’s linear discriminant direction function $\bbeta^{*}(\bu) = \bSigma^{-1}(\bu)(\bmu_{1}(\bu) - \bmu_{2}(\bu))$ given $\bU = \bu$, they first used the Nadaraya-Watson method to obtain estimators $\hmu_{1}(\bu)$, $\hmu_{2}(\bu)$ and $\hbSigma(\bu)$. Then they estimated $\boldsymbol{\beta}^{*}(\bu)$ using the linear programming approach \cite{cai2011direct, candes2007}:
\begin{equation}\label{DLPD}
	\widehat{\bbeta}(\bu) = \arg\min_{\bbeta \in \R} \left\{\onenorm{\bbeta}\text{ subject to }\infnorm{\hbSigma(\bu)\bbeta-\left(\hmu_{1}(\bu)-\hmu_{2}(\bu)\right)}\leq \lambda_n\right\},
\end{equation}
where $\lambda_n$ is a tuning parameter. However, this classification procedure is computationally expensive for large scale prediction problem. For each new observation $(\bX_{\text{new}}, \boldsymbol{U}_{\text{new}})$, DLPD method needs to re-estimate $\hmu_{1}(\boldsymbol{U}_{\text{new}})$, $\hmu_{2}(\boldsymbol{U}_{\text{new}})$ and $\hbSigma(\boldsymbol{U}_{\text{new}})$ and re-solve the corresponding large scale linear programming (\ref{DLPD}). In addition, the support set of discriminant direction $\bbeta^{*}(\bu)$ decides which variable contributes to classification but \eqref{DLPD} can not provide a invariant support set since it is a point-wise estimator. In some real applications, the varying support set of discriminant direction in DLPD method may lack interpretability.

The dynamic discriminant analysis shares the same semi-parametric spirit with the classical varying coefficient model \citep{hastie1993varying}, where the unknown parameters are assumed to be a smooth function of the exposure variable. In the past decades, the varying coefficient method has been applied to a variety of statistical models, such as linear regression model \citep{hoover1998nonparametric,fan1999statistical}, generalized linear model \citep{cai2000efficient, fan2008statistical}, quantile regression \citep{honda2004quantile, wang2009quantile} and support vector machine \citep{LU2018107}, etc. Motivated by the least square form of Bayes' discriminant direction, we propose a new estimation method for the discriminant direction function based on B-spline approximation, which can be applied in the classification for dynamic data. In high-dimensional regime, we can estimate the approximation coefficient by solving a penalized least square problem. The computational drawback of the DLPD rule \citep{jiang2020dynamic} is circumvented in our developed varying coefficient discriminant procedure. For each new observation, we only need to re-compute the B-spline basis vector. Hence it has a significant computational advantage over the DLPD rule. In the high-dimensional case, the support set of our proposed estimator is irrelevant with the value of exposure variable, which is indeed helpful to select important features contributing to classification.

The remainder of this paper is organized as follows. In section \ref{sec_method}, we propose a new discriminant direction function and its varying coefficient estimators in both low-dimensional and high-dimensional regimes. In section \ref{sec_result}, we establish the upper bounds for the estimation error and uniform excess misclassification risk for our proposed varying coefficient LDA procedure. In section \ref{sec_simu} and \ref{sec_realdata}, we verify the performance of our method through simulations on synthetic data and real data respectively.

\paragraph{Notations.}
We define some notations that will be used throughout the paper. For two real positive sequences $a_n$ and $b_n$, we write $a_n\lesssim b_n$ if there exists some positive constant $m$ such that $a_n\leq m b_n$. And we write $a_{n} \asymp b_{n}$ if $a_n\lesssim b_n$ and $b_n\lesssim a_n$. For a real-valued vector $\bx \in \mathbb{R}^p$, we use $\onenorm{\bx} = \sum_{j=1}^p|x_j|$, $\twonorm{\bx} = (\sum_{j=1}^p|x_j|^2)^{1/2}$ and $\infnorm{\bx} = \max_{1\leq j\leq p}|x_j|$ to denote the $\ell_1$, $\ell_2$ and $\ell_{\infty}$ norm respectively. For a subset $S\subseteq \{1,2,...,p\}$, we use $\bx_{S}$ to denote the sub-vector $(x_j:j\in S)$. Specially, for a vector $\bb \in \mathbb{R}^{pq}$, we write sub-vector $\bb_{(j)} = (b_{(j-1)q+1},\cdots,b_{jq})^{\top}$ for $j=1,2,...,p$. And for a subset $S\subseteq \{1,2,...,p\}$, we use $\bb_{(S)}$ to denote the group sub-vector $(\bb_{(j)}:j\in S)$. For a real-valued matrix $\bA\in \mathbb{R}^{p\times q}$, $\ell_2$ (spectral) norm is defined by $\|\bA\|_2 = \sup_{\twonorm{\boldsymbol{x}}=1,\twonorm{\boldsymbol{y}}=1} |\boldsymbol{x}^{\top}\bA\boldsymbol{y}|$, the maximal entry in absolute value is denoted by $|\bA|_{\infty} = \max_{i,j}|A_{ij}|$. For two subsets $S\subseteq \{1,2,...,p\}$ and $T\subseteq \{1,2,...,q\}$, we write sub-matrix $\bA_{ST} = (A_{ij})$ for $i\in S$ and $j\in T$. We use $\bA \otimes \bB$ to denote the Kronecker product on two matrices $\bA$ and $\bB$ with proper sizes. Specially, for a matrix $\bA \in \mathbb{R}^{pq\times pq}$ and two subsets $S, T\subseteq \{1,2,...,p\}$, we write the group sub-matrix as $\bA_{(ST)} = (A_{ij})$ for $i\in \{(k-1)q+1,...,kq: k\in S\}$ and $j\in \{(k-1)q+1,...,kq: k\in T\}$. For a sequence of real random variables $X_n$, we write $X_n = \Op(a_n)$ if for any $\epsilon >0$, there exists some constant $C > 0$ such that $\mathbb{P}(|X_n| > Ca_n) < \epsilon$.

\section{Varying coefficient LDA via B-spline approximation}\label{sec_method}
In this section, we provide a detailed description of the varying coefficient linear discriminant rule. Given the univariate exposure variable $U=u \in [0,1]$, we assume $\bX \sim \mathcal{N}(\bmu_1(u),\bSigma(u))$ if $Y=1$ and $\bX \sim \mathcal{N}(\bmu_2(u),\bSigma(u))$ if $Y=0$, then Bayes' discriminant direction is $\bbeta^{*}(u) = \bSigma^{-1}(u)(\bmu_{1}(u)-\bmu_{2}(u))$. We also denote the pooled mean as $\bmu(u) = \pi_1\bmu_1(u)+\pi_2\bmu_2(u)$, where $\pi_1 = \Prob(Y=1)$ and $\pi_2 = \Prob(Y=0)$. To introduce our new discriminant direction function, we define a new response variable as $Z = \pi_2$ if $Y = 1$ and $Z=-\pi_1$ if $Y = 0$. In addition, the exposure variable $U$ is assumed to be independent with the label $Y$. Motivated by the least square form of the plug-in Bayes' discriminant direction $\hbSigma^{-1}(\hmu_1 - \hmu_2)$ in static setting \citep{anderson1958introduction,mai2012direct}, we propose a new discriminant direction function $\btheta^{*}(U) = (\theta_1^{*}(U),\cdots,\theta_p^{*}(U))^{\top}$ as the minimizer of the following population least square problem
\begin{equation}\label{pop}
	 \min_{\theta_j(U)\in \mathcal{L}^2(\mathcal{P})}\E\LRm{\LRs{Z-\sum_{j=1}^p\theta_{j}(U)(X_j-\mu_{j}(U))}^2\big|U},
\end{equation}
where $\mathcal{P}$ is the joint distribution of $(\bX,Z,U)$ and $\mathcal{L}^2(\mathcal{P})$ denotes the $L^2$ space under measure $\mathcal{P}$. It is worthwhile noting that the representation \eqref{pop} is similar to the approximation of coefficient function in the varying coefficient linear model. Then the discriminant direction function $\btheta^{*}(U)$ satisfies
\begin{equation*}
		\E\LRl{(\bX-\bmu(U))(Z-(\bX-\bmu(U))^{\top}\btheta^{*}(U))| U}= \boldsymbol{0}.
\end{equation*}
A further computation gives rise to the following closed form,
\begin{equation*}
	\btheta^{*}(U) = \pi_1\pi_2 \bSigma^{-1}(U)(\bmu_1(U)-\bmu_2(U))[1-(\bmu_1(U)-\bmu_2(U))^{\top}\btheta^{*}(U)].
\end{equation*}
If the population covariance matrix $\bSigma(U)$ is positive definite and $\bmu_1(U)-\bmu_2(U)\neq 0$, we are guaranteed that $(\bmu_1(U)-\bmu_2(U))^{\top}\btheta^{*}(U)\in (0,1)$. As a consequence, Bayes' discriminant direction function satisfies that
\begin{equation}\label{relation}
	\bbeta^{*}(U) =  \bSigma^{-1}(U)(\bmu_1(U)-\bmu_2(U)) = c^{*}(U)\btheta^{*}(U),
\end{equation}
where $c^{*}(U) = 1/[\pi_1\pi_2-\pi_1\pi_2(\bmu_1(U)-\bmu_2(U))^{\top}\btheta^{*}(U)]$. 
For the equal-prior case ($\pi_1 = \pi_2$), given any new observation $(\bX_{\text{new}}, U_{\text{new}})$, we define the oracle varying coefficient discriminant rule as
\begin{equation}\label{rule}
	 \psi(\bX_{\text{new}}, U_{\text{new}}) = \dI\left((\bX_{\text{new}}-\bmu(U_{\text{new}}))^{\top} \btheta^{*}(U_{\text{new}}) \geq 0\right).
\end{equation}
Recall that $c^{*}(u)>0$, the classification result of (\ref{rule}) is consistent with using Bayes' discriminant direction $\bbeta^{*}(U_{\text{new}})$. 

\subsection{Approximation of discriminant direction function}
 Let $\bB(\cdot) = (B_1(\cdot),...,B_{L_n}(\cdot))$ be the scaled B-spline basis of the polynomial splines space, which satisfies that $B_k(\cdot) \geq 0$ and $\sum_{k=1}^{L_n}B_k(\cdot) = \sqrt{L_n}$. According to the B-spline approximation theory \citep{de1978practical}, under some regular conditions, each coordinate of discriminant direction $\btheta^*(u)$ can be approximated by $\theta_{j}^*(u) \approx \bgamma_{(j)}^{\top}\bB(u)$, where $\bgamma_{(j)} \in \mathbb{R}^{L_n}$ is the approximation coefficient. If $\bmu_{1}(u)$ and $\bmu_{2}(u)$ are known, the ``best'' approximation coefficients in population form is defined as
\begin{equation}\label{spline_p1}
	(\tgamma_{(1)}, \cdots, \tgamma_{(p)}) = \arg\min_{\substack{\bgamma_{(j)} \in \mathbb{R}^{L_n},\\ 1\leq j \leq p}}\E\LRm{\LRs{Z-\sum_{j=1}^{p}(X_j-\mu_{j}(U))\bgamma_{j}^{\top}\bB(U)}^2}.
\end{equation}
Let $\tgamma= (\tgamma_{(1)}^{\top}, \cdots, \tgamma_{(p)}^{\top})^{\top}$ and $\tB(U) = (\bX-\bmu(U))\otimes \bB(U)$, it is easy to show that
\begin{equation*}
	\tgamma =\LRs{\E[\tB(U)\tB(U)^{\top}]}^{-1}\E[\tB(U) Z].
\end{equation*}
For any $u\in [0,1]$, we may write the approximated discriminant direction as
\begin{equation*}\label{dis_spline_pop}
	\tbtheta(u) = \left(\tgamma_{(1)}^{\top}\bB(u),\cdots,\tgamma_{(p)}^{\top}\bB(u)\right)^{\top}.
\end{equation*}
Therefore, the data-driven discriminant procedure boils down to estimate the approximation coefficient $\tgamma$ and the mean functions $\bmu_1(u)$, $\bmu_{2}(u)$ based on collected samples. In the following subsections, we consider the equal-prior case, that is $\pi_1 = \pi_2 = 1/2$. And we provide the extension of our method to unbalanced case in Section \ref{sec:unequal}.

\subsection{Data-driven discriminant procedure}
Let $\{(\bX_i, U_i, Y_i): i=1,2,...,2n\}$ be an i.i.d. sample set. We denote the pseudo response variable by $Z_i = \dI(Y_i = 1)-\frac{1}{2}$ for $i=1,2,...,2n$ and denote the value of B-spline basis taken at $U_i$ by $\bB_i = (B_{1}(U_i),...,B_{L_n}(U_i))^{\top}$. Without loss of generality, we assume the sample size of two classes are equal. The sample index sets of two classes are $\mathcal{I}_1 = \{i:Y_i = 1\}$ and $\mathcal{I}_2 = \{i:Y_i = 0\}$ with $|\mathcal{I}_1| = |\mathcal{I}_2| = n$.

\subsubsection{Classical low-dimensional regime.}
To construct the sample form of problem (\ref{spline_p1}), we start with estimating the mean functions $\bmu_{1}(u)$ and $\bmu_{2}(u)$. By the B-spline theory, we may estimate the mean functions by
$\hmu_{l}(u) = (\halpha_{l1}^{\top}\bB(u),\cdots,\halpha_{lp}^{\top}\bB(u))^{\top}$ for $l=1,2$, where
\begin{equation*}
    \halpha_{lj} = \LRs{\sum_{i \in \mathcal{I}_l} \bB_i\bB_i^{\top}}^{-1}\sum_{i \in \mathcal{I}_l} \bB_i X_{ij},\quad \text{for }j=1,2,...,p.
\end{equation*}
 Let $\hmu(u) = (\hmu_{1}(u)+\hmu_{2}(u))/2$, the estimators $(\hgamma_{(1)}, \cdots, \hgamma_{(p)})$ can be obtained by solving the following least-square problem
\begin{equation}\label{low_problem}
	\min_{\substack{\bgamma_{(j)} \in \mathbb{R}^{L_n},\\ 1\leq j \leq p}}\frac{1}{2n}\sum_{i=1}^{2n}\LRs{Z_i -\sum_{j=1}^p(X_{ij}-\widehat{\mu}_{j}(U_i))\bB_i^{\top}\bgamma_{(j)}}^2.
\end{equation}
With slightly abusing notations, we denote $\hgamma = (\hgamma_{(1)}^{\top}, \cdots, \hgamma_{(p)}^{\top})^{\top}$ and $\tB_i = (\bX_i - \hmu(U_i))\otimes \bB_i$. In low-dimensional regime, the problem (\ref{low_problem}) has a closed form solution
\begin{equation}\label{spline_s}
	\hgamma = \left(\sum_{i=1}^{2n}\tB_i\tB_i^{\top}\right)^{-1}\sum_{i=1}^{2n}\tB_iZ_i.
\end{equation}

\subsubsection{Sparse high-dimensional regime.} 
In the high-dimensional case, we assume Bayes' discriminant function $\bbeta^{*}(u)$ is sparse with the support set $S:=\{j: \E[|\beta_{j}^{*}(U)|^2]>0\}$ and $|S| = s$. Without loss of generality, let $S = \{1,...,s\}$.

Since $\btheta^*(u)$ has the same support set with $\bbeta^*(u)$, the ``best'' coefficients for approximating $\theta_j^*(u)$ for $j\in S$ are defined as
\begin{equation}\label{oracle_spline}
	(\tgamma_{(1)}, \cdots, \tgamma_{(s)}) = \arg\min_{\substack{\bgamma_{(j)} \in \mathbb{R}^{L_n},\\1\leq j\leq s}}\E\LRm{\LRs{Z-\sum_{j= 1}^s(X_j-\mu_{j}(U))\bgamma_{(j)}^{\top}\bB}^2}.
\end{equation}
Consequently, for any $u\in [0,1]$, we shall approximate the discriminant direction function $\btheta^{*}(u)$ by
\begin{equation*}
	\tbtheta(u) = \left(\tgamma_{1}^{\top} \bB(u),\cdots,\tgamma_{s}^{\top} \bB(u),0,\cdots,0\right)^{\top}.
\end{equation*}
Let $\bD = \E[\tB(U)\tB(U)^{\top}]$ and $\bb = \E[\tB(U) Z]$, the approximation coefficient vector can be equivalently written as $\tgamma = (\tgamma_1^{\top},\cdots,\tgamma_s^{\top}, \boldsymbol{0}^{\top},\cdots,\boldsymbol{0}^{\top})^{\top} = (\tgamma_{(S)}^{\top}, \tgamma_{(S^c)}^{\top})^{\top}$ where $\tgamma_{(S)} = \bD_{(SS)}^{-1}\bb_{(S)}$ and $\tgamma_{(S^c)} = \boldsymbol{0}_{(p-s)L_n}$. It means that the estimator of approximation coefficient $\tgamma$ should have \emph{group sparsity} structure. Therefore, we add the group lasso penalty \citep{yuan2006model} to the objective function in (\ref{low_problem}), and then obtain the estimators $(\hgamma_{(1)},\cdots,\hgamma_{(p)})$ by solving
\begin{equation}\label{lasso_1}
	\min_{\substack{\bgamma_{(j)} \in \mathbb{R}^{L_n},\\ 1\leq j \leq p}}\frac{1}{2n}\sum_{i=1}^{2n}\LRs{Z_i -\sum_{j=1}^p(X_{ij}-\widehat{\mu}_{j}(U_i))\bB(U_i)^{\top}\bgamma_{(j)}}^2 + \lambda_n\sum_{j=1}^p\twonorm{\bgamma_{(j)}}
\end{equation}
where $\lambda_n$ is a tuning parameter. After some simplifications, the problem (\ref{lasso_1}) is equivalent to the following quadratic programming form
\begin{equation}\label{lasso_2}
\hgamma = \arg\min_{\bgamma \in \mathbb{R}^{pL_n}} \LRl{\frac{1}{2}\bgamma^{\top}\bD_n\bgamma - \bb_n^{\top}\bgamma + \lambda_n\sum_{j=1}^p\twonorm{\bgamma_{(j)}}},
\end{equation}
where 
\begin{equation*}
\bD_n = \frac{1}{2n}\sum_{i=1}^{2n}\tB_i\tB_i^{\top},\quad \bb_n =  \frac{1}{2n}\sum_{i=1}^{2n}\tB_iZ_i.
\end{equation*}
The problem (\ref{lasso_2}) can be efficiently solved by several well studied optimization methods, such as group coordinate descent algorithm and iterative shrinkage thresholding algorithm (ISTA) \citep{beck2009fast}. We provide a detailed description about ISTA to solve \eqref{lasso_2} in Appendix \ref{appen:ISTA}.

\subsubsection{Discriminant rule and asymptotic optimality.}
After obtaining $\hmu(u)$ and $\hgamma$, the estimator of discriminant direction function is given by
\begin{equation}\label{dis_spline_sam}
	\hbtheta(u) = \left(\hgamma_{(1)}^{\top} \bB(u),\cdots,\hgamma_{(p)}^{\top} \bB(u)\right)^{\top}.
\end{equation}
For any new observation $(\bX_{\text{new}}, U_{\text{new}})$, the data-driven varying coefficient linear discriminant rule is
\begin{equation}\label{sample_rule}
	\widehat{\psi}(\bX_{\text{new}},U_{\text{new}}) = \dI\left((\bX_{\text{new}}-\hmu(U_{\text{new}}))^{\top} \hbtheta(U_{\text{new}}) \geq 0\right).
\end{equation}

For any $u\in [0,1]$, the optimal misclassification risk of oracle rule (\ref{rule}) is $R(u) = \Phi(-\Delta(u)/2)$, where $\Delta(u) = \sqrt{(\bmu_{1}(u)- \bmu_{2}(u))^{\top}\bSigma^{-1}(u)(\bmu_{1}(u)- \bmu_{2}(u))}$ and $\Phi(\cdot)$ is the cumulative distribution function of a standard normal random variable.
Given the samples and $u\in [0,1]$, the conditional misclassification risk of data-driven rule (\ref{sample_rule}) is
\begin{equation*}
	R_n(u) := \frac{1}{2}\Phi\left(\frac{(\hmu(u) - \bmu_{1}(u))^{\top}\hbtheta(u)}{ \sqrt{\widehat{\boldsymbol{\theta}}^{\top}(u) \boldsymbol{\Sigma}(u) \widehat{\boldsymbol{\theta}}(u)}}\right)+\frac{1}{2}\bar{\Phi}\left(\frac{(\hmu(u) - \bmu_{2}(u))^{\top}\hbtheta(u)}{\sqrt{\widehat{\boldsymbol{\theta}}^{\top}(u) \boldsymbol{\Sigma}(u) \widehat{\boldsymbol{\theta}}(u)}}\right),
\end{equation*}
where $\hmu(u) = \left(\hmu_{1}(u) + \hmu_{2}(u)\right)/2$ and $\bar{\Phi}(\cdot) = 1- \Phi(\cdot)$. 
Through utilizing the technique developed in \cite{tony2019high}, we have the following proposition to provide an upper bound for the excess misclassification risk.

\begin{proposition}\label{pro_mis_rate}
	Suppose that for any $u\in [0,1]$, $\twonorm{\bSigma(u)}$ is uniformly upper bounded from infinity and $\Delta(u)$ is uniformly lower bounded away from zero. In addition, if $\twonorm{\hmu_{1}(u) -\bmu_{1}(u)} = o(1)$, $\twonorm{\hmu_{2}(u) -\bmu_{2}(u)} = o(1)$ and $\twonorm{\hbtheta(u) -\btheta^*(u)} = o(1)$, we have for any $u \in [0,1]$
	\begin{equation}\label{mis_bound}
		|R_n(u) - R(u)|\lesssim \twonorm{\hbtheta(u) - \btheta^{*}(u)}^2+|\left(\hmu(u) - \bmu(u)\right)^{\top}\bbeta^{*}(u)|^2.
	\end{equation}
\end{proposition}


\section{Theoretical results}\label{sec_result}
In this section, we will present the estimation error bounds and the convergence rates of excess misclassification risk of our proposed varying coefficient LDA procedure in both low-dimensional regime and high-dimensional regime. Specially, for two function vectors $\bnu(\cdot) = (\nu_1(\cdot),...,\nu_m(\cdot))^{\top}$ and $\bxi(\cdot) = (\xi_1(\cdot),...,\xi_m(\cdot))$ mapping from $[0,1]$ to $\mathbb{R}^m$, we define the $L_2$ distance between $\bnu(\cdot)$ and $\bxi(\cdot)$ as
\begin{align*}
    \|\bnu - \bxi\|_{L_2} = \LRs{\int_0^1 \twonorm{\bnu(u) - \bxi(u)}^2 du}^{\frac{1}{2}}.
\end{align*}

\subsection{Classical low-dimensional regime}\label{sec_low}
Before presenting the convergence rates of our proposed estimator, we introduce the following necessary technical assumptions for the clarity of ensuing theoretical results.

\begin{enumerate}[($\mathbf{C}1$)]
	\item There exist two constant $0<\lambda_0\leq\lambda_1<\infty$ such that for any $u\in [0,1]$
	\begin{equation*}
		\lambda_0 \leq \lambda_{\min}\left(\bSigma(u)\right) \leq \lambda_{\max}\left(\bSigma(u)\right) \leq \lambda_1,
	\end{equation*}
	where $\lambda_{\min}\left(\bSigma(u)\right)$ and $\lambda_{\max}\left(\bSigma(u)\right)$ are respectively the minimum and maximum eigenvalues of $\bSigma(u)$.
	\item The density function $h$ of $U$ satisfies that $0<D_1\leq h(u)\leq D_2<\infty$ for two positive constants $D_1$ and $D_2$ and any $u\in [0,1]$.
	\item Each entry of functions $\bmu_{1}(u)$, $\bmu_{2}(u)$ and $\bSigma(u)^{-1}(\bmu_{1}(u)-\bmu_{2}(u))$ belongs to the following function space
	\begin{equation*}
	\begin{aligned}
		\mathcal{W}^{d}([0,1]) := \{f: [0,1]\to \mathbb{R},\ &\sup_{x}|f^{(\ell)}(x)|\leq D \text{ for }\ell = 0,1,...,t\text{ and }\\
		&\sup_{x,x^{'}}|f^{(t)}(x)-f^{(t)}(x^{'})|\leq L|x-x^{'}|^{r} \}
	\end{aligned}
	\end{equation*}
	where $d = r+t\geq 1$ and $f^{(s)}$ denotes the $s$-th derivative of function $f$ and $f^{(0)} = f$.
	\item Assume $\sup_{u\in [0,1]}\max\{\twonorm{\bmu_{1}(u)-\bmu_{2}(u)}, \twonorm{\btheta^*(u)}\}= \delta_p \leq M$ for some large constant $M$. In addition, $p = o(n^{(2d-1)/(2d+1)})$.
\end{enumerate}
Assumption \cond{1} is very common in high-dimensional linear discriminant analysis literature \citep{mai2012direct, cai2011direct, jiang2020dynamic}. Assumption \cond{2} and \cond{3} are regular conditions in B-spline approximation theory, similar assumptions also appeared in \citep{JMLR:v13:xue12a,fan2014nonparametric}. For the simplicity of convergence rates, we assume $\twonorm{\bmu_{1}(u)-\bmu_{2}(u)}$ and $\twonorm{\btheta^*(u)}$ are both uniformly bounded in \cond{4}. The condition on the dimensionality ensures that $L_n\sqrt{p \log n/n} = o(1)$ under the optimal length of B-spline basis $L_n \asymp n^{1/(2d+1)}$, which guarantees the optimality of our proposed estimator.

Note that $L_2$ error of our proposed estimator can be decomposed into two parts: the approximation error $\|\tbtheta - \btheta^{*}\|_{L_2}$ and the estimation error $\|\hbtheta - \tbtheta\|_{L_2}$. Our first result shows that the approximation error shrinks as the length of spline basis vector $L_n$ grows, which also attains the optimal convergence rate of classical B-spine approximation error (see \cite{huang2003local,schumaker2007spline}). The proof of Theorem \ref{thm_appro_error} is given in Appendix \ref{proof:thm_appro_error}.
\begin{theorem}\label{thm_appro_error}
	Assume the assumptions \cond{1}-\cond{4} hold, then the approximation error in $L_2$ distance is bounded by
	\begin{equation}\label{low_appro_error}
		\|\tbtheta - \btheta^{*}\|_{L_2} \lesssim \sqrt{p}L_n^{-d}.
	\end{equation}
\end{theorem}


The following theorem provides the upper bound of estimation error for the discriminant direction function estimator (\ref{dis_spline_sam}). Compared with the analysis in the varying coefficient linear model, the theoretical development in this paper is more challenging. The reason is two-fold:
\begin{itemize}
    \item There is no direct relation between the pseudo response variable $Z_i$ and the covariate $\bX_i$. The empirical processes in the proof are established upon fine-grained decomposition to $\bD_n - \bD$ (see Appendix \ref{proof:lemma:tB_matrix_concentration}).
    
    \item The estimator for approximation coefficient $\hgamma$ in (\ref{low_problem}) involves the mean function estimators $\hmu_1$ and $\hmu_2$ computed from the same samples. We utilize chaining technique to establish several concentration inequalities on the operator norm of matrices and $\ell_2$ norm of matrix-vector-products (see Lemma \ref{lemma_matrix_bound_4}-\ref{lemma_matrix_bound_5}).
\end{itemize}
The proof of Theorem \ref{thm_low_dim_ell2} is deferred to Appendix \ref{proof:thm_ld_ell2}.
\begin{theorem}\label{thm_low_dim_ell2}
	Assume conditions \cond{1}-\cond{4} hold. Let $a_n = \sqrt{L_n\log n/n}+L_n^{-d}$, the estimation error in $L_2$ distance is bounded by
	\begin{equation}\label{sup_low_bound}
		\|\hbtheta - \tbtheta\|_{L_2}  =\Op\LRs{\sqrt{\frac{pL_n\log n}{n}} + a_n pL_n \sqrt{\frac{\log n}{n}}+\sqrt{p} L_n^{-d}}.
	\end{equation}
\end{theorem}
\begin{remark}
	Together with the approximation error in Theorem \ref{thm_appro_error} and assumption \cond{4}, if we take the length of B-spline vector as $L_n \asymp \LRs{n/\log n}^{\frac{1}{2(d+1)}}$, it is easy to see that the $L_2$ error can be bounded by
	\begin{equation}\label{eq:low_dim_theta_error}
		\|\hbtheta - \btheta^*\|_{L_2} = \Op\left(\sqrt{p}\LRs{\frac{\log n}{n}}^{\frac{d}{2d+1}}\right).
	\end{equation}
	According to \cite{stone1982optimal}, the minimax convergence rate for one-dimensional function in function space $\mathcal{W}^{d}([0,1])$ is $n^{-d/(2d+1)}$. Apparently, our proposed estimation procedure is optimal up to a logarithmic factor.
\end{remark}
From assumptions \cond{1} and \cond{4}, we know $|\left(\hmu(u) - \bmu(u)\right)^{\top}\bbeta^{*}(u)|^2 \lesssim \twonorm{\hmu(u) - \bmu(u)}^2$. In Proposition \ref{pro_mean_bound}, we establish the uniform bound for the mean function estimator, that is
\begin{align*}
    \sup_{u\in [0,1]}\twonorm{\hmu(u) - \bmu(u)} = \Op\LRs{\sqrt{\frac{p L_n \log n}{n}} + \sqrt{p}L_n^{-d}}.
\end{align*}
In addition, we also have $\sup_{u\in [0,1]}\twonorm{\hbtheta(u) - \tbtheta(u)} = O_{\mathbb{P}}(L_n\sqrt{p\log n/n})$ since $\twonorm{\hgamma - \tgamma}$ shares the same bound with \eqref{sup_low_bound} (see Appendix \ref{proof:thm_ld_ell2}) and $\twonorm{\bB(u)} \leq \sqrt{L_n}$. In conjunction with \eqref{eq:low_dim_theta_error}, we can obatin the $L_2$ bound of the excess misclassification risk in the following corollary.
\begin{corollary}
    Under the same settings of Theorem \ref{thm_low_dim_ell2}, we assume $\Delta(u)\geq c>0$ for some constant $c$ and take $L_n \asymp \LRs{n/\log n}^{\frac{1}{2(d+1)}}$, then it holds
\begin{equation*}
	\|R_n - R\|_{L_2} = \Op\LRs{p\LRs{\frac{\log n}{n}}^{\frac{2d}{2d+1}}}.
\end{equation*}
\end{corollary}

\subsection{Sparse high-dimensional regime}\label{sec_high}
The following assumption plays a similar role as condition \cond{4} in low-dimensional regime.
\begin{enumerate}[($\mathbf{C}1$)]
	\setcounter{enumi}{4}
	\item Assume $\sup_{u\in [0,1]}\max\{\twonorm{(\bmu_{1}(u)-\bmu_{2}(u)_{S}}, \twonorm{\btheta^*(u)}\}= \delta_s \leq M$ for some large constant $M$. In addition, $s = o(n^{(2d-1)/4(d+1)})$.
\end{enumerate} 

The approximation error bound under sparse setting is presented in the following theorem, which can be easily obtained by tracing the proof of Theorem \ref{thm_appro_error} since $\theta^{*}_{j}(\cdot) = \widetilde{\theta}_{j}(\cdot) = 0$ for $j \in S^c$.
\begin{theorem}\label{thm_high_appro_error}
	Assume the assumptions \cond{1}-\cond{3} and \cond{5} hold, then the approximation error in high-dimensional case is
	\begin{equation*}
		\|\tbtheta - \btheta^{*}\|_{L_2} \lesssim \sqrt{s}L_n^{-d}.
	\end{equation*}
\end{theorem}
Below we provide the estimation error bound for the group-sparse estimator in \eqref{lasso_2}, and the proof is deferred to Appendix \ref{proof:thm_hd_ell2}.
\begin{theorem}\label{thm_hd_ell2}
	Assume conditions \cond{1}, \cond{2}, \cond{3} and \cond{5} hold, let $a_n = \sqrt{L_n\log n/n}+L_n^{-d}$, for any $\vartheta>0$, if we take
	\begin{equation}\label{lambda}
		\lambda_n\geq C\left(\sqrt{\frac{L_n\log p}{n}} + a_nL_ns\sqrt{\frac{\log p}{n}} + \sqrt{s} L_n^{-d}\right)
	\end{equation}
	for some sufficiently large positive constant $C$, then
	\begin{equation*}
		\|\hbtheta - \tbtheta\|_{L_2} \lesssim \sqrt{s}\lambda_n
	\end{equation*}
	holds with probability at least $1 - L_np^{-\vartheta} - L_np^{-\vartheta s}$.
\end{theorem}
\begin{remark}
    To interpret the orders in (\ref{lambda}), we introduce the crucial quantity in the proof of Theorem \ref{thm_hd_ell2}: $\max_{1\leq j\leq p}\twonorm{(\bD_{n})_{(jS)}\tgamma_{(S)} - (\bb_{n})_{(j)}}$, which can be bounded by
    \begin{equation}\label{bound_decom}
    	\begin{aligned}
    	\twonorm{(\bD_{n})_{(jS)}\tgamma_{(S)} - (\bb_{n})_{(j)}} &\leq \twonorm{(\bD_{n} - \bD)_{(jS)}\tgamma_{(S)}} + \twonorm{(\bb_{n} - \bb)_{(j)}}\\
    	&+ \twonorm{\bD_{(jS)}\tgamma_{(S)} - \bb_{(j)}}.
    	\end{aligned}
    \end{equation}
    For any $1 \leq j\leq p$, the first two terms in (\ref{bound_decom}) can be bounded by $\sqrt{L_n\log p/n} + a_nL_ns\sqrt{\log p/n}$ through concentration. For $j\in S$, $\bD_{(jS)}\tgamma_{(S)} - \bb_{(j)} = \boldsymbol{0}$ holds due to the definition of $\tgamma$ in \eqref{oracle_spline}. For $j \not\in S$, despite the fact $\bD_{(S^cS)}\tgamma_{(S)} - \bb_{(S^c)}\neq \boldsymbol{0}$, we can still show that it is bounded by $\twonorm{(\btheta^{*}(u) - \tbtheta(u))_{S}}$ (see Appendix \ref{proof:thm_hd_ell2}), which is exactly the last term $\sqrt{s}L_n^{-d}$ in (\ref{lambda}).
\end{remark}
If we set the length of B-spline basis vector as $L_n \asymp \LRs{ns/\log p}^{\frac{1}{2d+1}}$, the $L_2$ error of the group-sparse estimator $\hbtheta(\cdot)$ will be
	\begin{equation}\label{opt_rate_hd}
	\|\hbtheta - \btheta^{*}\|_{L_2} = \Op\left(s^{\frac{d+1}{2d+1}}\LRs{\frac{\log p}{n}}^{\frac{d}{2d+1}}\right).
	\end{equation}
Compared with the oracle minimax rate $\sqrt{s}n^{-\frac{2d}{2d+1}}$, there is an additional factor $s^{1/(2d+1)}$ in \eqref{opt_rate_hd} due to the bias $\bD_{(S^cS)}\tgamma_{(S)} - \bb_{(S^c)}$. To obtain the convergence rate for the excess misclassification risk, it suffices to control the upper bound of $|(\hmu(u) - \bmu(u))^{\top}\bbeta^{*}(u)|^2$. Recall the fact $\bSigma(u)\bbeta^{*}(u) = \bmu_1(u)-\bmu_2(u)$, then simple algebra shows that $\bbeta_{S}^{*}(u) = (\bSigma_{SS}(u))^{-1}(\bmu_{1}(u)-\bmu_{2}(u))_S$. Combining with condition \cond{5}, we have $|(\hmu(u) - \bmu(u))^{\top}\bbeta^{*}(u)|^2\lesssim \twonorm{(\hmu(u) - \bmu(u)_S}^2$. Then the following corollary is a direct result of \eqref{opt_rate_hd} and Proposition \ref{pro_mean_bound}.

\begin{corollary}
    With the same conditions and choice of $\lambda_n$ in Theorem \ref{thm_hd_ell2}, if we take $L_n \asymp \LRs{ns/\log p}^{\frac{1}{2d+1}}$, the excess misclassification risk of $\hbtheta$ satisfies that
    \begin{equation*}
        \|R_n - R\|_{L_2} = \Op\LRs{s^{\frac{2d+2}{2d+1}}\LRs{\frac{\log p}{n}}^{\frac{2d}{2d+1}}}.
    \end{equation*}
\end{corollary}

\section{Extensions}\label{sec_extension}
This section will generalize our approach to more general classification problems in dynamic data.

\subsection{Binary classification with unequal prior}\label{sec:unequal}
For general static binary classification problem, Bayes' discriminant rule is given by
\begin{equation}\label{gen_rule}
	\psi(\bX) = \dI\left((\bX-\bmu)^{\top} \bSigma^{-1}(\bmu_1-\bmu_2) + \log\frac{\pi_1}{\pi_2} \geq 0\right)
\end{equation}
where $\pi_1 = \Prob(Y = 1), \pi_2 = \Prob(Y = 0)$ and $\bmu = \pi_1 \bmu_1 + \pi_2 \bmu_2$. In varying coefficient regime, according to (\ref{relation}), (\ref{gen_rule}) can be generalized to the following form
\begin{equation}\label{gen_vc_rule}
	\psi(\bX, U) = \dI\left((\bX-\bmu(U))^{\top}c^{*}(U)\btheta^{*}(U) + \log\frac{\pi_1}{\pi_2} \geq 0\right),
\end{equation}
where $c^{*}(U) = 1/[\pi_1\pi_2-\pi_1\pi_2(\bmu_1(U)-\bmu_2(U))^{\top}\btheta^{*}(U)]$. The prior probabilities can be estimated by $\widehat{\pi}_1 = \sum_{i=1}^{N}\dI(Y_i=1)/N$ and $\widehat{\pi}_2 = \sum_{i=1}^{N}\dI(Y_i=0)/N$, where $N$ is the total sample size. To estimate $\btheta^{*}(u)$, we only need to set $Z_i = \widehat{\pi}_2$ if $Y_i = 1$ and $Z_i = -\widehat{\pi}_1$ if $Y_i = 0$ in (\ref{lasso_1}). As a consequence, for any new observation $(\bX_{\text{new}},U_{\text{new}})$, we can perform varying coefficient discriminant rule by plugging in corresponding estimators into (\ref{gen_vc_rule}).

\subsection{Multivariate Exposure Variable}
For multivariate $\bU \in \mathbb{R}^m$, we may consider the following single-index extension. Specially, given $\bU = \boldsymbol{u}$, we assume the covariate $\bX \sim \mathcal{N}(\bmu_1(\boldsymbol{u}^{\top}\boldsymbol{\varphi}^{*}), \bSigma(\boldsymbol{u}^{\top}\boldsymbol{\varphi}^{*}))$ if $Y = 1$ and $\bX \sim \mathcal{N}(\bmu_2(\boldsymbol{u}^{\top}\boldsymbol{\varphi}^{*}), \bSigma(\boldsymbol{u}^{\top}\boldsymbol{\varphi}^{*}))$ if $Y = 0$. Then Bayes's discriminant direction is also a function of $\boldsymbol{u}^{\top}\boldsymbol{\varphi}^{*}$, that is 
$\theta_j^{*}(\boldsymbol{u}) = g_j^{*}(\boldsymbol{u}^{\top}\boldsymbol{\varphi}^{*})$ for $j=1,...,p$, where $g_j^{*}(\cdot)$ is a smooth univariate function. 
Similar to \eqref{pop}, $g_j^{*}$s are defined as the solution of the following least-square problem
\begin{equation*}
    \min_{g_j \in \mathcal{L}^2(\mathcal{P}), j=1,...,p}\E\LRm{\LRs{Z - \sum_{j=1}^p g_j(\mathbf{U}^{\top}\boldsymbol{\varphi}^{*})(X_j - \mu_j(\mathbf{U}^{\top}\boldsymbol{\varphi}^{*}))}^2 \Big| \mathbf{U}}.
\end{equation*}
If $\boldsymbol{\varphi}^*$ is known, we can approximate the function $g_j^{*}(\mathbf{U}^{\top}\boldsymbol{\varphi}^{*})$ by $\bgamma_j^{\top}\bB(\mathbf{U}^{\top}\boldsymbol{\varphi}^{*})$. And the optimal approximation coefficients are defined as
\begin{equation*}
    \min_{\bgamma_j \in \mathbb{R}^{L_n}, j=1,...,p}\E\LRm{\LRs{Z - \sum_{j=1}^p (X_j - \mu_j(\mathbf{U}^{\top}\boldsymbol{\varphi}^{*})) \bgamma_j^{\top}\bB(\mathbf{U}^{\top}\boldsymbol{\varphi}^{*})}^2}.
\end{equation*}

As for the initial estimator of $\boldsymbol{\varphi}$, according to our assumption, we can equivalently write the covariate $\bX_i$ as the form of the standard single index model
\begin{align*}
    \bX_{i} &= \bmu_{1}(\bU_i^\top \boldsymbol{\varphi}^{*}) + \left(\bSigma(\bU_i^\top \boldsymbol{\varphi}^{*})\right)^{1/2} \boldsymbol{\epsilon}_i\quad \text{if} \quad Y_i = 1,\\
    \bX_{i} &= \bmu_{2}(\bU_i^\top \boldsymbol{\varphi}^{*}) + \left(\bSigma(\bU_i^\top \boldsymbol{\varphi}^{*})\right)^{1/2} \boldsymbol{\epsilon}_i\quad \text{if} \quad Y_i = 0,
\end{align*}
where $\boldsymbol{\epsilon}_i \sim \mathcal{N}(\boldsymbol{0}, \bI_p)$.
We may utilize the method proposed in \cite{Xia2006} to obtain the estimator of $\boldsymbol{\varphi}^{*}$, denoted by $\widehat{\boldsymbol{\varphi}}$. By plugging in $\widehat{\boldsymbol{\varphi}}$, the estimators of univariate functions $g_j^{*}$s can be estimated by the B-spline procedure in our paper.

\section{Numerical experiments}\label{sec_simu}
This section investigates the numerical performance of the proposed varying coefficient discriminant procedure. In our simulation study, we only consider the balanced case where the sample sizes of the two classes are equal.

The exposure variable $U_i$ for $i = 1,2,...,2n$ are generated independently from uniform distribution on $[0,1]$ in the following experiments. After generating $U_i$, we sample the covariate $\bX_i$ with $Y_i = 1$ from $\mathcal{N}(\bmu_1(U_i), \bSigma(U_i))$ for $i = 1,...,n$ and sample $\bX_i$ with $Y_i = 0$ from $\mathcal{N}(\bmu_2(U_i), \bSigma(U_i)) $ for $i = n+1,...,2n$, where $\bmu_1(u) = \boldsymbol{0}$ and $\bmu_2(u) = \bSigma(u) \bbeta(u)$. Several combinations of $\bbeta(u)$ and $\bSigma(u)$ are considered in our simulation. Each entry of Bayes' discriminant direction function take values as:
\begin{itemize}
	\item Direction 1: $\beta_{j}^{(1)}(u) = 1$ for $1\leq j \leq p$ (or $s$);
	\item Direction 2: $\beta_{j}^{(2)}(u) = u$ for $1\leq j \leq p$ (or $s$);
	\item Direction 3: $\beta_{j}^{(2)}(u) = \sin(4u)$ for $1\leq j \leq p$ (or $s$);
	\item Direction 4: $\beta_{j}^{(4)}(u) = e^{u}$ for $1\leq j \leq p$ (or $s$).
\end{itemize}
In high-dimensional case, we set $\beta_{j}(\cdot) = 0$ for $s+1\leq j \leq p$. Three covariance matrices are considered in our simulations:
\begin{itemize}
    \item Covariance matrix 1, $\sigma_{i,j}^{(1)}(u) = 0.5^{\abs{i-j}},\text{ for } 1\leq i,j\leq p$;
    \item Covariance matrix 2. $\sigma_{i,j}^{(2)}(u) = u^{\abs{i-j}}, \text{ for } 1\leq i,j\leq p$;
    \item Covariance matrix 3. $\sigma_{i,j}^{(3)}(u) = u\dI(i\neq j) + \dI(i = j) \text{ for } 1\leq i,j\leq p$.
\end{itemize}
The combination of \textit{Direction 1} and \textit{Covariance matrix 1} is a classical static setting, where each entry of the mean vector and covariance matrix is a constant value. The other combinations are dynamic settings. We use the cubic spline in our simulation, and select the number of spline basis functions by 5-fold cross-validation. We compute the misclassification risk based on an independently generated test set with size 200.

\begin{table}[tb]
\centering
\caption{Misclassification risk and its standard error (in parentheses) of each method in low-dimensional case.}\label{Tab0}
\begin{adjustbox}{width=0.8\textwidth}
\begin{tabular}{@{}cccccccc@{}}
\toprule
$p$  & $\bSigma$ & Orcale & VCLDA        & LDA              & Orcale & VCLDA        & LDA                  \\ \midrule
    &       & \multicolumn{3}{c}{$\bbeta^{(1)}$} & \multicolumn{3}{c}{$\bbeta^{(2)}$}  \\ \cmidrule(lr){3-5} \cmidrule(lr){6-8}
\multirow{3}{*}{5}  & 1     & 0.048  & 0.075(0.021) & 0.050(0.016)   & 0.227  & 0.259(0.039) & 0.255(0.032) \\
                    & 2     & 0.055  & 0.078(0.021) & 0.119(0.028) & 0.221  & 0.249(0.034) & 0.272(0.032) \\
                    & 3     & 0.039  & 0.058(0.020)  & 0.093(0.023)  & 0.202  & 0.221(0.032) & 0.245(0.030) \\ \cmidrule(lr){3-5} \cmidrule(lr){6-8}
\multirow{3}{*}{10} & 1     & 0.005  & 0.028(0.014) & 0.006(0.006) & 0.155  & 0.193(0.033) & 0.197(0.029) \\
                    & 2     & 0.014  & 0.038(0.015) & 0.152(0.025) & 0.163  & 0.198(0.029) & 0.270(0.035)   \\
                    & 3     & 0.004  & 0.027(0.012) & 0.104(0.023)  & 0.126  & 0.155(0.031) & 0.212(0.026)  \\ \cmidrule(lr){3-5} \cmidrule(lr){6-8}
\multirow{3}{*}{20} & 1     & 0.000      & 0.020(0.012)  & 0.000(0.001)  & 0.108  & 0.157(0.032) & 0.166(0.028)  \\
                    & 2     & 0.002  & 0.041(0.019) & 0.215(0.034)  & 0.125  & 0.182(0.034) & 0.312(0.035) \\ 
                    & 3     & 0.000      & 0.042(0.017) & 0.117(0.025)  & 0.081  & 0.128(0.024) & 0.222(0.029) \\ \midrule
 &       & \multicolumn{3}{c}{$\bbeta^{(3)}$} & \multicolumn{3}{c}{$\bbeta^{(4)}$} \\ \cmidrule(lr){3-5} \cmidrule(lr){6-8}
\multirow{3}{*}{5}  & 1     & 0.194  & 0.234(0.034) & 0.317(0.034)  & 0.010   & 0.024(0.012) & 0.029(0.013) \\
                    & 2     & 0.192  & 0.234(0.033) & 0.382(0.047) & 0.025  & 0.040(0.015)  & 0.156(0.025) \\
                    & 3     & 0.173  & 0.209(0.028) & 0.351(0.047) & 0.018  & 0.033(0.014) & 0.137(0.023) \\ \cmidrule(lr){3-5} \cmidrule(lr){6-8}
\multirow{3}{*}{10} & 1     & 0.125  & 0.186(0.027) & 0.281(0.038) & 0.001  & 0.011(0.008) & 0.016(0.009)\\
                    & 2     & 0.117  & 0.183(0.030)  & 0.437(0.046) & 0.007  & 0.027(0.014) & 0.195(0.029)\\
                    & 3     & 0.092  & 0.154(0.031) & 0.353(0.051) & 0.003  & 0.022(0.013) & 0.157(0.023) \\ \cmidrule(lr){3-5} \cmidrule(lr){6-8}
\multirow{3}{*}{20} & 1     & 0.083  & 0.189(0.031) & 0.286(0.039) & 0.000      & 0.014(0.009) & 0.011(0.008) \\
                    & 2     & 0.077  & 0.200(0.037)   & 0.476(0.041) & 0.001  & 0.041(0.023) & 0.246(0.035) \\
                    & 3     & 0.054  & 0.176(0.037) & 0.391(0.059) & 0.000      & 0.044(0.022) & 0.172(0.029) \\ \bottomrule 
\end{tabular}
\end{adjustbox}
\end{table}

\subsection{Low-dimensional case}
 For low-dimensional case, the sample size of each class is fixed as $n = 100$ and the dimensionality $p$ is varying from $\{5, 10, 20\}$. The proposed method in this paper (abbreviated as VCLDA) is deployed to the generated data. For comparison, we also conduct the following two classification rules:
\begin{enumerate}
	\item Oracle: use the population Bayes' discriminant direction $\bSigma(u)^{-1}(\bmu_1(u) - \bmu_2(u))$ to conduct classification.
	\item LDA: use the \emph{static} estimators of mean vectors and covariance matrix, i.e., the sample means and sample covariance matrix, to compute discriminant direction.
\end{enumerate}

We report the averaged misclassification risks computed from the test set in Table \ref{Tab0}. The oracle classification rule is the most accurate among all procedures. In a static setting, we can see that LDA achieves nearly oracle performance. As we expected, the performance of LDA procedure degrades drastically in the dynamic case. Meanwhile, the misclassification risk of VCLDA is significantly lower than LDA, and very close to the oracle procedure in all dynamic settings.

\subsection{High-dimensional case}
In high-dimensional simulation, we fix the sample size of each class as $n = 100$ and consider the dimensionality $p = 100$ and $p = 200$. Moreover, the sparsity under each dimensionality varies in $\{5, 10, 20\}$. For comparison, we also conduct the oracle rule, the static LPD rule \citep{cai2011direct} and DLPD rule \citep{jiang2020dynamic} in the test set. The misclassification risks and their standard errors under four discriminant direction functions are summarized in Table \ref{Tab1}-\ref{Tab4} respectively. Undoubtedly, the oracle classification rule is the most accurate among all procedures. In a static setting, it can be seen that the DLPD rule almost achieves the same performance as the LPD rule in static settings (see Table \ref{Tab1}). As we expected, the performance of the classical LDA procedure degrades drastically in the dynamic case, which performs like random guessing in a highly dynamic setting. Except for the static setting, we can see that the misclassification risk of our proposed VCLDA rule is significantly lower than the DLPD rule, especially for the setting with \textit{Covariance matrix 2}. In addition, the results indicate that the performance of VCLDA is most close to the oracle procedure.

In fact, VCLDA fully uses the information that the discrimination direction varies with different values of $U$ while the active set of the discrimination coefficient will not change in our simulation settings. The former leads to a lower misclassification risk than the static LPD rule, and the latter leads to better performance over the DLPD rule.

\begin{table}[tb]
\centering
\caption{Misclassification risk and its standard error (in parenthesis) of each method under Direction 1 in high-dimensional case.}\label{Tab1}
\begin{adjustbox}{width=0.9\textwidth}
\begin{tabular}{@{}cccccccccc@{}}
\toprule
$s$   & $\bSigma$ & Orcale & VCLDA        & LPD          & DLPD         & Orcale & VCLDA        & LPD          & DLPD         \\ \midrule
 &       & \multicolumn{4}{c}{$p = 100$} & \multicolumn{4}{c}{$p = 200$}              \\ \cmidrule(lr){3-6} \cmidrule(lr){7-10}
\multirow{3}{*}{5}  & 1     & 0.048  & 0.076(0.019) & 0.053(0.012) & 0.053(0.015) & 0.048  & 0.071(0.018) & 0.057(0.016) & 0.057(0.016) \\
                    & 2     & 0.055  & 0.070(0.017)  & 0.195(0.154) & 0.202(0.035) & 0.056  & 0.067(0.017) & 0.153(0.106) & 0.134(0.040) \\
                    & 3     & 0.039  & 0.060(0.017)  & 0.332(0.192) & 0.085(0.018) & 0.039  & 0.060(0.017) & 0.150(0.058) & 0.100(0.029) \\ \cmidrule(lr){3-6} \cmidrule(lr){7-10}
\multirow{3}{*}{10} & 1     & 0.005  & 0.015(0.009) & 0.103(0.193) & 0.101(0.187) & 0.005  & 0.025(0.012) & 0.007(0.006) & 0.007(0.006) \\
                    & 2     & 0.014  & 0.041(0.017) & 0.152(0.028) & 0.168(0.034) & 0.014  & 0.043(0.019) & 0.148(0.023) & 0.160(0.030) \\
                    & 3     & 0.004  & 0.012(0.009) & 0.123(0.026) & 0.040(0.019)  & 0.004  & 0.016(0.011) & 0.128(0.021) & 0.040(0.015) \\ \cmidrule(lr){3-6} \cmidrule(lr){7-10}
\multirow{3}{*}{20} & 1     & 0.000      & 0.009(0.006) & 0.055(0.157) & 0.055(0.157) & 0.000      & 0.004(0.005) & 0.000(0.001) & 0.000(0.001) \\
                    & 2     & 0.002  & 0.018(0.010)  & 0.208(0.067) & 0.176(0.030)  & 0.002  & 0.014(0.009) & 0.198(0.031) & 0.164(0.026) \\
                    & 3     & 0.000      & 0.009(0.007) & 0.123(0.022) & 0.026(0.013) & 0.000      & 0.009(0.008) & 0.125(0.021) & 0.029(0.018) \\ \bottomrule 
\end{tabular}
\end{adjustbox}
\end{table}

\begin{table}[tb]
\centering
\caption{Misclassification risk and its standard error (in parentheses) of each method under Direction 2 in high-dimensional case.}\label{Tab2}
\begin{adjustbox}{width=0.9\textwidth}
\begin{tabular}{@{}cccccccccc@{}}
\toprule
$s$     & $\bSigma$ & \multicolumn{1}{c}{Orcale} & \multicolumn{1}{c}{VCLDA} & \multicolumn{1}{c}{LPD} & \multicolumn{1}{c}{DLPD} & \multicolumn{1}{c}{Orcale} & \multicolumn{1}{c}{VCLDA} & \multicolumn{1}{c}{LPD} & \multicolumn{1}{c}{DLPD} \\ \midrule
&    & \multicolumn{4}{c}{$p = 100$}                                                                        & \multicolumn{4}{c}{$p = 200$}                                                                       \\ \cmidrule(lr){3-6} \cmidrule(lr){7-10}
\multirow{3}{*}{5}  & 1     & 0.225                      & 0.243(0.031)              & 0.371(0.109)            & 0.248(0.031)             & 0.227                      & 0.244(0.032)              & 0.281(0.045)            & 0.300(0.042)             \\
                    & 2     & 0.217                      & 0.252(0.028)              & 0.274(0.032)            & 0.338(0.046)             & 0.220                       & 0.237(0.031)              & 0.280(0.042)            & 0.370(0.043)             \\
                    & 3     & 0.199                      & 0.241(0.031)              & 0.268(0.030)             & 0.251(0.027)             & 0.204                      & 0.232(0.031)              & 0.272(0.054)            & 0.239(0.032)             \\ \cmidrule(lr){3-6} \cmidrule(lr){7-10}
\multirow{3}{*}{10} & 1     & 0.158                      & 0.173(0.027)              & 0.204(0.029)            & 0.210(0.030)               & 0.160                       & 0.189(0.028)              & 0.206(0.031)            & 0.210(0.031)             \\
                    & 2     & 0.164                      & 0.208(0.031)              & 0.260(0.050)              & 0.347(0.033)             & 0.165                      & 0.185(0.026)              & 0.258(0.026)            & 0.345(0.034)             \\
                    & 3     & 0.126                      & 0.156(0.026)              & 0.212(0.025)            & 0.166(0.026)             & 0.127                      & 0.140(0.024)              & 0.212(0.030)            & 0.174(0.028)             \\ \cmidrule(lr){3-6} \cmidrule(lr){7-10}
\multirow{3}{*}{20} & 1     & 0.107                      & 0.126(0.024)              & 0.182(0.032)            & 0.146(0.030)              & 0.108                      & 0.132(0.025)              & 0.166(0.025)            & 0.146(0.027)             \\
                    & 2     & 0.126                      & 0.184(0.028)              & 0.280(0.064)             & 0.336(0.028)             & 0.124                      & 0.179(0.028)              & 0.270(0.027)            & 0.329(0.029)             \\
                    & 3     & 0.081                      & 0.099(0.021)              & 0.207(0.056)            & 0.116(0.020)              & 0.081                      & 0.093(0.020)              & 0.202(0.026)            & 0.124(0.023)             \\ \bottomrule
\end{tabular}
\end{adjustbox}
\end{table}

\begin{table}[tb]
\centering
\caption{Misclassification risk and its standard error (in parenthesis) of each method under Direction 3 in high-dimensional case.}\label{Tab3}
\begin{adjustbox}{width=0.9\textwidth}
\begin{tabular}{@{}cccccccccc@{}}
\toprule
$s$                 & $\bSigma$ & Orcale & VCLDA        & LPD          & DLPD         & Orcale & VCLDA        & LPD          & DLPD         \\ \midrule
                    &           & \multicolumn{4}{c}{$p = 100$}                       & \multicolumn{4}{c}{$p = 200$}                       \\ \cmidrule(lr){3-6} \cmidrule(lr){7-10}
\multirow{3}{*}{5}  & 1         & 0.193  & 0.244(0.033) & 0.337(0.055) & 0.291(0.04)  & 0.194  & 0.221(0.030) & 0.335(0.053) & 0.286(0.031) \\
                    & 2         & 0.193  & 0.211(0.029) & 0.412(0.068) & 0.302(0.032) & 0.192  & 0.214(0.028) & 0.395(0.068) & 0.298(0.031) \\
                    & 3         & 0.172  & 0.203(0.033) & 0.334(0.048) & 0.298(0.029) & 0.172  & 0.209(0.028) & 0.341(0.051) & 0.291(0.032) \\ \cmidrule(lr){3-6} \cmidrule(lr){7-10}
\multirow{3}{*}{10} & 1         & 0.123  & 0.147(0.024) & 0.280(0.038)  & 0.232(0.028) & 0.125  & 0.152(0.025) & 0.278(0.039) & 0.237(0.027) \\
                    & 2         & 0.119  & 0.148(0.026) & 0.450(0.075)  & 0.293(0.05)  & 0.122  & 0.145(0.028) & 0.442(0.075) & 0.270(0.032) \\
                    & 3         & 0.090   & 0.115(0.025) & 0.282(0.046) & 0.214(0.03)  & 0.091  & 0.126(0.025) & 0.274(0.036) & 0.226(0.035) \\ \cmidrule(lr){3-6} \cmidrule(lr){7-10}
\multirow{3}{*}{20} & 1         & 0.080   & 0.110(0.022)  & 0.249(0.036) & 0.201(0.027) & 0.084  & 0.109(0.021) & 0.260(0.042) & 0.201(0.031) \\
                    & 2         & 0.080   & 0.114(0.023) & 0.487(0.045) & 0.234(0.027) & 0.079  & 0.114(0.023) & 0.498(0.020) & 0.229(0.027) \\
                    & 3         & 0.050   & 0.084(0.023) & 0.218(0.037) & 0.184(0.026) & 0.052  & 0.071(0.019) & 0.218(0.031) & 0.221(0.033) \\ \bottomrule
\end{tabular}
\end{adjustbox}
\end{table}

\begin{table}[tb]
\centering
\caption{Misclassification risk and its standard error (in parentheses) of each method under Direction 4 in high-dimensional case.}\label{Tab4}
\begin{adjustbox}{width=0.9\textwidth}
\begin{tabular}{@{}cccccccccc@{}}
\toprule
$s$                 & $\bSigma$ & Orcale & VCLDA        & LPD          & DLPD         & Orcale & VCLDA        & LPD          & DLPD         \\ \midrule
                    &           & \multicolumn{4}{c}{$p = 100$}                       & \multicolumn{4}{c}{$p = 200$}                       \\ \cmidrule(lr){3-6} \cmidrule(lr){7-10}
\multirow{3}{*}{5}  & 1         & 0.010   & 0.018(0.009) & 0.032(0.014) & 0.015(0.010)  & 0.010   & 0.022(0.012) & 0.033(0.012) & 0.019(0.010) \\
                    & 2         & 0.025  & 0.049(0.016) & 0.174(0.087) & 0.207(0.035) & 0.025  & 0.035(0.013) & 0.271(0.163) & 0.177(0.029) \\
                    & 3         & 0.018  & 0.037(0.014) & 0.299(0.166) & 0.056(0.021) & 0.018  & 0.031(0.014) & 0.175(0.021) & 0.060(0.020) \\ \cmidrule(lr){3-6} \cmidrule(lr){7-10}
\multirow{3}{*}{10} & 1         & 0.001  & 0.010(0.007)  & 0.020(0.011)  & 0.005(0.009) & 0.001  & 0.006(0.006) & 0.020(0.011) & 0.003(0.004) \\
                    & 2         & 0.007  & 0.025(0.011) & 0.408(0.141) & 0.195(0.031) & 0.006  & 0.020(0.011) & 0.319(0.144) & 0.174(0.028) \\
                    & 3         & 0.003  & 0.022(0.012) & 0.179(0.041) & 0.033(0.014) & 0.003  & 0.016(0.010) & 0.186(0.060) & 0.037(0.014) \\ \cmidrule(lr){3-6} \cmidrule(lr){7-10}
\multirow{3}{*}{20} & 1         & 0.000     & 0.007(0.006) & 0.012(0.008) & 0.000(0.001)     & 0.000      & 0.007(0.007) & 0.011(0.008) & 0.000(0.003) \\
                    & 2         & 0.001  & 0.011(0.009) & 0.426(0.116) & 0.189(0.029) & 0.001  & 0.015(0.009) & 0.479(0.063) & 0.179(0.027) \\
                    & 3         & 0.000     & 0.008(0.005) & 0.170(0.026)  & 0.027(0.016) & 0.000      & 0.009(0.008) & 0.178(0.022) & 0.027(0.015) \\ \bottomrule
\end{tabular}
\end{adjustbox}
\end{table}

\section{Real Data Analysis}\label{sec_realdata}

Diffuse large B-cell lymphoma (DLBCL) is a heterogeneous disease with recognized variability in clinical outcome, genetic features, and cells of origin. It is of vital importance for precision medicine if we can predict DLBCL in advance. Using the data provided in \cite{Monti2005Molecular}, we establish the model to predict DLBCL according to the gene expression. It is mentioned in \cite{Monti2005Molecular} that tumors had less frequent genetic abnormalities in younger patients. Thus, our proposed method VCLDA seems suitable for setting up the prediction model by setting \emph{the age} as the exposure variable $U$.

The original data has 124 patients and 44972 gene expression levels. The binary response means whether a germinal center B-cell is normal or not, which is the significant signal of DLBCL. We screen out 150 gene expression levels to build a model according to the $t$ test on the binary response. We conduct the following four procedures: LPD (exclude age as a covariate), LPD (include age as a covariate), DLPD (regard age as $U$), and VCLDA (regard age as $U$). We randomly choose ten patients as the test sample in each trial and regard the remaining samples as the training set to run the classification procedure. The average results of misclassification risks on the test sample over 100 trials are reported in Table \ref{Tab_real}. It shows that the contribution of $U$ is negligible as a covariate in the static LPD rule. In contrast, it improves the classification accuracy tremendously as an exposure variable in the dynamic model.

\begin{table}[tb]
\centering
\caption{The average misclassification risk and its standard error of each method in DLBCL dataset.}\label{Tab_real}
\begin{adjustbox}{width=1\textwidth}
\begin{tabular}{@{}ccccc@{}}
\toprule
Method & LPD (exclude age)   & LPD (include age) & DLPD ($U =$ age)  & VCLDA ($U =$ age) \\ \midrule
Avg    & 0.432 & 0.432      & 0.192 & 0.171 \\
SE     & 0.211 & 0.211      & 0.167 & 0.122\\ \bottomrule
\end{tabular}
\end{adjustbox}
\end{table}

Additionally, the active sets selected by the DLPD method under different ages are highly coincident. It means that the genes influencing DLBCL will not change significantly with age, which is also reasonable in the gene analysis. Two genes are excluded from the active set by the DLPD method during a very short age interval, which may be confusing and misleading to the relative researchers. Nearly all active genes selected by the DLPD method are also selected by the VCLDA method. Besides, as we find in coefficients estimated by VCLDA, most of the genes have a weak influence on DLBCL when U is small, which collaborates with the conclusion in \cite{Monti2005Molecular} that tumors have less frequent genetic abnormalities in younger patients.

\section{Discussion}
This paper investigates the LDA model for dynamic data and proposes a new varying coefficient discriminant rule. The proposed classification procedure is more efficient than the dynamic linear programming rule \citep{jiang2020dynamic}. We also establish the upper bounds for estimation error and uniform excess misclassification risk. The synthetic and real data experiments also demonstrate a better classification performance of our varying coefficient LDA method.

The Gaussian graphical model (GGM) is an essential formalism to infer dependence structures of contemporary data sets, whose structure is equivalent to the support of the precision matrix. Recently, \cite{qiao2020} proposed the functional graphical model and assumed the covariate is a $p$-dimensional functional data. The authors proposed an estimator of the precision matrix function based on kernel smoothing and CLIME \citep{cai2011constrained}. Therefore, studying the high-dimensional, varying coefficient GGM under a dynamic setting will be of great interest.

\bibliography{docref}

\begin{appendices}
\section{Preliminaries}\label{appen:pre}
\subsection{Background of B-spline approximation}\label{appen:Bspline}
	From now on, we will omit the argument in random vector $\bB(U)$ and $\tB(U)$ and write $\bB$ and $\tB$ respectively whenever the context is clear. We introduce the following facts about standard B-spline basis $\bB^{*}(u) = (B_1^{*}(u),...,B_{L_n}^{*}(u))^{\top}$ (see \citep{de1978practical,fan2014nonparametric}), which will be used in our proof:
	\begin{enumerate}
		\item For any $u\in [0,1]$, both $0\leq \max_{1\leq k\leq L_n}B_k^{*}(u)\leq 1$ and $\sum_{k=1}^{L_n}B_k^{*}(u) = 1$ hold.
		\item For any $\eta_k\in \mathbb{R}$, $k=1,2,...,L_n$, we have
			\begin{equation}\label{int_bound}
				 L_n^{-1}\sum_{k=1}^{L_{n}} \eta_{k}^{2} \lesssim \int\left(\sum_{k=1}^{L_{n}} \eta_{k} B_{k}^{*}(w)\right)^{2} d w \lesssim L_n^{-1} \sum_{k=1}^{L_{n}} \eta_{k}^{2}.
			\end{equation}
	\end{enumerate}
	From the facts displayed above, for any $r\geq 1$, we also have
		\begin{equation*}
			\E\LRm{|B_k^{*}(U)|^r}\asymp L_n^{-1},
		\end{equation*}
		and
		\begin{equation*}
			\twonorm{\E[\bB^{*}]} = \sup_{\twonorm{\bnu} = 1}|\E[\bnu^{\top}\bB^{*}]|\leq \sup_{\twonorm{\bnu} = 1}\LRs{\E[\bnu^{\top}\bB^{*}]^2}^{1/2} = O(L_n^{-1/2}).
		\end{equation*}
	Similarly, we can also obtain that
		\begin{equation*}
			L_n^{-1} \lesssim \lambda_{\min}(\E[\bB^{*}\bB^{*\top}])\leq \lambda_{\max}(\E[\bB^{*}\bB^{*\top}]) \lesssim L_n^{-1}.
		\end{equation*}
	Writing the scaled B-spline basis as $\bB(u) = \sqrt{L_n}\bB^{*}(u)$, we have
	\begin{equation*}
			\twonorm{\E[\bB]}  = O(1),
		\end{equation*}
	and
	\begin{equation*}
			\lambda_{\min}(\E[\bB\bB^{\top}]) = O(1),\quad \lambda_{\max}(\E[\bB\bB^{\top}])  = O(1).
	\end{equation*}
	
	\subsection{Concentration inequality}
	The following concentration inequality will be used throughout the proof.
	\begin{lemma}[Lemma 1, \cite{cai2011direct}]\label{bern_ineq}
	Let $\xi_1,...,\xi_n$ be independent random variables with mean 0. Suppose that there exists some $\phi>0$ and $s_n$ such that $\sum_{i=1}^n\E[\xi_i^2e^{\phi|\xi_i|}]\leq s_n^2$. Then for $0<x<s_n^2$,
	\begin{equation*}
		\Prob\left(\sum_{i=1}^n\xi_i\geq C_{\phi}s_n x\right)\leq \exp(-x^2)
	\end{equation*}
	where $C_{\phi} = \phi + \phi^{-1}$.
\end{lemma}
Next lemma gives the moment inequalities for normal random variable.
\begin{lemma}\label{lemma:normal_moment_bound}
    Let $X \sim \mathcal{N}(0, \sigma^2)$, then for any $0\leq \phi \leq \frac{1}{\sqrt{2}\sigma}$,
    \begin{align*}
        \E[X^2 e^{\phi |X|}] \leq \frac{e}{\phi^2} \frac{1}{\sqrt{1 - 2\phi^2 \sigma^2}};
    \end{align*}
    and for any $\phi \geq 0$ and $k\geq 1$
    \begin{align*}
        \E[X^k e^{\phi |X|}] \leq e^{\frac{\phi^2\sigma^2}{2}}\LRs{\E[X_{-\phi\sigma^2}^k] + \E[X_{\phi\sigma^2}^k]}\quad \text{holds for any } \phi \geq 0,
    \end{align*}
    where $X_{-\phi \sigma^2} \sim \mathcal{N}(-\phi \sigma^2, \sigma^2)$ and $X_{\phi \sigma^2} \sim \mathcal{N}(\phi \sigma^2, \sigma^2)$.
\end{lemma}
\begin{proof}[Proof of Lemma \ref{lemma:normal_moment_bound}]
    Using the basic inequality $s^2e^s \leq e^{2s}$ for any $s \geq 0$, we have
    \begin{align*}
        \E[X^2 e^{\phi |X|}] &\leq \phi^{-2} \E[e^{2\phi |X|}]\\
        &\leq \phi^{-2} \E\LRm{e^{1 + \phi^2 X^2}}\\
        &= \frac{e}{\phi^2} \frac{1}{\sqrt{2\pi}\sigma}\int_{-\infty}^{+\infty}e^{-\frac{x^2}{2\sigma^2} + \phi^2 x^2} dx\\
        & = \frac{e}{\phi^2} \frac{1}{\sqrt{1 - 2\phi^2 \sigma^2}}.
    \end{align*}
    In addition, we also have
    \begin{align*}
        \E[X^k e^{\phi |X|}] &= \frac{1}{\sqrt{2\pi}\sigma}\int_{-\infty}^{+\infty}x^k e^{-\frac{x^2}{2\sigma^2} + \phi |x|} dx\\
        &=\frac{1}{\sqrt{2\pi}\sigma} \LRs{\int_{-\infty}^{0}x^k e^{-\frac{x^2}{2\sigma^2} - \phi x} dx + \int_{0}^{+\infty}x^k e^{-\frac{x^2}{2\sigma^2} + \phi x} dx}\\
        &\leq \frac{e^{\frac{\phi^2\sigma^2}{2}}}{\sqrt{2\pi}\sigma} \LRs{\int_{-\infty}^{+\infty}x^k e^{-\frac{(x + \phi \sigma^2)^2}{2\sigma^2}} dx + \int_{-\infty}^{+\infty}x^k e^{-\frac{(x - \phi\sigma^2)^2}{2\sigma^2}} dx}.
    \end{align*}
\end{proof}

\subsection{Estimation error bound for mean functions}
The estimation error bounds for mean function vectors in the following proposition contribute to establish the convergence rates of discriminant direction estimator and the excess misclassification risk.
\begin{proposition}\label{pro_mean_bound}
	Denote the estimator of the mean functions by $\hmu_{1}(u) = (\halpha_{11}^{\top} \bB(u),\cdots, \halpha_{1p}^{\top} \bB(u))^{\top}$ and $\hmu_{2}(u) = (\halpha_{21}^{\top} \bB(u),\cdots, \halpha_{2p}^{\top} \bB(u))^{\top}$, then under condition \cond{1}-\cond{4}, for any $\vartheta > 0$ we have
	\begin{equation*}
		\sup_{u\in [0,1]}\infnorm{\hmu_{1}(u)-\bmu_{1}(u)}\lesssim \sqrt{\frac{L_n\log n}{n}}+L_n^{-d},
	\end{equation*}
	and
	\begin{equation*}
		\sup_{u\in [0,1]}\infnorm{\hmu_{2}(u)-\bmu_{2}(u)}\lesssim \sqrt{\frac{L_n\log n}{n}}+L_n^{-d},
	\end{equation*}
	hold with probability at least $1-3pL_nn^{-\vartheta}$ respectively.
\end{proposition}

\begin{lemma}\label{lemma_Bnorm_bound}
	For any $\vartheta>0$, there exists some positive constant $C$ such that
	\begin{equation*}
		\Prob\left(\LRtwonorm{\frac{1}{n}\sum_{i = 1}^n\bB_i\bB_i^{\top} - \E[\bB\bB^{\top}]}\geq CL_n\sqrt{\frac{\log n}{n}}\right)\leq n^{-\vartheta L_n}.
	\end{equation*}
\end{lemma}
The proof of Lemma \ref{lemma_Bnorm_bound} is deferred to Appendix \ref{proof::lemma_Bnorm_bound}.
\begin{remark}
    It is worthwhile noting that Lemma \ref{lemma_Bnorm_bound} is more tight than the results Lemma A.7 in \cite{fan2014nonparametric}, where the authors established the following bound
\begin{align*}
    \LRtwonorm{\frac{1}{n}\sum_{i = 1}^n\bB_i^{*}(\bB_i^{*})^{\top} - \E[\bB^{*}(\bB^{*})^{\top}]} = \Op\LRs{\sqrt{\frac{L_n \log n}{n}}}.
\end{align*}
By the fact that $\bB_i = \sqrt{L_n}\bB_i^{*}$ and $\bB = \sqrt{L_n}\bB^{*}$, using the relation above, we can only obtain the following worse bound
\begin{align*}
    \LRtwonorm{\frac{1}{n}\sum_{i = 1}^n\bB_i\bB_i^{\top} - \E[\bB\bB^{\top}]} = \Op\LRs{L_n^{3/2}\sqrt{\frac{ \log n}{n}}}.
\end{align*}
\end{remark}

\begin{proof}[Proof of Proposition \ref{pro_mean_bound}]
	First we introduce the population form of the approximation coefficient,
	\begin{equation}\label{pop_mean}
		\talpha_{1j} = \arg\min_{\balpha \in \mathbb{R}^{L_n}}\E\left(X_j-\balpha^{\top}\bB(U)\big|Y = 1\right)^2,
	\end{equation}
	and denote $\tmu_1(u) = (\talpha_{11}^{\top}\bB,\cdots,\talpha_{1p}^{\top}\bB)^{\top}$. According to the splines’ approximation property \citep{de1978practical,huang2003local}, we are guaranteed that
	\begin{equation*}
		\sup_{u\in [0,1]}\infnorm{\tmu_1(u) - \bmu_1(u)} \lesssim L_n^{-d}.
	\end{equation*}
	In addition, we have
	\begin{align}
	    \halpha_{1j}- \talpha_{1j}&= \left(\frac{1}{n}\sum_{i\in \mathcal{I}_1}\bB_i\bB_i^{\top}\right)^{-1}\left(\frac{1}{n}\sum_{i\in \mathcal{I}_1}\bB_i(X_{ij} - \bB_i^{\top}\talpha_{1j})\right)\nonumber\\
			&=\left(\frac{1}{n}\sum_{i\in \mathcal{I}_1}\bB_i\bB_i^{\top}\right)^{-1}\left(\frac{1}{n}\sum_{i\in \mathcal{I}_1}\bB_i[X_{ij}  - \bB_i^{\top}\talpha_{1j}]\right).
			\label{eq:alpha_decompose}
	\end{align}
	For any positive definite matrix $\bA\in \mathbb{R}^{p\times p}$ and $\twonorm{\Delta \bA} = o(1)$,
	\begin{equation}\label{inverse_fact}
		\begin{aligned}
		        \twonorm{(\bA+\Delta \bA)^{-1}-\bA^{-1}}&\leq \twonorm{\bA^{-1}}\twonorm{(\bI + \bA^{-1}\Delta \bA)^{-1}-\bI}\\
		        &\leq \twonorm{\bA^{-1}}\LRs{\twonorm{\bA^{-1}\Delta \bA} + o(1)}\\
		        &\leq 2\twonorm{\bA^{-1}}^2\twonorm{\Delta \bA}.
		\end{aligned}
	\end{equation}
	Now let $\bA = \E[\bB \bB^{\top}]$ and $\Delta \bA = \sum_{i=1}^n\bB_i\bB_i^{\top}/n-\E[\bB \bB^{\top}]$. Since Lemma \ref{lemma_Bnorm_bound} claims that $\twonorm{\Delta \bA} = o(1)$ almost surely, \eqref{inverse_fact} results in
	\begin{equation*}
		\begin{aligned}
			&\LRtwonorm{\left(\frac{1}{n}\sum_{i\in \mathcal{I}_1}\bB_i\bB_i^{\top}\right)^{-1} - \E[\bB\bB^{\top}]^{-1}}\\
			&\qquad \leq 2\LRtwonorm{ \E[\bB \bB^{\top}]^{-1}}^2\LRtwonorm{\frac{1}{n}\sum_{i=1}^n\bB_i\bB_i^{\top}-\E[\bB \bB^{\top}]}.
		\end{aligned}
	\end{equation*}
	In conjunction with Lemma \ref{lemma_Bnorm_bound} and $\lambda_{\min}(\E[\bB\bB^{\top}])\geq M_1$ we have
	\begin{equation}\label{ell_1}
		\LRtwonorm{\left(\frac{1}{n}\sum_{i\in \mathcal{I}_1}\bB_i\bB_i^{\top}\right)^{-1}}\leq \frac{1}{M_1}+ C\sqrt{\frac{L_n\log n}{n}}
	\end{equation}
	holds with probability at least $1-L_nn^{-\vartheta}$. Substituting \eqref{ell_1} into \eqref{eq:alpha_decompose} yields
	\begin{align}\label{eq:alpha_bound}
	    \twonorm{\halpha_{1j} - \talpha_{1j}} \lesssim \LRtwonorm{\frac{1}{n}\sum_{i\in \mathcal{I}_1}\bB_i[X_{ij} - \mu_{1j}(U_i) + \mu_{1j}(U_i) - \bB_i^{\top}\talpha_{1j}]}.
	\end{align}
	On the other hand, we note that given $Y_i = 1$ and $U_i$, $X_{ij}\sim \mathcal{N}(\mu_{1j}(U_i), \bSigma_{jj}(U_i))$. By choosing any $\eta > 0$ and using Lemma \ref{lemma:normal_moment_bound}, we have
	\begin{equation*}
		\begin{aligned}
			&\E\left[(B_k^{*}(U_i))^2\left(X_{ij} -\mu_{1j}(U_i) \right)^2\exp\left(\eta|B_k^{*}(U_i)||X_{ij} -\mu_{1j}(U_i)|\right)\right] \\
			\leq &\E\left\{(B_k^{*}(U_i))^2\E\left[\LRs{X_{ij} -\mu_{1j}(U_i)}^2\exp\left(\eta|X_{ij} -\mu_{1j}(U_i)|\right)\big| U_i\right]\right\}\\
			\leq &2\E\LRm{(B_k^{*}(U_i))^2 e^{\frac{\eta^2 \bSigma_{jj}(U_i)}{2}} \LRs{\eta^2 \bSigma_{jj}^2(U_i) + \bSigma_{jj}(U_i)}}\\
			\lesssim & \E[(B_k^{*}(U_i))^2] \leq L_n^{-1},
		\end{aligned}
	\end{equation*}
	where the first inequality follows from $0\leq B_k^{*}(U_i)\leq 1$ and $\Sigma_{jj}(U_i)$ is the conditional variance of $X_{ij}$ given $U_i$. Applying Lemma \ref{bern_ineq} and uniform bound, we can guarantee
	\begin{equation*}
		\Prob\left(\max_{1\leq k\leq L_n}\left|\frac{1}{n}\sum_{i\in \mathcal{I}_1}(X_{ij}-\mu_{1j}(U_i))B_k(U_i)\right|\lesssim \sqrt{\frac{\log n}{n}}\right) \leq L_n n^{-\vartheta}.
	\end{equation*}
	It yields that
	\begin{equation}\label{ell_2}
		\LRtwonorm{\frac{1}{n}\sum_{i\in \mathcal{I}_1}\bB_i(X_{ij} - \mu_{1j}(U_i))} \lesssim \sqrt{\frac{L_n\log n}{n}},
	\end{equation}
	holds with probability at least $1-2L_nn^{-\vartheta}$. In addition, we note that
	\begin{equation*}
		\begin{aligned}
			&\LRtwonorm{\frac{1}{n}\sum_{i\in \mathcal{I}_1}\bB_i[\mu_{1j}(U_i) - \bB_i^{\top}\talpha_{1j}]}\leq \LRtwonorm{\E[\bB[\mu_{1j}(U) - \bB^{\top}\talpha_{1j}]]}\\
			&\qquad +\LRtwonorm{\frac{1}{n}\sum_{i\in \mathcal{I}_1}\bB_i[\mu_{1j}(U_i) - \bB_i^{\top}\talpha_{1j}] - \E[\bB_i[\mu_{1j}(U_i) - \bB_i^{\top}\talpha_{1j}]]}.
		\end{aligned}
	\end{equation*}
	Using Lemma \ref{bern_ineq} again, we may show that
	\begin{equation*}
		\LRtwonorm{\frac{1}{n}\sum_{i\in \mathcal{I}_1}\bB_i[\mu_{1j}(U_i) - \bB_i^{\top}\talpha_{1j}] - \E[\bB_i[\mu_{1j}(U_i) - \bB_i^{\top}\talpha_{1j}]]} \lesssim L_n^{-d}\sqrt{\frac{L_n\log n}{n}},
	\end{equation*}
	holds with probability at least $1 - L_n n^{-\vartheta}$. Combining (\ref{ell_2}) and the following fact
	\begin{align*}
	    \twonorm{\E[\bB[\mu_{1j}(U) - \bB^{\top}\talpha_{1j}]]} \lesssim L_n^{-d} \twonorm{\E[\bB]} \lesssim L_n^{-d},
	\end{align*}
	we have with probability at least $1 - 3L_n n^{-\vartheta}$
	\begin{align}\label{eq:alpha_vec_bound}
	    \LRtwonorm{\frac{1}{n}\sum_{i\in \mathcal{I}_1}\bB_i[X_{ij} - \mu_{1j}(U_i) + \mu_{1j}(U_i) - \bB_i^{\top}\talpha_{1j}]} \lesssim \sqrt{\frac{L_n \log n}{n}}.
	\end{align}
	Together with \eqref{eq:alpha_bound}, we have proved the first assertion. 
	
    For each fixed $u \in [0,1]$, we denote $\bEta(u) = \LRs{\E[\bB\bB^{\top}]}^{-1} \bB^*(u)$.
	It holds that
	\begin{equation*}
		\begin{aligned}
			\LRinfnorm{\hmu_{1}(u)-\tmu_1(u)} &= \max_{j}|\bB(u)^{\top}(\hat{\balpha}_{1j}-\talpha_{1j})|\\
			&= \max_{j}\Bigg|\bB(u)^{\top}\LRs{\frac{1}{n}\sum_{i\in \mathcal{I}_1}\bB_i\bB_i^{\top}}^{-1}
			\frac{1}{n}\sum_{i\in \mathcal{I}_1}\bB_i[X_{ij} - \bB_i^{\top}\talpha_{1j}]\Bigg|\\
			&\lesssim L_n\max_{j}\LRabs{\frac{1}{n}\sum_{i\in \mathcal{I}_1}\bEta(u)^{\top}\bB_i^*[X_{ij} - \bB_i^{\top}\talpha_{1j}]},
		\end{aligned}
	\end{equation*}
	where the last inequality comes from Lemma \ref{lemma_Bnorm_bound} and \eqref{eq:alpha_vec_bound}. It follows from \eqref{int_bound} that
	\begin{align*}
	    \E\LRm{\LRs{\bEta(u)^{\top} \bB_i^*}^2} = \E\LRm{\LRs{\sum_{k=1}^{L_n} \eta_k(u) B^*(U)}^2} \lesssim L_n^{-1}\sum_{k=1}^{L_n} \eta_k^2(u) \lesssim L_n^{-1}.
	\end{align*}
	Further we have
	\begin{align*}
	    \E\LRm{\LRs{\bEta(u)^{\top} \bB_i^*}^2 (X_{ij} - \mu_{1j}(U_i))^2e^{\eta |\bEta(u)^{\top} \bB_i^*|X_{ij} - \mu_{1j}(U_i)||}} \lesssim L_n^{-1}.
	\end{align*}
	Then applying Lemma \ref{bern_ineq}, we can show that
	\begin{align*}
		&\max_{j}\LRabs{\frac{1}{n}\sum_{i\in \mathcal{I}_1}\bEta(u)^{\top} \bB_i^*[X_{ij} - \mu_{1j}(U_i) + \mu_{1j}(U_i) - \bB_i^{\top}\talpha_{1j}]}\lesssim \sqrt{\frac{\log n}{n L_n}}
	\end{align*}
	holds with probability at least $1 - p n^{-\vartheta}$. With the same probability, for any fixed $u$, it holds that
	\begin{align}\label{eq:fix_hmu_bound}
	    \LRinfnorm{\hmu_{1}(u)-\tmu_1(u)} \lesssim \sqrt{\frac{L_n \log n}{n}}.
	\end{align}
	Next we use chaining technique to prove the uniform result. Notice that, we may divide the interval $[0,1]$ to $n^{M}$ sub-intervals with end points $0=u_0\leq u_1 \leq \cdots u_{n^M} = 1$. Then for any $u \in [0,1]$, there exists some $0\leq \ell\leq n^M$ such that $|u - u_{\ell}|\leq n^{-M}$. Thus we have
	\begin{align*}
	    \LRinfnorm{\hmu_{1}(u)-\tmu_1(u)} &\leq \LRinfnorm{\hmu_{1}(u_{\ell})-\tmu_1(u_{\ell})}\\
	    &+ \LRl{\LRinfnorm{\hmu_{1}(u)-\hmu_1(u_{\ell})}+\LRinfnorm{\tmu_{1}(u)-\tmu_1(u_{\ell})}}\\
	    &\lesssim \LRinfnorm{\hmu_{1}(u_{\ell})-\tmu_1(u_{\ell})} + n^{-M},
	\end{align*}
	where we used both $\hmu_{1}$ and $\tmu$ are Lipschitz continuous. It follows that
	\begin{align*}
	    \sup_{u\in [0,1]}\LRinfnorm{\hmu_{1}(u)-\tmu_1(u)} \lesssim \max_{0\leq \ell \leq n^{M}}\LRinfnorm{\hmu_{1}(u_{\ell})-\tmu_1(u_{\ell})} + n^{-M}.
	\end{align*}
	By choosing $M = 1$ and large $\vartheta$, the second assertion follows from \eqref{eq:fix_hmu_bound} immediately.
\end{proof}

\section{Proofs of main results}\label{sec::main_proof}
\subsection{Proof of Proposition \ref{pro_mis_rate}}
    The proof is adapted from \cite{tony2019high}, here we provide it for completeness.
	\begin{lemma}[Lemma 7, \cite{tony2019high}]\label{lemma 4}
		For two vectors $\btheta_n$ and $\hat{\btheta}_n$, if $\|\btheta_n - \hat{\btheta}_n\|_2=o(1)$ as $n\to \infty$, and $\|\btheta\|_2\geq c$ for some constant $c$, then when $n\to \infty$,
		$$
		\|\btheta_n\|_2\|\hat{\btheta}_n\|_2-\btheta_n^{\top}\hat{\btheta}_n \asymp \|\btheta_n - \hat{\btheta}_n\|_2^2.
		$$
	\end{lemma}
	\begin{proof}
		Let $\bdelta(u) = \bmu_1(u) - \bmu_{2}(u)$ and $\hDelta(u) = c^{*}(u)(\hbtheta(u)^{\top}\bSigma(u)\hbtheta(u))^{1/2}$ for $u\in [0,1]$. Recall the relation $\bbeta^{*}(u) = \bSigma^{-1}(u)\bdelta(u) = c^{*}(u) \btheta^{*}(u)$ for $c^{*}(u)\in (0,1)$. To simplify notations, we will use $f$ to denote $f(u)$ for any function of $u$. Following the proof technique in \cite{tony2019high}, we define a intermediate quantity
		$$
		\widetilde{R}(u)=\frac{1}{2}\Phi\left(-\frac{c^{*}\bdelta^{\top}\hbtheta/2}{\widehat{\Delta}}\right)+\frac{1}{2}\bar{\Phi}\left(\frac{c^{*}\bdelta^{\top}\hbtheta/2}{\widehat{\Delta}}\right).
		$$
		Note that, $c^{*}\bdelta^{\top}\hbtheta$=$\bdelta^{\top}\bSigma^{-1/2}\bSigma^{1/2}(c^{*}\hbtheta)$, then using Lemma \ref{lemma 4}
		\begin{align}
		    \left|\Delta-\frac{c^{*}\bdelta^{\top}\hbtheta}{\widehat{\Delta}}\right|&=\left|\|\bdelta^{\top}\bSigma^{-1/2}\|_2-\frac{\bdelta^{\top}\bSigma^{-1/2}\bSigma^{1/2}(c^{*}\hbtheta)}{\left\|\bSigma^{1/2}(c^{*}\hbtheta)\right\|_2}\right|\nonumber\\
			&\lesssim \frac{\left\|\bSigma^{-1/2}\bdelta-\bSigma^{1/2}(c^{*}\hbtheta)\right\|_2^2}{\Delta}.\label{eq:dist_expan_1}
		\end{align}
		Using the fact that $c^{*}\btheta^{*} = \bbeta^* =\bSigma^{-1}\bdelta$ and $\|\bSigma\|_2$ is bounded,
		\begin{align}
		    \left\|\bSigma^{-1/2}\bdelta-\bSigma^{1/2}(c^{*}\hbtheta)\right\|_2^2&\leq \left\|\bSigma^{-1/2}\bdelta-\bSigma^{1/2}(c^{*}\hbtheta)\right\|_2^2\nonumber\\
			&=\left\|\bSigma^{-1/2}\bdelta-\bSigma^{1/2}(c^{*}\hbtheta-c^{*}\btheta^{*})-\bSigma^{1/2}c^{*}\btheta^{*}\right\|_2^2\nonumber\\
			&\leq (c^{*})^2\|\bSigma^{1/2}\|_2^2\|\hbtheta-\btheta^{*}\|_2^2.
			\label{eq:dist_expan_2}
		\end{align}
		Then take Taylor expansion to the two terrms of $\widetilde{R}$ around $-\Delta/2$ and $\Delta/2$ respectively,  we have
		\begin{align*}
			\widetilde{R}-R&=\frac{1}{2}\left(-\frac{c^{*}\bdelta^{\top}\hbtheta/2}{\widehat{\Delta}}+\frac{\Delta}{2}\right)\Phi^{'}\left(-\frac{\Delta}{2}\right)+\frac{1}{2}\left(\frac{\Delta}{2}-\frac{c^{*}\bdelta^{\top}\hbtheta/2}{\widehat{\Delta}}\right)\Phi^{'}\left(\frac{\Delta}{2}\right)\\
			&\qquad+\frac{1}{4}\left(\left(\Phi^{''}(b_{1n})+\Phi^{''}(b_{2n})\right)\left(\frac{c^{*}\bdelta^{\top}\hbtheta/2}{\widehat{\Delta}}-\frac{\Delta}{2}\right)^2\right)\\
			&=\frac{1}{\sqrt{2\pi}}\left(\Delta-\frac{c^{*}\bdelta^{\top}\hbtheta}{\widehat{\Delta}}\right)\exp\left(-\frac{\Delta^2}{8}\right)\\
			&\qquad+\frac{1}{4}\left(\left(\Phi^{''}(b_{1n})+\Phi^{''}(b_{2n})\right)\left(\frac{c^{*}\bdelta^{\top}\hbtheta/2}{\widehat{\Delta}}-\frac{\Delta}{2}\right)^2\right),
		\end{align*}
		where $b_{1n}$ is some point between $-\frac{\Delta}{2}$ and $-\frac{c^{*}\bdelta^{\top}\hbtheta/2}{\widehat{\Delta}}$ and $b_{2n}$ is some point between $\frac{\Delta}{2}$ and $\frac{c^{*}\bdelta^{\top}\hbtheta/2}{\widehat{\Delta}}$. Hence it holds that
		$$
		\Phi^{''}(b_{1n})\asymp\Phi^{''}(b_{2n})\asymp \frac{\Delta}{2}\exp\left(-\frac{\Delta^2}{8}\right).
		$$
		Together with \eqref{eq:dist_expan_1} and \eqref{eq:dist_expan_2}, we can obtain
		\begin{align}\label{decom_R_1}
		    |\widetilde{R}-R|&\lesssim \left|\Delta-\frac{c^{*}\bdelta^{\top}\hbtheta}{\widehat{\Delta}}\right| + \left|\Delta-\frac{c^{*}\bdelta^{\top}\hbtheta}{\widehat{\Delta}}\right|^2\nonumber\\
		    &\lesssim \left\|\bSigma^{-1/2}\bdelta-\bSigma^{1/2}(c^{*}\hbtheta)\right\|_2^2\nonumber\\
		    &\lesssim \twonorm{\hbtheta - \btheta^{*}}^2.
		\end{align}
		Next we bound $|R_n-\widetilde{R}|$, note that
		\begin{align*}
			&\left|\frac{c^{*}\bdelta^{\top}\hbtheta/2}{\widehat{\Delta}}-\frac{\left(\hmu - \bmu_{2}\right)^{\top} (c^{*}\hbtheta)}{\widehat{\Delta}}\right|\nonumber\\
			&\qquad=\frac{\left|(\bdelta/2-\hmu + \bmu_{2})^{\top}\bSigma^{-1}\bdelta+(\bdelta/2-\hmu + \bmu_{2})^{\top}(c^{*}\hbtheta-c^{*}\btheta^{*})\right|}{\widehat{\Delta}}\nonumber\\
			&\qquad \lesssim \frac{1}{\Delta}\left(|\left(\hmu - \bmu\right)^{\top}\bbeta^{*}|+\twonorm{\hbtheta - \btheta}\twonorm{\hmu - \bmu}\right)\nonumber\\
			&\qquad\lesssim \frac{|\left(\hmu - \bmu\right)^{\top}\bbeta^{*}|}{\Delta},
		\end{align*}
		where the first inequality follows from that $\Delta$ is bounded. Take Taylor expansion on the two terms of $R_n(t)$ around $-\frac{c^{*}\bdelta^{\top}\hbtheta/2}{\widehat{\Delta}}$ and $\frac{c^{*}\bdelta^{\top}\hbtheta/2}{\widehat{\Delta}}$ respectively, we have
		\begin{align}
			R_n(t)-\widetilde{R}&=\frac{1}{2}\left(\frac{\left(\hmu - \bmu_{1}\right)^{\top} (c^{*}\hbtheta)}{\widehat{\Delta}}+\frac{c^{*}\bdelta^{\top}\hbtheta/2}{\widehat{\Delta}}\right)\Phi^{'}\left(-\frac{c^{*}\bdelta^{\top}\hbtheta/2}{\widehat{\Delta}}\right)\nonumber\\
			&-\frac{1}{2}\left(\frac{\left(\hmu - \bmu_{2}\right)^{\top} (c^{*}\hbtheta)}{\widehat{\Delta}}-\frac{c^{*}\bdelta^{\top}\hbtheta/2}{\widehat{\Delta}}\right)\Phi^{'}\left(\frac{c^{*}\bdelta^{\top}\hbtheta/2}{\widehat{\Delta}}\right)\nonumber\\
			&+\frac{1}{4}\left(\Phi^{''}(b_{3n})\left(\frac{\left(\hmu - \bmu_{1}\right)^{\top} (c^{*}\hbtheta)}{\widehat{\Delta}}+\frac{c^{*}\bdelta^{\top}\hbtheta/2}{\widehat{\Delta}}\right)^2\right)\nonumber\\
			&+\frac{1}{4}\left(\Phi^{''}(b_{4n})\left(\frac{\left(\hmu - \bmu_{2}\right)^{\top} (c^{*}\hbtheta)}{\widehat{\Delta}}-\frac{c^{*}\bdelta^{\top}\hbtheta/2}{\widehat{\Delta}}\right)^2\right)\nonumber\\
			&=\frac{\Phi^{''}(b_{3n}) + \Phi^{''}(b_{4n})}{4} \frac{\left(\hmu - \bmu\right)^{\top} (c^{*}\hbtheta)}{\widehat{\Delta}},\label{eq:Rn_Rt_diff}
		\end{align}
		where $b_{3n}$ is some point between $\frac{\left(\hmu - \bmu_{1}\right)^{\top} (c^{*}\hbtheta)}{\widehat{\Delta}}$ and $\frac{c^{*}\bdelta^{\top}\hbtheta/2}{\widehat{\Delta}}$ and $b_{4n}$ is some point between $\frac{\left(\hmu - \bmu_{1}\right)^{\top} (c^{*}\hbtheta)}{\widehat{\Delta}}$ and $-\frac{c^{*}\bdelta^{\top}\hbtheta/2}{\widehat{\Delta}}$, then 
		$$
		\Phi^{''}(b_{3n})\asymp\Phi^{''}(b_{4n})\asymp \frac{\Delta}{2}\exp\left(-\frac{\Delta^2}{8}\right).
		$$
		In fact, we also used $\bmu - \bmu_1 = -\bdelta/2$ and $\bmu - \bmu_2 = \delta/2$ to obtain \eqref{eq:Rn_Rt_diff}.
		Then \eqref{eq:Rn_Rt_diff} implies
		\begin{align}
		    |R_n(t)-\widetilde{R}|&\lesssim \LRabs{\left(\hmu - \bmu\right)^{\top} (c^{*}\btheta^*)}^2 + \LRabs{\left(\hmu - \bmu\right)^{\top} (\hbtheta-\btheta^*)}^2\nonumber\\
			&\lesssim |\left(\hmu - \bmu\right)^{\top}\bbeta^{*}|^2.\label{decom_R_2}
		\end{align}
		Combining \eqref{decom_R_1} and \eqref{decom_R_2}, for any $u\in [0,1]$, it holds that
		\begin{align*}
		    |R_n(u)-R(u)|&\leq |R_n(u)-\widetilde{R}(u)|+|\widetilde{R}(u)-R(u)|\\
		    &\lesssim \twonorm{\hbtheta(u) - \btheta^{*}(u)}^2+\LRabs{\left(\hmu(u) - \bmu(u)\right)^{\top}\bbeta^{*}(u)}^2.
		\end{align*}
	\end{proof}

\subsection{Proof of Theorem \ref{thm_appro_error} and \ref{thm_high_appro_error}}\label{proof:thm_appro_error}
Here we only prove Theorem \ref{thm_appro_error}, and the proof of Theorem 3.1 can be easily obtained through the similar analysis. The following lemma provides the lower bound and upper bound for the eigenvalues of $\E[\tB\tB^{\top}]$, and the proof is deferred to Appendix \ref{proof::lemma_eigen}.
\begin{lemma}\label{lemma_eigen}
	Assume the assumptions hold, then there exist two positive constant $M_1$ and $M_2$ such that
	\begin{equation*}
		M_1\lambda_0\leq \lambda_{\min}(\E[\tB\tB^{\top}])\leq \lambda_{\max}(\E[\tB\tB^{\top}])\leq M_2(\lambda_1+\delta_p/4).
	\end{equation*}
\end{lemma}
\begin{proof}[Proof of Theorem \ref{thm_appro_error}]
	There exists $\bar{\theta}_j(U) = \bar{\bgamma}_{j}^{\top}\bB(U)$ for $j=0,1,...,p$ such that 
	\begin{equation}\label{spline_error}
		\sup_{u\in [0,1]}|\theta_j^{*}(u)-\bar{\theta}_j(u)|\leq M_0L_n^{-d}.
	\end{equation}
	Let $\bar{\bgamma}= (\bar{\bgamma}_{1}^{\top}, \cdots, \bar{\bgamma}_{p}^{\top})^{\top}$, then note that
	\begin{equation}\label{eq:tgamma_error}
		\begin{aligned}
			\tgamma - \bar{\bgamma} = \left[\E\left(\tB\tB^{\top}\right)\right]^{-1}\E\left[\tB \left(Z - \tB^{\top}\bar{\bgamma}\right)\right],
		\end{aligned}
	\end{equation}
	The optimality condition of $\theta_{j}^{*}(U)$ implies that
	\begin{equation*}
		\E\left[(X_j-\mu_{j}(U))\left(Z-\sum_{l=1}^p(X_j-\mu_{j}(U))\theta_{j}^{*}(U)\right)\bigg| U\right]=0,
	\end{equation*}
	which means
	\begin{align*}
	    \E\LRm{(X_j-\mu_{j}(U))\bB(U) \left(Z-\sum_{l=1}^p(X_j-\mu_{j}(U))\theta_{j}^{*}(U)\right)} = \boldsymbol{0}.
	\end{align*}
	Recall $\tB = (\bX - \bmu)\otimes \bB$, then we can get
	\begin{align*}
		\E\left[\tB \left(Z - \tB^{\top}\bar{\bgamma}\right)\right] &= \E\left[\tB \left( \sum_{j=1}^p\theta_{j}^{*}(U)(X_j-\mu_{j}(U))- \sum_{j=1}^p\bar{\theta}_j(U)(X_j-\mu_{j}(U)) \right)\right]\\
		&= \E\LRm{\tB (\bX - \bmu)^{\top}(\btheta^{*}(U) - \bar{\btheta}(U))}.
	\end{align*}
	Let $\bC(U) = \E\LRm{(\bX-\bmu(U))(\bX-\bmu(U))^{\top}|U}$, then simple calculation yields that
	\begin{align*}
	    \bC(U) = \bSigma(U)+\frac{1}{4}(\bmu_1(U)-\bmu_2(U))(\bmu_1(U)-\bmu_2(U))^{\top}.
	\end{align*}
	For $\bnu = (\bnu_{(1)}^{\top},...,\bnu_{(p)}^{\top})^{\top}$, we denote $\tnu(U) = (\bnu_{(1)}^{\top}\bB(U),...,\bnu_{(p)}^{\top}\bB(U))^{\top}$. Then we have
	\begin{align}
	    &\LRtwonorm{\E[\tB (Z - \tB^{\top}\bar{\bgamma})]}\nonumber\\
			= & \sup_{\twonorm{\bnu} = 1} \LRabs{\E[\bnu^{\top}\tB (Z - \tB^{\top}\bar{\bgamma})]}\nonumber\\
			= & \sup_{\twonorm{\bnu} = 1} \left|\E[\tnu(U)^{\top}(\bX-\bmu(U))(\bX-\bmu(U))^{\top}(\btheta^{*}(U) - \bar{\btheta}(U))]\right|\nonumber\\
			= & \sup_{\twonorm{\bnu} = 1} \left|\E\LRm{\twonorm{\tnu(U)}\twonorm{\bC(U)}\twonorm{\btheta^{*}(U) - \bar{\btheta}(U)}}\right|\nonumber\\
			\lesssim & \sqrt{p} L_n^{-d}\sup_{\twonorm{\bnu} = 1}\E[\twonorm{\tnu(U)}],
			\label{eq:tgamma_error_1}
	\end{align}
	where the last inequality follows from \eqref{spline_error} and $\twonorm{\bC(u)} \leq \twonorm{\bSigma(u)} + \delta_p$. Using the inequality (\ref{int_bound}) and $\twonorm{\bnu} = 1$, we have 
	\begin{equation*}
		\E[\twonorm{\tnu}]\leq \LRs{\E[\twonorm{\tnu}^2]}^{1/2} = L_n\LRs{\sum_{j=1}^p\E[(\bnu_{(j)}^{\top}\bB^*)^2]}^{1/2} \lesssim \LRs{\sum_{j=1}^p \twonorm{\bnu_{(j)}}^2}^{1/2}.
	\end{equation*}
	Combining \eqref{eq:tgamma_error} and \eqref{eq:tgamma_error_1}, we are guaranteed that
	$\twonorm{\tgamma - \bar{\bgamma}} \lesssim \sqrt{p}L_n^{-d}$. Recall that $\tbtheta_j(u) = \tgamma_j^{\top} \bB(u)$, together with \eqref{int_bound}, we can have
	\begin{align*}
	    \|\tbtheta-\bar{\btheta}\|_{L_2}^2 = \int_{0}^1\LRtwonorm{\tbtheta(u) - \bar{\btheta}(u)}^2du &=L_n \sum_{j=1}^p\int_{0}^1\LRs{\LRs{\tgamma_{(j)} - \bar{\btheta}_{(j)}}^{\top} \bB^*(u)}^2 du\\
	    &\lesssim \sum_{j=1}^p \twonorm{\tgamma_{(j)} - \bar{\btheta}_{(j)}}^2\\
	    &=\twonorm{\tgamma - \bar{\bgamma}}^2.
	\end{align*}
	It yields that
	\begin{align*}
	    \|\btheta^{*}-\tbtheta\|_{L_2} \leq \|\btheta^{*}-\bar{\btheta}\|_{L_2} + \|\tbtheta-\bar{\btheta}\|_{L_2} \lesssim \sqrt{p}L_n^{-d}.
	\end{align*}
\end{proof}

\subsection{Proof of Theorem \ref{thm_low_dim_ell2}}\label{proof:thm_ld_ell2}
Let $\bD_n = \frac{1}{2n}\sum_{i=1}^{2n}\tB_i\tB_i^{\top}$ and $\bb_n = \frac{1}{2n}\sum_{i=1}^{2n}\tB_i Z_i$. Correspondingly, we write $\bD = \E[\tB\tB^{\top}]$ and $\bb = \E[\tB Z]$. The following two lemmas give the concentration bounds for two terms in estimation error $\|\hgamma - \tgamma\|_2$. We defer the proofs in Section \ref{sec:main_lemma_proofs}.

\begin{lemma}\label{lemma:tB_matrix_concentration}
	Under the conditions of Theorem \ref{thm_low_dim_ell2}, we have
	    \begin{equation}\label{eq:tB_matrix_concentration}
	        \LRtwonorm{\bD_n - \bD} \lesssim L_n\sqrt{\frac{p\log n}{n}}+pL_n^{3/2}a_n \sqrt{\frac{\log n}{n}},
        \end{equation}
    holds with probability at least $1 - n^{-\vartheta L_np} - pL_n n^{-\vartheta}$.
\end{lemma}
\begin{lemma}\label{lemma:tB_tgamma_concentration}
    Under the conditions of Theorem \ref{thm_low_dim_ell2}, we have
	    \begin{equation}\label{eq:tB_matrix_concentration_matvec}
	        \LRtwonorm{\bD_n\tgamma - \bD\tgamma} \lesssim \sqrt{\frac{p L_n\log n}{n}}+pL_na_n \sqrt{\frac{\log n}{n}},
        \end{equation}
    holds with probability at least $1 - n^{-\vartheta L_np} - pL_n n^{-\vartheta}$.
\end{lemma}
\begin{lemma}\label{lemma:bn_concentration}
    Under the conditions of Theorem \ref{thm_low_dim_ell2}, we have
    \begin{equation}
        \LRinfnorm{\bb_n - \bb} \lesssim \sqrt{\frac{\log n}{n}}+ L_n^{-\frac{1}{2}}a_n,
    \end{equation}
    holds with probability $1 - pL_n n^{-\vartheta}$.
\end{lemma}

\begin{proof}[Proof of Theorem \ref{thm_low_dim_ell2}]
    From the definition of $\hgamma$ and $\tgamma$, we have
	\begin{align}
	    \hgamma - \tgamma &= \bD_n^{-1}\bb_n - \tgamma = \bD_n^{-1}\LRs{\bb_n - \bD_n \tgamma}.
	    \label{eq:hgamma_decomposition}
	\end{align}
	Now let us recall the optimal condition of $\tgamma$,
	\begin{align}\label{eq:bb_equality}
	    \boldsymbol{0} = \E\LRm{\tB\LRs{Z - \tB^{\top}\tgamma}} = \bb - \bD\tgamma.
	\end{align}
	In addition, notice that
	\begin{align*}
	    \bb &= \E\LRm{\LRs{\bX - \bmu(U)}\otimes \bB Z}\\
	    &=\frac{1}{2}\E\LRm{\LRs{\bmu_1(U) - \bmu(U)}\otimes \bB} - \frac{1}{2}\E\LRm{\LRs{\bmu_2(U) - \bmu(U)}\otimes \bB}\\
	    &=\E\LRm{\LRs{\bmu_1(U) - \bmu_2(U)}\otimes \bB}.
	\end{align*}
	For $\bnu = (\bnu_{(1)}^{\top},...,\bnu_{(p)}^{\top})^{\top} \in \mathbb{R}^{pL_n}$, we denote $\tnu(U) = (\bnu_{(1)}^{\top}\bB(U),...,\bnu_{(p)}^{\top}\bB(U))^{\top}$. By the definition of $\twonorm{\cdot}$ and \eqref{int_bound}, we have
	\begin{align}\label{eq:b_bound}
	    \twonorm{\bb} &= \sup_{\bnu \in \mathbb{S}^{pL_n-1}} \LRabs{\bnu^{\top}\bb}\nonumber\\
	    &=\sup_{\bnu \in \mathbb{S}^{pL_n-1}}\LRabs{\E\LRm{\bnu^{\top}\LRs{\LRs{\bmu_1(U) - \bmu_2(U)}\otimes \bB}}}\nonumber\\
	    &=\sup_{\bnu \in \mathbb{S}^{pL_n-1}}\LRabs{\E\LRm{\tnu(U)^{\top}\LRs{\bmu_1(U) - \bmu_2(U)}}}\nonumber\\
	    &\leq \delta_p\sup_{\bnu \in \mathbb{S}^{pL_n-1}}\E\LRm{\twonorm{\tnu(U)}}\nonumber\\
	    &\leq \delta_p\sup_{\bnu \in \mathbb{S}^{pL_n-1}} \LRs{\E[\twonorm{\tnu(U)}^2]}^{1/2}\nonumber\\
	    &=\delta_p\sup_{\bnu \in \mathbb{S}^{pL_n-1}}\LRs{\sum_{j=1}^p\E\LRm{(\bnu_{(j)}^{\top}\bB(U))^2}}^{1/2}\nonumber\\
	    &\lesssim \delta_p \sup_{\bnu \in \mathbb{S}^{pL_n-1}}\LRs{\sum_{j=1}^p \twonorm{\bnu_{(j)}}^2}^{1/2} = \delta_p,
	\end{align}
	where $\delta_p = \sup_{u\in [0,1]}\twonorm{\bmu_1(u) - \bmu_2(u)}$. According to Lemma \ref{lemma:tB_matrix_concentration}, we know that $\twonorm{\bD_n - \bD} = o_{\mathbb{P}}(1)$. Using the inequality (\ref{inverse_fact}), we get
	\begin{align*}
	    \LRtwonorm{\bD_n^{-1}-\bD^{-1}} &\leq 2\LRtwonorm{\bD^{-1}}\LRtwonorm{\bD_n - \bD}\\
		&\leq 2M_1^{-1}\lambda_0^{-1}\LRtwonorm{\bD_n - \bD},
	\end{align*}
	where the second inequality follows from Lemma \ref{lemma_eigen}. Hence we can guarantee that $\twonorm{\bD_n^{-1}} = O(1)$ with high probability. By plugging the bounds in Lemma \ref{lemma:tB_tgamma_concentration} and \ref{lemma:bn_concentration}, together with \eqref{eq:bb_equality}, we have
	\begin{align}
	    \twonorm{\hgamma - \tgamma} &\leq \twonorm{\bD_n^{-1}}\twonorm{\bb_n - \bD_n\tgamma}\nonumber\\
	    &=\twonorm{\bD_n^{-1}}\twonorm{\bb_n - \bD_n\tgamma - \bb + \bD\tgamma}\nonumber\\
	    &\lesssim \twonorm{\bb_n - \bb} + \twonorm{\bD_n\tgamma -\bD\tgamma}\nonumber\\
	    &\lesssim \sqrt{\frac{p L_n\log n}{n}}+a_n pL_n \sqrt{\frac{\log n}{n}}.
	    \label{eq:gamma_est_bound_low_dim}
	\end{align}
	Recall $\hbtheta = (\bB(u)^{\top}\hgamma_{(1)},...,\bB(u)^{\top}\hgamma_{(p)})^{\top}$ and $\tbtheta(u) = (\bB(u)^{\top} \tgamma_{(1)},...,\bB(u)^{\top} \tgamma_{(p)})^{\top}$. Applying \eqref{int_bound}, we have
	\begin{align*}
	    \int_{0}^1 \twonorm{\hbtheta(u) - \tbtheta(u)}^2 du &= L_n\sum_{j=1}^p \int_{0}^1 \LRs{\bB^{*}(u)^{\top}(\hgamma_{(j)} - \tgamma_{(j)})}^2 du\\
	    &\lesssim \sum_{j=1}^p \twonorm{\hgamma_{(j)} - \tgamma_{(j)}}^2= \twonorm{\hgamma - \tgamma}^2.
	\end{align*}
	Then we have finished the proof of Theorem \ref{thm_low_dim_ell2} by plugging \eqref{eq:gamma_est_bound_low_dim}.
\end{proof}

\subsection{Proof of Theorem \ref{thm_hd_ell2}}\label{proof:thm_hd_ell2}
 The following lemma provides the $\ell_2$ error bound for general quadratic group lasso problem. We defer the proof of Lemma \ref{g_lasso} to Appendix \ref{proof::g_lasso}.
\begin{lemma}\label{g_lasso}
	For general quadratic group lasso problem
	\begin{equation*}
	\hgamma = \arg\min_{\bgamma \in \mathbb{R}^{pL_n}} \frac{1}{2}\bgamma^{\top}\mathbf{A}\bgamma - \bb^{\top}\bgamma + \lambda\sum_{j=1}^p\twonorm{\bgamma_{j}},
	\end{equation*}
	if the following two conditions hold
	\begin{enumerate}
	 \item $\mathbf{A}$ satisfies the restrictive eigenvalue condition with parameter $\zeta$: for any $\bxi \in \mathbb{R}^{pL_n}$ such that $\onenorm{\bxi}\leq 4\sqrt{sL_n}\twonorm{\bxi}$, it holds that
	\begin{equation*}
	\bxi^{\top}\mathbf{A}\bxi\geq \zeta \twonorm{\bxi}.
	\end{equation*}
	\item for any $\cgamma \in \mathbb{R}^{pL_n}$ such that $\cgamma_{(j)} = \boldsymbol{0}$ for $j \in S^c$ and
	\begin{equation}\label{em_process}
	\max_{1\leq j\leq p}\twonorm{(\mathbf{A}\cgamma - \bb)_{(j)}}\leq \frac{\lambda}{2}.
	\end{equation}
	\end{enumerate}
	then we have
	\begin{equation*}
	\twonorm{\hgamma - \cgamma} \leq \frac{12\sqrt{s}\lambda}{\zeta}\quad \text{and}\quad \onenorm{\hgamma - \cgamma} \leq \frac{48s\sqrt{L_n}\lambda}{\zeta}.
	\end{equation*}
\end{lemma}
\begin{lemma}\label{lemma_hd_inf}
	Under conditions \cond{1}-\cond{5}, let $\bnu\in \mathbb{R}^{pL_n}$ be a fixed vector with $\bnu_{(S^c)} = \boldsymbol{0}$, then for any $\vartheta>0$ we have
	\begin{equation*}
		\max_{1\leq j \leq p}\twonorm{(\bD_n\tgamma - \bD\tgamma)_{(j)}}\lesssim \twonorm{\bnu}\left(\sqrt{\frac{L_n\log p}{n}}+a_nL_n s\sqrt{\frac{\log p}{n}}\right),
	\end{equation*}
holds with probability at least $1-L_np^{-\vartheta} - L_nsp^{-\vartheta} - sp^{-\vartheta L_n}$.
\end{lemma}

\begin{proof}[Proof of Theorem \ref{thm_hd_ell2}]
    According to Lemma \ref{g_lasso}, it suffices to show the restrictive eigenvalue condition of $\bD_n$ and the inequality (\ref{em_process}). For any $\bxi \in \mathbb{R}^{pL_n}$ such that $\onenorm{\bxi}\leq 4\sqrt{sL_n}\twonorm{\bxi}$, we have
	\begin{equation*}
		\begin{aligned}
			\bxi^{\top}\bD_n\bxi &= \bxi^{\top}\E[\tB\tB^{\top}]\bxi + \bxi^{\top}(\bD_n-\bD)\bxi\\
			&\geq M_1\lambda_0\twonorm{\bxi}^2 - \onenorm{\bxi}\infnorm{(\bD_n-\bD)\bxi}\\
			&\geq M_1\lambda_0\twonorm{\bxi}^2 - \onenorm{\bxi}^2\left|\bD_n-\bD\right|_{\infty}\\
			&\geq \left(M_1\lambda_0 - 4s L_n\left|\bD_n-\bD\right|_{\infty}\right)\twonorm{\bxi}^2.
		\end{aligned}
	\end{equation*}
	By tracing the proof of Lemma \ref{lemma_hd_inf}, we can guarantee $s L_n|\bD_n-\bD|_{\infty} = o_{\mathbb{P}}(1)$. It implies that there exists some positive constant $\zeta$ such that,
	\begin{align*}
		\bxi^{\top}\bD_n\bxi \geq \zeta \twonorm{\bxi}^2.
	\end{align*}
	Hence we have verified the restrictive eigenvalue condition. Recall the approximation coefficient $\tgamma_{(S)} = \bD_{(SS)}^{-1}\bb_{(S)}$ and $\tgamma_{(S^c)} = \boldsymbol{0}$, then we have
	\begin{equation}\label{em_1}
		(\bD_n\tgamma - \bb_n)_{(S)} = (\bD_n)_{(SS)}\bD_{(SS)}^{-1}\bb_{(S)} - (\bb_n)_{(S)},
	\end{equation}
	and
	\begin{equation}\label{em_2}
		(\bD_n\tgamma - \bb_n)_{(S^c)} = (\bD_n)_{(S^cS)}\bD_{(SS)}^{-1}\bb_{S} - (\bb_n)_{(S^c)}.
	\end{equation}
	In addition, similar to \eqref{eq:b_bound}, we can verify $\twonorm{\bb_{(S)}} \lesssim \delta_s$. Together with Lemma \ref{lemma_eigen}, we have
	\begin{align*}
	    \twonorm{\tgamma_{(S)}}&\leq \twonorm{\bD_{(SS)}^{-1}}\twonorm{\bb_{(S)}}\lesssim \delta_s.
	\end{align*}
	Then from (\ref{em_1}), Lemma \ref{lemma_hd_inf} and \ref{lemma:bn_concentration}, for any $j \in S$
	\begin{equation*}
		\begin{aligned}
			\twonorm{(\bD_n\tgamma - \bb_n)_{(j)}} &\leq \max_{j\in S}\LRtwonorm{\left[(\bD_n)_{(j,S)} - \bD_{(j,S)}\right]\tgamma_{(S)}} + \sqrt{L_n}\LRinfnorm{\bb_{(S)} - (\bb_n)_{(S)}}\\
			&\lesssim \sqrt{\frac{L_n\log p}{n}}+a_nL_ns\sqrt{\frac{\log p}{n}} + \sqrt{\frac{L_n\log p}{n}} + a_n\\
			&\lesssim \sqrt{\frac{L_n\log p}{n}} + a_nL_ns\sqrt{\frac{\log p}{n}} + a_n,
		\end{aligned}
	\end{equation*}
	holds with probability at least $1-L_np^{-\vartheta} - L_nsp^{-\vartheta} - sp^{-\vartheta L_n}$. Next we will derive the bound for $j \in S^c$. From (\ref{em_2}), we claim that for any $j \in S^c$
	\begin{equation*}
		\begin{aligned}
			\twonorm{(\bD_n\tgamma - \bb_n)_{(j)}}&\leq \LRtwonorm{(\bD_{(S^cS)}\tgamma_{(S)})_{(j)} - \bb_{(j)}} + \LRtwonorm{\left[(\bD_n)_{(j,S)} - \bD_{(j,S)}\right]\tgamma_{(S)}}\\
			&+ \sqrt{L_n}\LRinfnorm{\bb_{(S^c)} - (\bb_n)_{(S^c)}}.
		\end{aligned}
	\end{equation*}
	Using the optimality of $\btheta^{*}(U)$, we have
	\begin{equation*}
		\E\left\{\tB_{(S^c)}\left[Z - \sum_{j\in S}(X_j  - \mu_{j}(U))\theta_{j}^{*}(U)\right]\right\} = \boldsymbol{0}.
	\end{equation*}
	Together with $\widetilde{\theta}_j(U) = \bB(U)^{\top}\tgamma_{(j)}$, $\bb_{(S^c)} = \E[\tB_{(S^c)}Z]$ and $\tB_{(j)} = (X_j - \mu_j(U))\bB(U)$, we also have
	\begin{equation*}
		\begin{aligned}
			\bD_{(S^cS)}\tgamma_{(S)} &= \E\left\{\tB_{(S^c)}\tB_{(S)}^{\top}\tgamma_{(S)}\right\}\\
			& = \E\left\{\tB_{(S^c)}\sum_{j\in S}(X_j  - \mu_{j}(U))\bB(U)^{\top}\tgamma_{(j)}\right\}\\
			&=\E\left\{\tB_{(S^c)}\left[\sum_{j\in S}(X_j  - \mu_{j}(U))\widetilde{\theta}_j(U) - Z\right] \right\} + \bb_{(S^c)}\\
			&=\E\left\{\tB_{(S^c)}\left[\sum_{j\in S}(X_j  - \mu_{j}(U))(\widetilde{\theta}_j(U) - \theta_j^*(U)) \right] \right\} + \bb_{(S^c)}.
		\end{aligned}
	\end{equation*}
	Let $\boldsymbol{c}_{j,S}(U) = \E[(X_j - \mu_{j}(U))(\bX - \bmu(U))_{S}|U]$, then for $j\in S^c$, it holds that
	\begin{align*}
	    (\bD_{(S^cS)}\tgamma_{(S)})_{(j)} - \bb_{(j)} = \E\left[\tB_{(j)}(\bX - \bmu(U))_{S}^{\top}(\btheta^{*}(U) - \tbtheta(U))_{S}\right].
	\end{align*}
	For any $\bnu \in \mathbb{S}^{L_n -1}$, 
	\begin{align*}
	    &\LRabs{\E\LRm{\bnu^{\top}\tB_{(j)}(\bX - \bmu(U))_{S}^{\top}(\btheta^{*}(U) - \tbtheta(U))_{S}}}\\
	    &= \LRabs{\E\LRm{\bnu^{\top}\bB (X_j - \mu_j(U))(\bX - \bmu(U))_{S}^{\top}(\btheta^{*}(U) - \tbtheta(U))_{S}}}\\
	    &= \LRabs{\E\LRm{\bnu^{\top}\bB \boldsymbol{c}_{j,S}(U)^{\top}(\btheta^{*}(U) - \tbtheta(U))_{S}}}\\
	    &\leq \sup_{u\in [0,1]}\left\{\twonorm{\boldsymbol{c}_{j,S}(u)}\twonorm{(\btheta^{*}(u) - \tbtheta(u))_{S}}\right\}\E\LRm{|\bnu^{\top}\bB|}\\
	    &\lesssim \delta_s \sqrt{s}L_n^{-d} \LRs{\E\LRm{\LRs{\bnu^{\top}\bB}^2}}^{1/2}\\
	    &\lesssim \sqrt{s}L_n^{-d},
	\end{align*}
	where we used Theorem \ref{thm_high_appro_error} and
	\begin{align*}
	    \sup_{u\in [0,1]}\twonorm{\boldsymbol{c}_{j,S}(U)} &= \sup_{u\in [0,1]}\twonorm{(\mu_{1j}(u) - \mu_{j}(u))(\bmu_1(u) - \bmu(u))_{S}}\\
	    &\lesssim \sup_{u\in [0,1]}\twonorm{(\bmu_1(u) - \bmu_2(u))_{S}} \leq \delta_s. 
	\end{align*}
	Hence we have for any $j \in S^c$,
	\begin{equation*}
		\LRtwonorm{(\bD_{(S^cS)}\tgamma_{(S)})_{(j)} - \bb_{(j)}}\leq \sup_{\bnu \in \mathbb{S}^{L_n-1}}\LRabs{\bnu^{\top}\LRs{(\bD_{(S^cS)}\tgamma_{(S)})_{(j)} - \bb_{(j)}}} \lesssim \delta_s \sqrt{s}L_n^{-d}.
	\end{equation*}
	Then applying Lemma \ref{lemma_hd_inf} and \ref{lemma:bn_concentration}, we claim that
	\begin{equation*}
		\max_{j\in S^c}\twonorm{(\bD_n\tgamma - \bb_n)_{(j)}}\lesssim  \sqrt{\frac{L_n\log p}{n}} + a_nL_ns\sqrt{\frac{\log p}{n}} + a_n + \sqrt{s} L_n^{-d}
	\end{equation*}
	holds with probability at least $1 - 10L_np^{-\vartheta}$.
\end{proof}

\section{Deferred proofs of Section \ref{appen:pre} and \ref{sec::main_proof}}\label{sec:main_lemma_proofs}
\subsection{Proof of Lemma \ref{lemma_Bnorm_bound}} \label{proof::lemma_Bnorm_bound}
\begin{proof}
		Let $\mathbb{S}^{L_n-1}$ be the unit sphere in $\mathbb{R}^{L_n}$, we denote the $\frac{1}{8}$-covering of $\mathbb{S}^{L_n-1}$ by $\{\bnu_1,...,\bnu_K\}$ with $K\leq 17^{L_n}$. Let $\bQ = \sum_{i=1}^n\bB_i\bB_i^{\top}/n - \E[\bB\bB^{\top}]$, then we have
		\begin{equation*}
			\twonorm{\bQ} = \sup_{\bnu \in \mathbb{S}^{L_n-1}}|\bnu^{\top}\bQ\bnu|.
		\end{equation*}
		Based on the definition of covering set, for any $\bnu \in \mathbb{S}^{L_n-1}$, there exists some $1\leq k\leq K$ such that $\twonorm{\bnu -\bnu_k} \leq 1/8$. It follows that
		\begin{equation*}
			\begin{aligned}
				|\bnu^{\top}\bQ\bnu| &\leq |\bnu_k^{\top}\bQ\bnu_k| + 2|\bnu_k^{\top}\bQ(\bnu_k -\bnu)| + |(\bnu_k -\bnu)^{\top}\bQ(\bnu_k -\bnu)|\\
				&\leq |\bnu_k^{\top}\bQ\bnu_k| + \frac{1}{4}\twonorm{\bQ} + \frac{1}{64}\twonorm{\bQ}\\
				&\leq |\bnu_k^{\top}\bQ\bnu_k| + \frac{1}{2}\twonorm{\bQ}.
			\end{aligned}
		\end{equation*}
		Thus we have
		\begin{equation}
			\twonorm{\bQ} \leq 2\max_{1\leq k\leq K}|\bnu_k^{\top}\bQ\bnu_k| = 2\max_{1\leq k\leq K}\left|\frac{1}{n}\sum_{i=1}^n(\bnu_k^{\top}\bB_i)^2 - \E[(\bnu_k^{\top}\bB)^2]\right|.
		\end{equation}
		Since $\twonorm{\bB_i^{*}}^2 = \sum_{k=1}^{L_n}(B_{k}^{*}(U_i))^2 \leq \sum_{k=1}^{L_n}B_{k}^{*}(U_i) = 1$, together with \eqref{int_bound}, we have
		\begin{equation*}
			\begin{aligned}
				\E\LRm{(\bnu_k^{\top}\bB_i^{*})^4\exp\{\eta(\bnu_k^{\top}\bB_i^{*})^2\}} &\leq e^{\eta}\E\LRm{(\bnu_k^{\top}\bB_i^{*})^2}\lesssim L_n^{-1} \twonorm{\bnu_k}^2 = L_n^{-1},
			\end{aligned}
		\end{equation*}
		together with Lemma \ref{bern_ineq} we claim that
		\begin{equation*}
			\Prob\LRs{\max_{1\leq k\leq K}\left|\frac{1}{n}\sum_{i=1}^n(\bnu_k^{\top}\bB_i)^2 - \E[(\bnu_k^{\top}\bB)^2]\right|\geq CL_n\sqrt{\frac{L_n\log n}{nL_n}}}\leq K n^{-\vartheta L_n^3}\leq n^{-\vartheta L_n}.
		\end{equation*}
		Then the conclusion follows immediately.
	\end{proof}
\subsection{Proof of Lemma \ref{lemma_eigen}}\label{proof::lemma_eigen}
	\begin{proof}[Proof of Lemma \ref{lemma_eigen}]
	Let $\bC(u) = \E[(\bX-\bmu(u))(\bX-\bmu(u))^{\top}]$, and then it holds
	\begin{align*}
		\bC(u) &= \frac{1}{2}\LRs{\E\LRm{(\bX-\bmu(u))(\bX-\bmu(u))^{\top}| Y=1} + \E\LRm{(\bX-\bmu(u))(\bX-\bmu(u))^{\top}| Y=0}}\\
		&= \bSigma(u)+\frac{1}{4}(\bmu_{1}(u)-\bmu_{2}(u))(\bmu_{1}(u)-\bmu_{2}(u))^{\top}.
	\end{align*}
	From conditions \cond{1} and \cond{4} we have
	\begin{equation}
		\lambda_0\leq \lambda_{\min}(\bC(u))\leq \lambda_{\max}(\bC(u))\leq \lambda_1+\frac{1}{4}\delta_p,
	\end{equation}
	holds for any $u \in [0,1]$. In addition, we notice that
	\begin{align*}
	    \E[\tB\tB^{\top}] &= \E\LRm{\LRs{(\bX-\bmu(U))(\bX-\bmu(U))^{\top}}\otimes \LRs{\bB\bB^{\top}}}\\
	    &=\E\LRm{\bC(U)\otimes \LRs{\bB\bB^{\top}}}.
	\end{align*}
	Then for any $\bEta = \trans{(\bEta_{(1)}^{\top},\cdots,\bEta_{(p)}^{\top})} \in \mathbb{R}^{L_np}$ with $\twonorm{\bEta} = 1$, we have
	\begin{equation*}
		\begin{aligned}
			\trans{\bEta} \E[\tB\tB^{\top}] \bEta &= \E\left(\sum_{j=1}^{p}(X_j-\mu_{j}(U))\trans{\bB}\bEta_j\right)^2 \\
			&= \E\LRm{\LRs{\bB^{\top}\bEta_{(1)},\cdots,\bB^{\top}\bEta_{(p)}} \bC(U) \LRs{\bB^{\top}\bEta_{(1)},\cdots,\bB^{\top}\bEta_{(p)}}^{\top}}\\
			&\geq \inf_{u \in [0,1]} \lambda_{\min}\LRs{\bC(u)}\E\left[\sum_{j=1}^p \trans{\bEta_{(j)}}\bB\trans{\bB}\bEta_{(j)}\right]\\
			&\geq \lambda_0\lambda_{\min}(\E[\bB\bB^{\top}]).
		\end{aligned}
	\end{equation*}
	Similarly, we have
	\begin{equation*}
		\trans{\bEta} \E[\tB\tB^{\top}] \bEta\leq \LRs{\lambda_1 + \frac{\delta_p}{4}}\lambda_{\max}(\E[\bB\bB^{\top}]).
	\end{equation*}
	Then the result follows from $\lambda_{\min}(\E[\bB\bB^{\top}])= O(1)$ and $\lambda_{\max}(\E[\bB\bB^{\top}])= O(1)$ (see Section \ref{appen:Bspline}).
\end{proof}
	
\subsection{Proof of Lemma \ref{lemma:tB_matrix_concentration}}\label{proof:lemma:tB_matrix_concentration}
To prove Lemma \ref{lemma:tB_matrix_concentration}, we impose the following five lemmas on the concentration inequalities of random matrices. The proofs can be found in Appendix \ref{proof:lemma_matrix_bound_1} - \ref{proof:lemma_matrix_bound_5}.

\begin{lemma}\label{lemma_matrix_bound_1}
	Let $\bZ_i = \bZ(U_i) = (\bX_i-\bmu_{k}(U_i))\otimes \bB_i$, then under condition \cond{1}-\cond{4}, we have for any $\vartheta>0$ and $k=1,2$
	\begin{equation*}
		\LRtwonorm{\frac{1}{n}\sum_{i\in \mathcal{I}_k}\left\{\bZ_i\bZ_i^{\top} - \E\left[\bZ_i\bZ_i^{\top}|Y_i=k\right]\right\}} \lesssim L_n\sqrt{\frac{p\log n}{n}},
	\end{equation*}
	and
	\begin{align*}
	    \LRtwonorm{\frac{1}{n}\sum_{i\in \mathcal{I}_k}\left\{\bZ_i\bZ_i^{\top} - \E\left[\bZ_i\bZ_i^{\top}|Y_i=k\right]\right\}\tgamma} \lesssim \sqrt{\frac{pL_n\log n}{n}},
	\end{align*}
	hold with probability at least $1-n^{-\vartheta p L_n}$.
\end{lemma}
\begin{lemma}\label{lemma_matrix_bound_2}
	Under conditions \cond{1}-\cond{4}, for $k=1,2$ and any $\vartheta>0$, we have
	\begin{align*}
	    \LRtwonorm{\frac{1}{n}\sum_{i\in \mathcal{I}_k}\left[(\bX_i-\bmu_{k}(U_i))(\bmu_1(U_i) - \bmu_2(U_i))^{\top}\right]\otimes \left(\bB_i\bB_i^{\top}\right)} \lesssim L_n\sqrt{\frac{p\log n}{n}},\\
	    \LRtwonorm{\frac{1}{n}\sum_{i\in \mathcal{I}_k}\LRl{\left[(\bX_i-\bmu_{k}(U_i))(\bmu_1(U_i) - \bmu_2(U_i))^{\top}\right]\otimes \left(\bB_i\bB_i^{\top}\right)}\tgamma} \lesssim \sqrt{\frac{pL_n\log n}{n}},
	\end{align*}
	hold with probability at least $1-n^{-\vartheta L_np}$.
\end{lemma}

\begin{lemma}\label{lemma_matrix_bound_3}
	Under conditions \cond{1}-\cond{4}, let $\bA_i = [(\bmu_1(U_i)-\bmu_2(U_i))(\bmu_1(U_i)-\bmu_2(U_i))^{\top}]\otimes (\bB_i\bB_i^{\top})$, then for any $\vartheta>0$,
	\begin{align*}
	    \LRtwonorm{\frac{1}{n}\sum_{i=1}^n \bA_i - \E[\bA_i]}\lesssim L_n\sqrt{\frac{p\log n}{n}},\\
	    \LRtwonorm{\frac{1}{n}\sum_{i=1}^n \LRs{\bA_i - \E[\bA_i]}\tgamma} \lesssim \sqrt{\frac{p L_n\log n}{n}}
	\end{align*}
	hold with probability at least $1-n^{-\vartheta L_np}$.
\end{lemma}
\begin{lemma}\label{lemma_matrix_bound_4}
	Under conditions \cond{1}-\cond{4}, then for $k=1,2$ and any $\vartheta>0$,
	\begin{align*}
		&\LRtwonorm{\frac{1}{n}\sum_{i \in \mathcal{I}_k}[(\bX_i - \bmu_k(U_i))(\hmu(U_i) - \bmu(U_i))^{\top}]\otimes (\bB_i\bB_i^{\top})}\\
		&\qquad \lesssim p L_n^{3/2}\sqrt{\frac{\log n}{n}} \LRs{\sqrt{\frac{L_n\log n}{n}}+L_n^{-d}},\\
	    &\LRtwonorm{\frac{1}{n}\sum_{i \in \mathcal{I}_k}\LRl{[(\bX_i - \bmu_k(U_i))(\hmu(U_i) - \bmu(U_i))^{\top}]\otimes (\bB_i\bB_i^{\top})}\tgamma}\\
		&\qquad \lesssim p L_n\sqrt{\frac{\log n}{n}} \LRs{\sqrt{\frac{L_n\log n}{n}}+L_n^{-d}},
	\end{align*}
	hold with probability at least $1- n^{-\vartheta pL_n}- p L_n n^{-\vartheta}$.
\end{lemma}
\begin{lemma}\label{lemma_matrix_bound_5}
	Under conditions \cond{1}-\cond{4}, let $\bG(U_i) = [(\bmu_{1}(U_i)-\bmu_{2}(U_i))(\hmu(U_i)- \bmu(U_i))^{\top}]\otimes (\bB_i\bB_i^{\top})$,
	we have for any $\vartheta>0$
	\begin{align*}
	\LRtwonorm{\frac{1}{n}\sum_{i\in \mathcal{I}_1}\bG(U_i)-\frac{1}{n}\sum_{i\in \mathcal{I}_2}\bG(U_i)} \lesssim p L_n^{3/2}\sqrt{\frac{\log n}{n}} \LRs{\sqrt{\frac{L_n\log n}{n}}+L_n^{-d}},\\
	    \LRtwonorm{\frac{1}{n}\sum_{i\in \mathcal{I}_1}\bG(U_i)\tgamma -\frac{1}{n}\sum_{i\in \mathcal{I}_2}\bG(U_i)\tgamma } \lesssim p L_n\sqrt{\frac{\log n}{n}} \LRs{\sqrt{\frac{L_n\log n}{n}}+L_n^{-d}},
	\end{align*}
	hold with probability at least $1- n^{-\vartheta pL_n} - p L_n n^{-\vartheta}$.
\end{lemma}
\begin{proof}[Proof of Lemma \ref{lemma:tB_matrix_concentration}]
    Note that we can rewrite $\E[\tB\tB^{\top}]$ and $\sum_{i=1}^{2n}\tB_i\tB_i^{\top}/2n$ as
	\begin{equation*}
		\begin{aligned}
			\E\LRm{\tB\tB^{\top}} &= \underbrace{\frac{1}{2}\E\left\{\left[\left(\bX-\bmu_{1}(U)\right)\left(\bX-\bmu_{1}(U)\right)^{\top}\right]\otimes (\bB\bB^{\top})\big| Y=1\right\}}_{\bI_1^{*}}\\
			&+\underbrace{\frac{1}{2}\E\left\{\left[\left(\bX-\bmu_{2}(U)\right)\left(\bX-\bmu_{2}(U)\right)^{\top}\right]\otimes (\bB\bB^{\top})\big| Y=0\right\}}_{ \bI_2^{*}}\\
			&+\underbrace{\frac{1}{4}\E\left\{\left[(\bmu_{1}(U)-\bmu_{2}(U))(\bmu_{1}(U)-\bmu_{2}(U))^{\top}\right]\otimes (\bB\bB^{\top})\right\}}_{\bI_3^{*}}
		\end{aligned}
	\end{equation*}
and
\begin{equation*}
	\begin{aligned}
		\frac{1}{2n}\sum_{i=1}^{2n}\tB_i\tB_i^{\top} &= \underbrace{\frac{1}{2n}\sum_{i\in \mathcal{I}_1}\left\{\left[(\bX_i-\hmu(U_i))(\bX_i-\hmu(U_i))^{\top}\right]\otimes (\bB_i\bB_i^{\top}) \right\}}_{\bI^1} \\
		&+ \underbrace{\frac{1}{2n}\sum_{i\in \mathcal{I}_2}\left\{\left[(\bX_i-\hmu(U_i))(\bX_i-\hmu(U_i))^{\top}\right]\otimes (\bB_i\bB_i^{\top})\right\}}_{\bI^2}.
	\end{aligned}
\end{equation*}
To upper bound $\twonorm{\sum_{i=1}^{2n}\tB_i\tB_i^{\top}/2n - \E[\tB\tB^{\top}]}$, we begin with the following decompositions for $\bI_1$ and $\bI_2$,
\begin{equation*}
	\begin{aligned}
		&\bI^1=\underbrace{\frac{1}{2n}\sum_{i\in \mathcal{I}_1}\left[(\bX_i-\bmu(U_i))(\bX_i-\bmu(U_i))^{\top}\right]\otimes (\bB_i\bB_i^{\top})}_{\bI^1_1}\\
		+&\underbrace{\frac{1}{2n}\sum_{i\in \mathcal{I}_1}\left\{\left[(\bX_i-\bmu(U_i))(\bmu(U_i)-\hmu(U_i))^{\top}\right]\otimes (\bB_i\bB_i^{\top})\right\}}_{\bI^1_2}\\
		+&\underbrace{\frac{1}{2n}\sum_{i\in \mathcal{I}_1}\left\{\left[(\bmu(U_i)-\hmu(U_i))(\bX_i-\bmu(U_i))^{\top}\right]\otimes (\bB_i\bB_i^{\top})\right\}}_{\bI^1_3}\\
		+&\underbrace{\frac{1}{2n}\sum_{i\in \mathcal{I}_1}\left[(\bmu(U_i)-\hmu(U_i))(\bmu(U_i)-\hmu(U_i))^{\top}\right]\otimes (\bB_i\bB_i^{\top})}_{\bI^1_4}
	\end{aligned}
\end{equation*}
and
\begin{equation*}
	\begin{aligned}
		\bI^2&=\underbrace{\frac{1}{2n}\sum_{i\in \mathcal{I}_2}\left[(\bX_i-\bmu(U_i))(\bX_i-\bmu(U_i))^{\top}\right]\otimes (\bB_i\bB_i^{\top})}_{\bI^2_1}\\
		+&\underbrace{\frac{1}{2n}\sum_{i\in \mathcal{I}_2}\left\{\left[(\bX_i-\bmu(U_i))(\bmu(U_i)-\hmu(U_i))^{\top}\right]\otimes (\bB_i\bB_i^{\top})\right\}}_{\bI^2_2}\\
		+&\underbrace{\frac{1}{2n}\sum_{i\in \mathcal{I}_2}\left\{\left[(\bmu(U_i)-\hmu(U_i))(\bX_i-\bmu(U_i))^{\top}\right]\otimes (\bB_i\bB_i^{\top})\right\}}_{\bI^2_3}\\
		+&\underbrace{\frac{1}{2n}\sum_{i\in \mathcal{I}_2}\left[(\bmu(U_i)-\hmu(U_i))(\bmu(U_i)-\hmu(U_i))^{\top}\right]\otimes (\bB_i\bB_i^{\top})}_{\bI^2_4}
	\end{aligned}
\end{equation*}
\paragraph{Step 1.1.  upper bounding $\twonorm{\bI_1^1-\bI_1^{*} + \bI_1^2-\bI_2^{*}-\bI_3^{*}}$.}First, we decompose $\bI_1^1$ as
\begin{equation*}
	\begin{aligned}
		&\bI_1^1 = \underbrace{\frac{1}{2n}\sum_{i\in \mathcal{I}_1}\left[(\bX_i-\bmu_1(U_i))(\bX_i-\bmu_1(U_i))^{\top}\right]\otimes (\bB_i\bB_i^{\top})}_{\bI_{11}^1}\\
		&+ \underbrace{\frac{1}{4n}\sum_{i\in \mathcal{I}_1}\left[(\bX_i-\bmu_1(U_i))(\bmu_1(U_i)-\bmu_2(U_i))^{\top}\right]\otimes (\bB_i\bB_i^{\top})}_{\bI_{12}^1}\\
		&+\underbrace{\frac{1}{4n}\sum_{i\in \mathcal{I}_1}\left[(\bmu_1(U_i)-\bmu_2(U_i))(\bX_i-\bmu_1(U_i))^{\top}\right]\otimes (\bB_i\bB_i^{\top})}_{\bI_{13}^1}\\
		&+\underbrace{\frac{1}{8n}\sum_{i\in \mathcal{I}_1}\left[(\bmu_1(U_i)-\bmu_2(U_i))(\bmu_{1}(U_i)-\bmu_2(U_i))^{\top}\right]\otimes (\bB_i\bB_i^{\top})}_{\bI_{14}^1}
	\end{aligned}
\end{equation*}
By Lemma \ref{lemma_matrix_bound_1}, we have
\begin{equation*}
	\Prob\left(\LRtwonorm{\bI_{11}^1 - \bI_1^{*}}\geq CL_n\sqrt{\frac{p\log n}{n}}\right)\leq n^{-\vartheta L_np}.
\end{equation*}
By Lemma \ref{lemma_matrix_bound_2}, we have
\begin{equation*}
	\Prob\left(\LRtwonorm{\bI_{12}^1} \lesssim L_n\sqrt{\frac{ p\log n}{n}}\right)\geq 1 - n^{-\vartheta L_np},
\end{equation*}
and
\begin{equation*}
	\Prob\left(\LRtwonorm{\bI_{13}^1} \lesssim CL_n\sqrt{\frac{ p\log n}{n}}\right)\geq 1- n^{-\vartheta L_np}.
\end{equation*}
By Lemma \ref{lemma_matrix_bound_3}, we have
\begin{equation*}
	\Prob\left(\twonorm{\bI_{14}^1 - \frac{1}{2}\bI_3^{*}}\lesssim L_n\sqrt{\frac{p\log n}{n}}\right)\geq 1 - pL_n n^{-\vartheta}.
\end{equation*}
Combining the results displayed above, it follows that
\begin{equation}\label{step_1_41}
	\Prob\left(\twonorm{\bI_1^1-\bI_1^{*}-\frac{1}{2}\bI_3^{*}}\lesssim L_n\sqrt{\frac{p\log n}{n}}\right) \geq 1 - 3n^{-\vartheta L_np} - pL_n n^{-\vartheta}.
\end{equation}
Similarly, we also have
\begin{equation}\label{step_1_42}
	\Prob\left(\twonorm{\bI_1^2-\bI_2^{*}-\frac{1}{2}\bI_3^{*}}\lesssim L_n\sqrt{\frac{p\log n}{n}}\right) \geq 1 - 3n^{-\vartheta L_np} - pL_n n^{-\vartheta}.
\end{equation}

\paragraph{Step 1.2. upper bounding $\twonorm{\bI_2^{1}+\bI_{2}^{2}}$, $\twonorm{\bI_3^{1}+\bI_3^{2}}$ and $\twonorm{\bI_4^{1}}+\twonorm{\bI_4^{2}}$.} Note that
\begin{equation*}
	\begin{aligned}
		\twonorm{\bI_2^1 +\bI_{2}^{2}}\leq & \LRtwonorm{\frac{1}{2n}\sum_{i\in \mathcal{I}_1}\left[(\bX_i-\bmu_1(U_i))(\hmu(U_i) - \bmu(U_i))^{\top}\right]\otimes(\bB_i\bB_i^{\top})}\\
		  + & \LRtwonorm{\frac{1}{2n}\sum_{i\in \mathcal{I}_2}\left[(\bX_i-\bmu_2(U_i))(\hmu(U_i) - \bmu(U_i))^{\top}\right]\otimes(\bB_i\bB_i^{\top})}\\
		+ & \LRtwonorm{\frac{1}{4n}\sum_{i\in \mathcal{I}_1}\bG(U_i)  - \frac{1}{4n}\sum_{i\in \mathcal{I}_2}\bG(U_i)},
	\end{aligned}
\end{equation*}
where $\bG(U_i) = [(\bmu_{1}(U_i)-\bmu_{2}(U_i))(\hmu(U_i)- \bmu(U_i))^{\top}]\otimes (\bB_i\bB_i^{\top})$. By invoking Lemma \ref{lemma_matrix_bound_5} and \ref{lemma_matrix_bound_3}, we have
\begin{equation}\label{step_1_5}
	\Prob\left(\twonorm{\bI_2^1 +\bI_{2}^{2}} \lesssim pL_n^{3/2}a_n\sqrt{\frac{\log n}{n}}\right)\geq 1 - n^{-\vartheta L_np},
\end{equation}
where $a_n = \sqrt{L_n\log n/n} + L_n^{-d}$. Similarly, we can obtain
\begin{equation}\label{step_1_6}
	\Prob\left(\twonorm{\bI_3^1 + \bI_3^2} \lesssim pL_n^{3/2}a_n \sqrt{\frac{\log n}{n}}\right)\geq 1 - n^{-\vartheta L_np}.
\end{equation}
In addition, by Proposition \ref{pro_mean_bound} we have
\begin{equation}\label{step_1_7}
	\begin{aligned}
	    \twonorm{\bI_4^{1}} + \twonorm{\bI_4^{2}}&\leq 2\max_{1\leq i \leq n}\twonorm{\hmu(U_i) - \bmu(U_i)}^2\twonorm{\bB_i}^2\\
	    &\leq 2pL_na_n^2
	\end{aligned}
\end{equation}
Combining (\ref{step_1_41})-(\ref{step_1_7}), we have
\begin{equation}\label{step_1}
	\LRtwonorm{\frac{1}{2n}\sum_{i=1}^{2n}\tB_i\tB_i^{\top} - \E[\tB\tB^{\top}]} \lesssim L_n\sqrt{\frac{p\log n}{n}}+pL_n^{3/2}a_n \sqrt{\frac{\log n}{n}}
\end{equation}
holds with probability at least $1 - n^{-\vartheta L_np} - pL_n n^{-\vartheta}$.
\end{proof}

\subsection{Proof of Lemma \ref{lemma:tB_tgamma_concentration}}\label{proof:lemma:tB_tgamma_concentration}
\begin{proof}
By replacing the bounds for the operator norm of matrices with those for $\ell_2$-norm of matrix-vector-products in Section \ref{proof:lemma:tB_matrix_concentration}, we can finish the proof
Due to the fact that $\tB_i^{\top} \tgamma_{(j)} = \widetilde{\theta}_j(U_i)$ is bounded, we can drop a $\sqrt{L_n}$ factor for matrix-vector-product bounds in Lemma \ref{lemma_matrix_bound_1}- \ref{lemma_matrix_bound_3}.
\end{proof}

\subsection{Proof of Lemma \ref{lemma:bn_concentration}}
\begin{proof}
Note that for any $j=1,2,...,p$ and $k=1,2,...,L_n$, we find that
\begin{equation*}
	\begin{aligned}
		&\frac{1}{2n}\sum_{i=1}^{2n} B_k(U_i)(X_{ij}-\widehat{\mu}_j(U_i))Z_i\\
		= &\frac{1}{4n}\sum_{i\in \mathcal{I}_1} B_k(U_i)(X_{ij}-\widehat{\mu}_j(U_i)) - \frac{1}{4n}\sum_{i\in \mathcal{I}_2} B_k(U_i)(X_{ij}-\widehat{\mu}_j(U_i))\\
		= & \frac{1}{4}\underbrace{\frac{1}{2n}\sum_{i = 1}^{2n} B_k(U_i)\widehat{\mu}_{1j}(U_i)}_{I_1} - \frac{1}{4}\underbrace{\frac{1}{2n}\sum_{i = 1}^{2n} B_k(U_i)\widehat{\mu}_{2j}(U_i)}_{I_2},
	\end{aligned}
\end{equation*}
where we used $\widehat{\mu}_{1j} = \bB_i^{\top}\halpha_{1j}$, $\widehat{\mu}_{2j} = \bB_i^{\top}\halpha_{2j}$ and the optimal condition for $\halpha_{1j}$ and $\halpha_{2j}$ such that
\begin{align*}
    \sum_{i\in \mathcal{I}_1}\bB_i(X_{ij} - \bB_i^{\top}\halpha_{1j}) = \boldsymbol{0},\quad \sum_{i\in \mathcal{I}_2}\bB_i(X_{ij} - \bB_i^{\top}\halpha_{2j}) = \boldsymbol{0}.
\end{align*}
Moreover, note that
\begin{equation*}
\begin{aligned}
		&\E[B_k(U)(X_j - \mu_{j}(U))Z]\\
		= & \frac{1}{4}\E[B_k(U)(X_j - \mu_{j}(U))|Y=1] - \frac{1}{4}\E[B_k(U)(X_j - \mu_{j}(U))|Y=0]\\
		= & \frac{1}{4}\underbrace{\E[B_k(U)\mu_{1j}(U)]}_{I_1^{*}} - \frac{1}{4}\underbrace{\E[B_k(U)\mu_{2j}(U)]}_{I_2^{*}}
\end{aligned}
\end{equation*}
The remaining detail is to upper bound $|I_1 - I_1^{*}|$. According to Proposition \ref{pro_mean_bound} and $B_k = \sqrt{L_n}B_k^{*}$, we have
\begin{equation*}
	\begin{aligned}
		|I_1 - I_1^{*}| &\leq \sqrt{L_n}\left|\frac{1}{2n}\sum_{i = 1}^{2n} B_k^{*}(U_i)[\widehat{\mu}_{1j}(U_i) - \mu_{1j}(U_i)]\right|\\
		&\qquad + \sqrt{L_n}\left|\frac{1}{2n}\sum_{i = 1}^{2n} B_k^{*}(U_i)\mu_{1j}(U_i) - \E[B_k^{*}(U)\mu_{1j}(U)]\right|\\
		&\leq a_n\sqrt{L_n}\left|\frac{1}{2n}\sum_{i = 1}^{2n} |B_k^{*}(U_i)| -\E[|B_k^{*}(U_i)|]\right| + a_n\sqrt{L_n}\E[|B_k^{*}(U)|]\\
		&\qquad+ \sqrt{L_n}\left|\frac{1}{2n}\sum_{i = 1}^{2n} B_k^{*}(U_i)\mu_{1j}(U_i) - \E[B_k^{*}(U)\mu_{1j}(U)]\right|.
	\end{aligned}
\end{equation*}
Using Lemma \ref{bern_ineq}, we can verify
\begin{equation*}
	\Prob\left(\left|\frac{1}{2n}\sum_{i = 1}^{2n} B_k^{*}(U_i)\mu_{1j}(U_i) - \E[B_k^{*}(U)\mu_{1j}(U)]\right| \lesssim \sqrt{\frac{\log n}{L_nn}}\right)\geq 1 - n^{-\vartheta}
\end{equation*}
and
\begin{equation*}
	\Prob\left(\left|\frac{1}{2n}\sum_{i = 1}^{2n} |B_k^{*}(U_i)| -\E[|B_k^{*}(U_i)|]\right| \lesssim \sqrt{\frac{\log n}{L_nn}}\right)\geq 1 - n^{-\vartheta}.
\end{equation*}
In addition, recall the fact that $\E|B_k^{*}(U)|\leq M_2L_n^{-1}$, which yields that
\begin{equation*}
	\Prob\left(|I_1 - I_1^{*}| \lesssim \sqrt{\frac{\log n}{n}} + L_n^{-1/2}a_n\right)\geq 1- n^{-\vartheta}.
\end{equation*}
Thus we are guaranteed that
\begin{equation}\label{step_2}
	\Prob\left(\LRabs{\frac{1}{2n}\sum_{i=1}^{2n}\tB_iZ_i-\E[\tB Z]} \lesssim \sqrt{\frac{\log n}{n}}+L_n^{-1/2}a_n\right)\geq 1 - 4pL_n n^{-\vartheta}.
\end{equation}
\end{proof}

\subsection{Proof of Lemma \ref{g_lasso}}\label{proof::g_lasso}
\begin{proof}
	By the optimality of $\hgamma$, we have
	\begin{equation}\label{basic_ineq}
	\begin{aligned}
	\frac{1}{2}\left(\hgamma - \cgamma\right)^{\top}\mathbf{A}\left(\hgamma - \cgamma\right) + \lambda_n \sum_{j=1}^p \twonorm{\hgamma_{(j)}}\leq \left(\mathbf{A}\cgamma - \bb\right)^{\top}\left(\cgamma - \hgamma\right) + \lambda_n \sum_{j=1}^p \twonorm{\cgamma_{(j)}}.
	\end{aligned}
	\end{equation}
	Using the condition (\ref{em_process}) and dropping the first non-negative term in the left hand side of (\ref{basic_ineq}), we are guaranteed that
	\begin{equation*}
	\begin{aligned}
	\lambda_n \sum_{j=1}^p \twonorm{\hgamma_{(j)}}&\leq \sum_{j=1}^p\twonorm{(\mathbf{A}\cgamma - \bb)_{(j)}}\twonorm{(\cgamma - \hgamma)_{(j)}} + \lambda_n \sum_{j=1}^p \twonorm{\cgamma_{(j)}}\\
	&\leq \frac{\lambda_n}{2}\sum_{j=1}^p\twonorm{(\cgamma - \hgamma)_{(j)}} +  \lambda_n \sum_{j=1}^p \twonorm{\cgamma_{(j)}}\\
	& = \frac{\lambda_n}{2}\sum_{j\in S}\twonorm{(\cgamma - \hgamma)_{(j)}} + \frac{\lambda_n}{2}\sum_{j\in S^c}\twonorm{(\cgamma - \hgamma)_{(j)}} + \lambda_n \sum_{j\in S} \twonorm{\cgamma_{(j)}}.
	\end{aligned}
	\end{equation*}
	It follows from the assumption $\cgamma_{(j)} = \boldsymbol{0}$ for $j \in S^c$ that
	\begin{align}
	    \frac{1}{2}\sum_{j\in S^c}\twonorm{(\cgamma - \hgamma)_{(j)}}&\leq \frac{1}{2}\sum_{j\in S}\twonorm{(\cgamma - \hgamma)_{(j)}} + \sum_{j\in S} \twonorm{\cgamma_{(j)}} - \sum_{j\in S} \twonorm{\hgamma_{(j)}}\nonumber\\
	    &\leq \frac{1}{2}\sum_{j\in S}\twonorm{(\cgamma - \hgamma)_{(j)}} + \sum_{j\in S}\twonorm{(\cgamma - \hgamma)_{(j)}}\nonumber\\
	    &\leq \frac{3}{2}\sum_{j\in S}\twonorm{(\cgamma - \hgamma)_{(j)}},\label{eq:cone_relation}
	\end{align}
	where we used the fact $\twonorm{\cgamma_{(j)}} - \twonorm{\hgamma_{(j)}} \leq \twonorm{(\cgamma - \hgamma)_{(j)}}$.
	From (\ref{basic_ineq}), we can also obtain that
	\begin{align}
	    \frac{1}{2}\left(\hgamma - \cgamma\right)^{\top}\mathbf{A}\left(\hgamma - \cgamma\right) &\leq \frac{\lambda_n}{2}\sum_{j=1}^p\twonorm{(\cgamma - \hgamma)_{(j)}} + \lambda_n \sum_{j=1}^p \twonorm{\cgamma_{(j)}} - \lambda_n \sum_{j=1}^p \twonorm{\hgamma_{(j)}}\nonumber\\
	    &\leq \frac{3\lambda_n}{2}\sum_{j=1}^p\twonorm{(\cgamma - \hgamma)_{(j)}}.
	    \label{basic_ineq2}
	\end{align}
	By the restrictive eigenvalue condition of $\mathbf{A}$, we know
	\begin{align*}
	    \frac{1}{2}\left(\hgamma - \cgamma\right)^{\top}\mathbf{A}\left(\hgamma - \cgamma\right) \geq \frac{\zeta}{2}\twonorm{\cgamma - \hgamma}^2.
	\end{align*}
	Together with (\ref{basic_ineq2}) and \eqref{eq:cone_relation}, we further have
	\begin{align*}
	\zeta\twonorm{\cgamma - \hgamma}^2 &\leq 3\lambda_n\sum_{j=1}^p\twonorm{(\cgamma - \hgamma)_{(j)}}\\
	&= 3\lambda_n \LRs{\sum_{j\in S} \twonorm{(\cgamma - \hgamma)_{(j)}} + \sum_{j\in S^c} \twonorm{(\cgamma - \hgamma)_{(j)}}}\\
	&\leq 12 \lambda_n \sum_{j\in S} \twonorm{(\cgamma - \hgamma)_{(j)}}\\
	&\leq 12\lambda_n\sqrt{s}\twonorm{(\cgamma - \hgamma)_{(S)}}\\
	&\leq 12\lambda_n\sqrt{s} \twonorm{\cgamma - \hgamma},
	\end{align*}
	which yields the first conclusion in Lemma \ref{g_lasso}. In fact, we also used the following relation
	\begin{align*}
	    \LRs{\sum_{j\in S} \twonorm{(\cgamma - \hgamma)_{(j)}}}^2 \leq s \sum_{j\in S} \twonorm{(\cgamma - \hgamma)_{(j)}}^2 = s\twonorm{(\cgamma - \hgamma)_{(S)}}^2.
	\end{align*}
	And the second conclusion holds since
	\begin{equation*}
	\begin{aligned}
	\onenorm{\cgamma - \hgamma} &= \sum_{j\in S}\onenorm{(\cgamma - \hgamma)_{(j)}} + \sum_{j\in S^c}\onenorm{(\cgamma - \hgamma)_{(j)}}\\
	&\leq \sqrt{L_n}\sum_{j\in S}\twonorm{(\cgamma - \hgamma)_{(j)}} + \sqrt{L_n}\sum_{j\in S^c}\twonorm{(\cgamma - \hgamma)_{(j)}}\\
	&\leq 4\sqrt{L_n}\sum_{j\in S}\twonorm{(\cgamma - \hgamma)_{(j)}}\\
	&\leq 4\sqrt{sL_n}\twonorm{\cgamma - \hgamma}.
	\end{aligned}
	\end{equation*}
\end{proof}
\subsection{Proof of Lemma \ref{lemma_hd_inf}}\label{proof::lemma_hd_inf}
\begin{proof}[Proof of Lemma \ref{lemma_hd_inf}]
    Recall that $\tgamma_{(S^c)} = \boldsymbol{0}$, thus
    \begin{align*}
        \max_{1\leq j\leq p}\twonorm{(\bD_n\tgamma - \bD\tgamma)_{(j)}} &= \max_{1\leq j\leq p}\twonorm{(\bD_n - \bD)_{(j,S)}\tgamma_{(S)}}.
    \end{align*}
    Then it suffices to show that
    \begin{equation}\label{unit_form}
        \max_{1\leq j\leq p}\twonorm{(\bD_n - \bD)_{(j,S)}\tgamma_{(S)}} \lesssim \sqrt{\frac{\log p}{n}}+\frac{L_n\log p}{n}+L_n^{-2d},
    \end{equation}
    with high probability. We use the same decomposition for $\bD_n - \bD$ in the Section \ref{proof:lemma:tB_matrix_concentration} and only prove the counterpart to $\{i:Y_i =1\}$. Correspondingly, the expectation $\E[\cdot]$ means conditional expectation $\E[\cdot| Y_i = 1]$. We split the proof into the following two main steps.
	\paragraph{Step 1. upper bounding $\max_{1\leq j \leq p}\twonorm{(\bI_1^1-\bI_1^{*} + \bI_1^2-\bI_2^{*} -\bI_3^{*})_{(j,S)}\tgamma_{(S)}}$.} Recall
	\begin{align*}
	   \LRtwonorm{(\bI_1^1 - \bI_1^{*} - \frac{1}{2}\bI_3^{*})_{(j, S)}\tgamma_{(S)}} &\leq \LRtwonorm{(\bI_{11}^1 - \bI_1^*)_{(j,S)} \tgamma_{(S)}}\\
	   &+\LRtwonorm{(\bI_{12}^1)_{(j,S)} \tgamma_{(S)}} + \LRtwonorm{(\bI_{13}^1)_{(j,S)} \tgamma_{(S)}}\\
	   &+\LRtwonorm{(\bI_{14}^1 - \frac{1}{2}\bI_3^{*})_{(j,S)}\tgamma_{(S)}}.
	\end{align*}
	Notice that
	\begin{equation*}
	    \begin{aligned}
	            \LRtwonorm{(\bI_{11}^1)_{(j,S)}\tgamma_{(S)}} =\sqrt{L_n} \LRtwonorm{\frac{1}{n}\sum_{i=1}^n(X_{ij}-\mu_{1j}(U_i))\tbtheta_{S}(U_i)^{\top}\LRs{\bX - \bmu_1(U_i)}_{S} \bB_i^{*}}.
	    \end{aligned}
	\end{equation*}
	Given $Y_i = 1$ and $U_i$, $\tbtheta_{S}(U_i)^{\top}\LRs{\bX - \bmu_1(U_i)}_{S}$ is a normal random variable with mean-zero. Due to our assumption $\sup_{u\in [0,1]}\twonorm{\btheta^*(u)}\leq \delta_s$, together with Theorem \ref{thm_high_appro_error}, we have
	\begin{align*}
	    \E_1\LRm{\LRs{\tbtheta_{S}(U_i)^{\top}\LRs{\bX - \bmu_1(U_i)}_{S}}^2\big|U_i}&= \tbtheta_{S}(U_i)^{\top}\bSigma_{SS}(U_i)\tbtheta_{S}(U_i)\\
	    &\leq \lambda_1\twonorm{\tbtheta_{S}(U_i)}^2\\
	    &\leq 2\lambda_1\LRs{\twonorm{\tbtheta_{S}(U_i)}^2 + \twonorm{\tbtheta_{S}(U_i) - \btheta_{S}^*(U_i)}^2}\\
	    &\lesssim 2\lambda_1\LRs{\delta_s + s L_n^{-d}} = O(1).
	\end{align*}
	Let $T_{i,jk} = B_k^{*}(U_i)\tbtheta_{S}(U_i)^{\top}\LRs{\bX - \bmu_1(U_i)}_{S}(X_{ij}-\mu_{1j}(U_i))$ for $1\leq j\leq p$ and $1\leq k\leq L_n$, then we have
	\begin{equation}\label{step_1_inf_0}
	    \LRtwonorm{(\bI_{11}^1-\bI_1^{*})_{(j,S)}\tgamma_{(S)}}\leq L_n\max_{1\leq k \leq L_n}\left|\frac{1}{n}\sum_{i\in \mathcal{I}_1} T_{i,jk}-\E_1[T_{i,jk}]\right|.
	\end{equation}
	In addition,
	\begin{equation*}
		\begin{aligned}
			&\E_1\left\{(T_{i,jk} - \E[T_{i,jk}])^2\exp\left[\eta\left|T_{i,jk} - \E_1[T_{i,jk}]\right|\right]\right\}\\
			\leq&\E_1\left\{T_{i,jk}^2\exp\left[\eta\left|T_{i,jk} - \E_1[T_{i,jk}]\right|\right]\right\} + [\E_1[T_{i,jk}]]^2\E_1\left\{\exp\left[\eta\left|T_{i,jk} - \E_1[T_{i,jk}]\right|\right]\right\}\\
			\leq&\exp(\eta |\E_1[T_{i,jk}]|)\left(\E_1\left\{T_{i,jk}^2\exp\left[\eta|T_{i,jk}|\right]\right\}+ [\E_1[T_{i,jk}]]^2\E_1\left\{\exp\left[\eta\left|T_{i,jk}\right|\right]\right\}\right).
		\end{aligned}
	\end{equation*}
Recall that $\E_1[(X_{ij}-\mu_{1j}(U_i))(\bX - \bmu_1(U_i))|U_i] = \bSigma_{j,S}(U_i)$, we have
\begin{equation*}
	\begin{aligned}
	        \E_1[T_{i,jk}] \leq& \E_1\LRm{B_k^{*}(U_i)\tbtheta_{S}(U_i)^{\top}\LRs{\bX - \bmu_1(U_i)}_{S}(X_{ij}-\mu_{1j}(U_i))}\\
	        \leq& \E\LRm{B_k^{*}(U_i) \LRabs{\tbtheta_{S}(U_i)^{\top}\bSigma_{j,S}(U_i)}}\\
	        \leq& \lambda_1 \sup_{u\in [0,1]}\twonorm{\tbtheta_{S}(u)} \E\LRm{B_k^{*}(U_i)}\\
	        \lesssim& L_n^{-1}.
	\end{aligned}
\end{equation*}
Denote $H_i = \tbtheta_{S}(U_i)^{\top}\LRs{\bX - \bmu_1(U_i)}_{S}$. Applying the second assertion of Lemma \ref{lemma:normal_moment_bound}, for any $\eta >$ for sufficiently large $C > 0$, we have
\begin{equation*}
	\begin{aligned}
		&\E_1\left\{T_{i,jk}^2\exp\left[\eta|T_{i,jk}|\right]\right\}\\
		& \leq \E_1\left\{(B_k^{*}(U_i))^2(X_{ij}-\mu_{1j}(U_i))^2H_i^2\exp\left[\eta|(X_{ij}-\mu_{1j}(U_i))H_i|\right]\right\}\\
		&\leq \E_1\LRm{(B_k^{*}(U_i))^2\LRs{\E\LRm{(X_{ij}-\mu_{1j}(U_i))^4 e^{2\eta |X_{ij}-\mu_{1j}(U_i)|}\big|U_i} \E\LRm{H_i^4 e^{2\eta |H_i|}\big|U_i}}^{1/2}}\\
		&\lesssim \E\LRm{(B_k^{*}(U_i))^2} \lesssim L_n^{-1}.
	\end{aligned}
\end{equation*}
In fact, we also used the fact $\E[H^2(U_i)|U_i]$ and $\E[(X_{ij}-\mu_{1j}(U_i))^2]$ are both bounded.
Moreover, for $\eta = 1/(2C\lambda_1)$ for sufficiently large $C > 0$, it holds
\begin{align*}
    \E_1\LRm{e^{\eta |T_{i,jk}|}} &\leq \E_1\LRm{e^{\eta |(X_{ij}-\mu_{1j}(U_i))H_i|}}\\
    &\leq \E_1\LRm{e^{\eta (|X_{ij}-\mu_{1j}(U_i)|^2 + H_i^2|)}}\\
    &\leq \E_1\LRm{\LRs{\E\LRm{e^{2\eta|X_{ij}-\mu_{1j}(U_i)|^2}\big| U_i} \E_1\LRm{e^{2\eta |H_i|^2}\big| U_i}}^{1/2} }\\
    &= O(1).
\end{align*}
Combing the results above, we conclude that
\begin{equation*}
	\E_1\left\{(T_{i,jk} - \E[T_{i,jk}])^2\exp\left[\eta\left|T_{i,jk} - \E[T_{i,jk}]\right|\right]\right\}\lesssim L_n^{-1}.
\end{equation*}
According to Lemma \ref{bern_ineq}, we are guaranteed that
\begin{equation*}
	\Prob\left(\max_{j,k}\left|\frac{1}{n}\sum_{i\in \mathcal{I}_1}T_{i,jk} - \E_1[T_{i,jk}]\right|\lesssim \sqrt{\frac{\log p}{nL_n}}\right)\geq 1 - L_np^{-\vartheta}.
\end{equation*}
In conjunction with (\ref{step_1_inf_0}), it follows that
\begin{equation}\label{step_1_inf_1}
    \Prob\left(\max_{1\leq j \leq p}\LRtwonorm{(\bI_{11}^1-\bI_1^{*})_{(j,S)}\tgamma_{(S)}}\lesssim \sqrt{\frac{L_n \log p}{n}}\right) \geq 1- L_np^{-\vartheta}.
\end{equation}
For $\bI_{12}^1$ and $\bI_{13}^1$, we have
\begin{equation}\label{step_1_inf_2}
    \begin{aligned}
        \LRtwonorm{(\bI_{12}^1)_{(j,S)}\tgamma_{(S)}} \leq L_n \max_{1\leq k\leq L_n}\left|\frac{1}{n}\sum_{i=1}^n E_{i,jk}\right|,
    \end{aligned}
\end{equation}
and
\begin{equation}\label{step_1_inf_3}
    \begin{aligned}
        \LRtwonorm{(\bI_{13}^1)_{(j,S)}\tgamma_{(S)}} \leq  L_n \max_{1\leq k\leq L_n}\left|\frac{1}{n}\sum_{i=1}^n F_{i, jk}\right|,
    \end{aligned}
\end{equation}
where
\begin{equation*}
	E_{i,jk} = B_k^{*}(U_i)(X_{ij}-\mu_{1j}(U_i))\tbtheta_{S}(U_i)^{\top}\LRs{\bmu_1(U_i) - \bmu_2(U_i)}_{S},
\end{equation*}
and
\begin{equation*}
	F_{i,jk} = B_k^{*}(U_i)(X_{ij}-\mu_{1j}(U_i))\tbtheta_{S}(U_i)^{\top}\LRs{\bX - \bmu_1(U_i)}_{S}.
\end{equation*}
Due to the fact that $\sum_{l=1}^s(\tnu_{(l)}^{\top}\bB_i^{*})^2 \leq 1$, we can verify that
\begin{equation*}
    \max\left\{\E_1\left(E_{i,jk}^2e^{\eta|E_{i,jk}|}\right)^2, \E_1\left(F_{i,jk}^2e^{\eta|F_{i,jk}|}\right)^2\right\}\lesssim L_n^{-1}.
\end{equation*}
Combining with Lemma \ref{bern_ineq}, (\ref{step_1_inf_2}) and (\ref{step_1_inf_3}), we are guaranteed that 
\begin{equation}\label{step_1_inf_4}
	\Prob\left(\max_{1\leq j \leq p}\left\{\LRtwonorm{(\bI_{12}^1)_{(j,S)}\tgamma_{(S)}} + \LRtwonorm{(\bI_{13}^1)_{(j,S)}\tgamma_{(S)}}\right\} \leq \sqrt{\frac{L_n\log p}{n}}\right)\geq 1- L_np^{-\vartheta}.
\end{equation}
For
\begin{align*}
    \bI_{14}^1 - \bI_3^{*}/2 &= \frac{1}{8n}\sum_{i\in \mathcal{I}_1}\left[(\bmu_1(U_i)-\bmu_2(U_i))(\bmu_{1}(U_i)-\bmu_2(U_i))^{\top}\right]\otimes (\bB_i\bB_i^{\top})\\
    &- \frac{1}{8}\E\left\{\left[(\bmu_{1}(U)-\bmu_{2}(U))(\bmu_{1}(U)-\bmu_{2}(U))^{\top}\right]\otimes (\bB\bB^{\top})\right\},
\end{align*}
it is easy to obtain that
\begin{equation}\label{step_1_inf_5}
	\Prob\left(\max_{1\leq j\leq p}\LRtwonorm{(\bI_{14}^1 - \frac{1}{2}\bI_3^{*})_{(j,S)}\tgamma_{(S)}}\lesssim \sqrt{\frac{L_n\log p}{n}}\right)\geq 1 - L_np^{-\vartheta}.
\end{equation}
Combining (\ref{step_1_inf_2}), (\ref{step_1_inf_4}) and (\ref{step_1_inf_5}), we have
\begin{equation}\label{step_1_inf}
	\Prob\left(\max_{1\leq j\leq p}\LRtwonorm{(\bI_1^1-\bI_1^{*}-\frac{1}{2}\bI_3^{*})_{(j,S)}\tgamma_{(S)}}\lesssim \sqrt{\frac{L_n\log p}{n}}\right)\geq 1- L_np^{-\vartheta}.
\end{equation}
Similarly, we also have
\begin{equation}\label{step_1_inf_c}
	\Prob\left(\max_{1\leq j\leq p}\LRtwonorm{(\bI_1^2-\bI_2^{*}-\frac{1}{2}\bI_3^{*})_{(j,S)}\tgamma_{(S)}}\lesssim \sqrt{\frac{L_n\log p}{n}}\right)\geq 1- L_np^{-\vartheta}.
\end{equation}
\paragraph{Step 2. upper bounding $\max_{1\leq j\leq p}\twonorm{(\bI_2^1 + \bI_2^2+\bI_3^1 + \bI_3^2 +\bI_4^1 + \bI_4^2)_{(j,S)}\tgamma_{(S)}}$.} Recall that
\begin{align*}
    \bI_2^1+ \bI_2^2&=\frac{1}{2n}\sum_{i\in \mathcal{I}_1}\left\{\left[(\bX_i-\bmu(U_i))(\bmu(U_i)-\hmu(U_i))^{\top}\right]\otimes (\bB_i\bB_i^{\top})\right\}\\
    &+\frac{1}{2n}\sum_{i\in \mathcal{I}_2}\left\{\left[(\bX_i-\bmu(U_i))(\bmu(U_i)-\hmu(U_i))^{\top}\right]\otimes (\bB_i\bB_i^{\top})\right\},
\end{align*}
then we have
\begin{equation*}
	\begin{aligned}
		&\twonorm{(\bI_2^1+ \bI_2^2)_{(j,S)}\tgamma_{(S)}}\\
		&\leq L_n\max_{1\leq k\leq L_n}\left|\frac{1}{2n}\sum_{i\in \mathcal{I}_1}(X_{ij}-\mu_{1j}(U_i))B_k^{*}(U_i)\tbtheta(U_i)_{S}^{\top}(\hmu(U_i) -\bmu(U_i))_{S}\right|\\
		&+L_n\max_{1\leq k\leq L_n}\left|\frac{1}{2n}\sum_{i\in \mathcal{I}_2}(X_{ij}-\mu_{2j}(U_i))B_k^{*}(U_i)\tbtheta(U_i)_{S}^{\top}(\hmu(U_i) -\bmu(U_i))_{S}\right|\\
		&+L_n\max_{1\leq k\leq L_n}\left|\frac{1}{4n}\sum_{i\in \mathcal{I}_1} G_{i,jk}(U_i) - \frac{1}{4n}\sum_{i\in \mathcal{I}_2} G_{i,jk}(U_i)\right|,
	\end{aligned}
\end{equation*}
where $G_{i,jk}(U_i) = (\mu_{1j}(U_i)-\mu_{2j}(U_i))B_k^{*}(U_i)\tbtheta(U_i)_{S}^{\top}(\hmu(U_i) -\bmu(U_i))_{S}$. Notice that
\begin{align*}
    \tbtheta(U_i)_{S}^{\top}(\hmu(U_i) -\bmu(U_i))_{S} &= \sum_{l=1}^s (\halpha_{l} - \talpha_{l})^{\top}\bB_i \widetilde{\theta}_{l}(U_i)\\
    &+\sum_{l=1}^s (\talpha_{l}^{\top}\bB_i - \mu_{l}(U_i)) \widetilde{\theta}_{l}(U_i)\\
    &= \tbtheta(U_i)_{S}^{\top} \LRs{\hbM_{.S} - \tbM_{.S}}^{\top}\bB_i + \tbtheta(U_i)_{S}^{\top}\LRs{\tbM_{.S}^{\top}\bB_i},
\end{align*}
where $\hbM_{.S} = (\halpha_{1},...,\halpha_{s})$ and $\tbM_{.S} = (\talpha_1,...,\talpha_{s})$. Then we have
\begin{align*}
    &\left|\frac{1}{2n}\sum_{i\in \mathcal{I}_1}(X_{ij}-\mu_{1j}(U_i))B_k^{*}(U_i)\sum_{l=1}^s [\hat{\mu}_l(U_i) - \mu_l(U_i)]\widetilde{\theta}_{l}(U_i)\right|\\
    &\leq \LRabs{\frac{1}{2n}\sum_{i\in \mathcal{I}_1}(X_{ij}-\mu_{1j}(U_i))B_k^{*}(U_i) \tbtheta(U_i)_{S}^{\top} \LRs{\hbM_{.S} - \tbM_{.S}}^{\top}\bB_i}\\
    &+\left|\frac{1}{2n}\sum_{i\in \mathcal{I}_1}(X_{ij}-\mu_{1j}(U_i))B_k^{*}(U_i)\tbtheta(U_i)_{S}^{\top}\LRs{\tbM_{.S}^{\top}\bB_i}\right|.
\end{align*}
From Proposition \ref{pro_mean_bound}, we know that $\|\hbM_{.S} - \tbM_{.S}\|_{F}\lesssim \sqrt{s}a_n$. By utilizing the same chaining technique in Section \ref{proof:lemma_matrix_bound_4} to $\hbM_{.S} - \tbM_{.S}$, we can show that
\begin{equation}\label{step_2_inf}
	\Prob\left(\max_{1\leq j\leq p}\LRtwonorm{(\bI_2^1+\bI_2^2)_{(j,S)}\tgamma_{(S)}}\lesssim a_nL_n s\sqrt{\frac{\log p}{n}}\right)\geq 1 - p^{-\vartheta sL_n}.
\end{equation}
Recall that
\begin{align*}
    \bI_3^1+ \bI_3^2 &= \frac{1}{2n}\sum_{i\in \mathcal{I}_1}\left\{\left[(\bmu(U_i)-\hmu(U_i))(\bX_i-\bmu(U_i))^{\top}\right]\otimes (\bB_i\bB_i^{\top})\right\}\\
    &+\frac{1}{2n}\sum_{i\in \mathcal{I}_2}\left\{\left[(\bmu(U_i)-\hmu(U_i))(\bX_i-\bmu(U_i))^{\top}\right]\otimes (\bB_i\bB_i^{\top})\right\}.
\end{align*}
Similarly, we can verify
\begin{equation}\label{step_2_inf_2}
	\Prob\left(\max_{1\leq j\leq p}\LRtwonorm{(\bI_3^1+\bI_3^2)_{(j,S)}\tgamma_{(S)}}\lesssim a_nL_n\sqrt{\frac{\log p}{n}}\right)\geq 1- p^{-\vartheta L_n}.
\end{equation}
For
\begin{align*}
    \bI_4^1 + \bI_4^2 &= \frac{1}{2n}\sum_{i\in \mathcal{I}_1}\left[(\bmu(U_i)-\hmu(U_i))(\bmu(U_i)-\hmu(U_i))^{\top}\right]\otimes (\bB_i\bB_i^{\top})\\
    &+\frac{1}{2n}\sum_{i\in \mathcal{I}_2}\left[(\bmu(U_i)-\hmu(U_i))(\bmu(U_i)-\hmu(U_i))^{\top}\right]\otimes (\bB_i\bB_i^{\top}),
\end{align*}
it follows from Proposition \ref{pro_mean_bound} that
\begin{equation}\label{step_2_inf_3c}
\begin{aligned}
    &\LRtwonorm{(\bI_4^1 + \bI_4^2)_{(j,S)}\tgamma_{(S)}}\\
    &\leq L_n \max_{1\leq k \leq L_n, 1\leq i \leq 2n} \LRabs{B_k^{*}(U_i)[\hat{\mu}_j(U_i) - \mu_j(U_i)]\tbtheta(U_i)_{S}^{\top}(\hmu(U_i) -\bmu(U_i))_{S}}\\
    &\leq L_n a_n \max_{1\leq i \leq 2n} \LRabs{\tbtheta(U_i)_{S}^{\top}(\hmu(U_i) -\bmu(U_i))_{S}}\\
    &\lesssim L_n a_n \max_{1\leq i \leq 2n}\LRtwonorm{\LRs{\widehat{\bmu}(U_i) - \bmu(U_i)}_{S}}\twonorm{\tbtheta(U_i)_{S}}\\
    &\leq \sqrt{s}L_n a_n^2,
\end{aligned}
\end{equation}
holds with probability at least $1 - sL_np^{-\vartheta}$. Combining (\ref{step_1_inf})-(\ref{step_2_inf_3c}), we have
\begin{equation*}
	\max_{1\leq j\leq p}\twonorm{(\bD_n - \bD)_{(j,S)}\tgamma_{(S)}}\lesssim \sqrt{\frac{L_n\log p}{n}}+a_nL_n s\sqrt{\frac{\log p}{n}},
\end{equation*}
holds with probability at least $1-sL_np^{-\vartheta} - L_np^{-\vartheta sL_n} - L_np^{-\vartheta}$.
\end{proof}

\section{Proofs of auxiliary lemmas in Section \ref{sec:main_lemma_proofs}}\label{sec::lemma_proof}

\subsection{Proof of Lemma \ref{lemma_matrix_bound_1}}\label{proof:lemma_matrix_bound_1}
\begin{proof}[Proof of Lemma \ref{lemma_matrix_bound_1}]
	We only prove the bound for $\mathcal{I}_1$, and the case in $\mathcal{I}_2$ is similar. Here we use $\E_1[\cdot]$ to denote the conditional expectation $\E[\cdot| Y = 1]$. 
	
	Let $\mathbb{S}^{pL_n-1}$ be the unit sphere in $\mathbb{R}^{pL_n}$, we denote the $1/8$-covering set of $\mathbb{S}^{pL_n-1}$ by $\{\bnu_1,...,\bnu_{N}\}$ with $N\leq 17^{(pL_n)}$. It follows that for any $\bnu\in \mathbb{S}^{pL_n-1}$, there exist some $\bnu_l$ such that $\twonorm{\bnu - \bnu_l}\leq 1/8$. Let $\bQ_i = \bZ_i\bZ_i^{\top} - \E_1[\bZ_i\bZ_i^{\top}]$, then we have
		\begin{equation*}
		\LRtwonorm{\frac{1}{n}\sum_{i\in \mathcal{I}_1} \bQ_i}\leq 2\max_{1\leq l\leq N}\LRabs{\frac{1}{n}\sum_{i\in \mathcal{I}_1}\bnu_l^{\top}\bQ_i\bnu_l}.
	\end{equation*}
	Note that
	\begin{equation*}
		\begin{aligned}
			L_n^{-1}\bnu_l^{\top}\bQ_i\bnu_l &= \LRs{\bnu_l^{\top}[(\bX_i - \bmu_1(U_i))\otimes \bB_i^{*})]}^2 -\E_1\LRs{\bnu_l^{\top}[(\bX_i - \bmu_1(U_i))\otimes \bB_i^{*})]}^2\\
			& = [\tnu_{li}^{\top}(\bX_i - \bmu_1(U_i))]^2 - \E_1[\tnu_{li}^{\top}(\bX_i - \bmu_1(U_i))]^2
		\end{aligned}
	\end{equation*}
	where $\tnu_{l}(U_i) = ((\bnu_{l})_{(1)}^{\top}\bB_i^{*},...,(\bnu_{l})_{(p)}^{\top}\bB_i^{*})^{\top}$ for $\bnu_{l} = ((\bnu_{l})_{(1)}^{\top},...,(\bnu_{l})_{(p)}^{\top})^{\top}$. Let $R_{li} = \tnu_{l}(U_i)^{\top}(\bX_i - \bmu_1(U_i))$, then $R_{li} \sim \mathcal{N}(0, \tnu_{l}(U_i)^{\top}\bSigma(U_i)\tnu_{l}(U_i))$ given $U_i$ and $Y_i = 1$. Let $\sigma_{li}^2 = \tnu_{l}(U_i)^{\top}\bSigma(U_i)\tnu_{l}(U_i)$, then we have $\sigma_{li}^2 \leq \lambda_2$
	\begin{align*}
	    \E_1[\sigma_{li}^2] &\leq \E_1\LRm{\twonorm{\tnu_{l}(U_i)}^2 \twonorm{\bSigma(U_i)}} \leq \lambda_1\E\LRm{\twonorm{\tnu_{l}(U_i)}^2}\\
	    &\leq \lambda_1 \sum_{j=1}^p\E_1\LRm{\LRs{\bnu_{(j)}^{\top}\bB_i^*}^2} \lesssim L_n^{-1}\sum_{j=1}^p\twonorm{\bnu_{(j)}}^2 = L_n^{-1}.
	\end{align*}
	For $\eta = 1/(8\lambda_2M_2)$, simple calculation gives $\E_1[e^{\eta R_{li}^2}] = O(1)$. Using H\"older's inequality, we also have
	\begin{equation*}
		\begin{aligned}
			&\E_1\LRm{\LRs{R_{li}^2 - \E_1[R_{li}^2]}^2 \exp\LRs{\eta|R_{li}^2 - \E_1[R_{li}^2]|}}\\
			\leq& 2\exp\LRs{\eta\E_1[R_{li}^2]}\LRs{\E_1\LRm{R_{li}^4e^{\eta R_{li}^2}} + \LRs{\E_1\LRm{R_{li}^2}}^2 \E_1\LRm{e^{\eta R_{li}^2}}}\\
			\lesssim & \E_1\LRm{\sigma_{li}^4 \E_1\LRm{\frac{R_{li}^4}{\sigma_{li}^4} e^{\eta R_{li}^2}|U_i}}+ \LRs{\E_1\LRm{\sigma_{li}^2}}^2\\
			\leq & \E_1\LRm{\sigma_{li}^2 \LRs{ \E_1\LRm{(R_{li}/\sigma_{li})^8|U_i} \E_1\LRm{e^{2\eta R_{li}^2}|U_i}}^{1/2}}+ L_n^{-2}\\
			\lesssim & \E_1\LRm{\sigma_{li}^2} + L_n^{-2} \lesssim L_n^{-1}.
		\end{aligned}
	\end{equation*}
	Invoking Lemma \ref{bern_ineq}, we can obtain that
	\begin{equation*}
		\Prob\LRs{\max_{1\leq j\leq N}\LRabs{\frac{1}{n}\sum_{i=1}^n\bnu_l^{\top}\bQ_i\bnu_l}\lesssim L_n\sqrt{\frac{p\log n}{n}}} \geq 1 - n^{-\vartheta p L_n}.
	\end{equation*}
	
	Next we prove the conclusion for matrix-vector-product. It suffices to show that for any fixed $\bnu \in \mathbb{R}^{pL_n-1}$,
	\begin{align}\label{eq:mat_vec_bound}
	    \mathbb{P}\LRs{\LRabs{\frac{1}{n}\sum_{i\in \mathcal{I}_1} \bnu^{\top}\LRs{\bZ_i\bZ_i^{\top} - \E_1\LRm{\bZ_i\bZ_i^{\top}}} \tgamma} \lesssim \sqrt{\frac{L_n p \log n}{n}}} \leq n^{-\vartheta pL_n}.
	\end{align}
	Let $\tnu(U_i) = (\bnu_{(1)}^{\top}\bB_i^*,...,\bnu_{(p)}^{\top}\bB_i^*)^{\top}$, then notice that
	\begin{align*}
	    \bnu^{\top}\bZ_i\bZ_i^{\top} \tgamma &= \bnu^{\top}\LRs{\LRs{\bX_i - \bmu_1(U_i)}\otimes \bB_i}\tgamma^{\top}\LRs{\LRs{\bX_i - \bmu_1(U_i)}\otimes \bB_i}\\
	    &= \sqrt{L_n}\LRs{\bX_i - \bmu_1(U_i)}^{\top}\tnu(U_i) \LRs{\bX_i - \bmu_1(U_i)}^{\top}\tbtheta(U_i).
	\end{align*}
	Denote $R_i = \LRs{\bX_i - \bmu_1(U_i)}^{\top}\tnu(U_i)$ and $S_{i} = \LRs{\bX_i - \bmu_1(U_i)}^{\top}\tbtheta(U_i)$, then
	\begin{align*}
	    R_i|U_i \sim \mathcal{N}\LRs{0, \tnu(U_i)^{\top}\bSigma(U_i)\tnu(U_i)},\quad S_i|U_i \sim \mathcal{N}\LRs{0, \tbtheta(U_i)^{\top}\bSigma(U_i)\tbtheta(U_i)}.
	\end{align*}
	Also, we know that $\E_1[R_i^2|U_i] \leq \lambda_1 $ and $\E_1[S_i^2|U_i] \lesssim \lambda_1$ since $\twonorm{\tnu(u)}\leq 1$ and
	\begin{align*}
	    \sup_{u\in [0,1]}\twonorm{\tbtheta(u)} \lesssim \twonorm{\btheta^*(u)} + \sqrt{p}L_n^{-d} = O(1),
	\end{align*}
	which is true due to the assumptions $\sup_{u\in [0,1]}\twonorm{\btheta^*(u)} = O(1)$ and $\sqrt{p}L_n^{-d} = o(1)$. Using H\"older's inequality, we have
	\begin{align}
	    &\E_1\LRm{(R_iS_i - \E_1[R_iS_i])^2 e^{\eta |R_i S_i - \E_1[R_iS_i]|}}\nonumber\\
	    \leq & 2\exp\LRs{\eta |\E_1[R_iS_i]|}\LRs{\E_1\LRm{R_i^2S_i^2 e^{\eta|R_iS_i|}} + \LRs{\E_1\LRm{R_iS_i}}^2 \E_1\LRm{e^{\eta|R_iS_i|}}}\nonumber\\
	    \lesssim & \E_1\LRm{R_i^2S_i^2 e^{\eta|R_iS_i|}} + \E_1[R_i^2]\E_1[S_i^2]\E_1\LRm{e^{\eta|R_iS_i|}}.
	    \label{eq:Si_Ri_moment}
	\end{align}
	Denote $\sigma_i^2 = \tnu(U_i)^{\top}\bSigma(U_i)\tnu(U_i)$.
	Then applying Lemma \ref{lemma:normal_moment_bound} and moment formula of normal distribution, it holds that
	\begin{align*}
	    \E_1\LRm{R_i^4 e^{2\eta |R_i|}\big| U_i} &\leq 32e^{2\eta^2 \sigma_i^2}\LRs{(\eta\sigma_i^2)^4 + 6(\eta\sigma_i^2)^2\sigma_i^2 + 3\sigma_i^4}\\
	    &\lesssim \sigma_i^4 = \LRs{\tnu(U_i)^{\top}\bSigma(U_i)\tnu(U_i)}^2,
	\end{align*}
	where the second inequality holds since $\sigma_i$ is bounded. Similarly, we can also verify that
	\begin{align*}
	    \E_1\LRm{S_i^4 e^{2\eta |S_i|}\big| U_i} \lesssim \LRs{\tbtheta(U_i)^{\top}\bSigma(U_i)\tbtheta(U_i)}^2.
	\end{align*}
	It follows from $\tbtheta(U_i)^{\top}\bSigma(U_i)\tbtheta(U_i)$ is bounded that
	\begin{align}
	    \E_1\LRm{R_i^2S_i^2 e^{\eta|R_iS_i|}} &= \E_1\LRm{\E_1\LRm{R_i^2S_i^2 e^{\eta|R_iS_i|}\big| U_i}}\nonumber\\
	    &\leq \E_1\LRm{\LRs{\E_1\LRm{R_i^4 e^{2\eta |R_i|}\big| U_i} \E_1\LRm{S_i^4 e^{2\eta |S_i|}\big| U_i}}^{1/2}}\nonumber\\
	    &\lesssim \E_1\LRm{\LRs{\tnu(U_i)^{\top}\bSigma(U_i)\tnu(U_i)}\LRs{\tbtheta(U_i)^{\top}\bSigma(U_i)\tbtheta(U_i)}}\nonumber\\
	    &\lesssim \E_1\LRm{\tnu(U_i)^{\top}\bSigma(U_i)\tnu(U_i)}\nonumber\\
	    &\lesssim \E_1\LRm{\twonorm{\tnu(U_i)}^2} \lesssim L_n^{-1}.
	    \label{eq:Si_Ri_moment_1}
	\end{align}
	Then we take $\eta = 1/(C\lambda_1)$ for sufficiently large $C > 0$, it holds that
	\begin{align}
	    \E_1\LRm{e^{\eta|R_iS_i|}} &\leq \LRs{\E_1\LRm{e^{\eta (R_i^2 + S_i^2)}}}\leq \LRs{\E_1\LRm{e^{2\eta R_i^2}} \E_1\LRm{e^{2\eta S_i^2}}}^{1/2} = O(1).
	    \label{eq:Si_Ri_moment_2}
	\end{align}
	Substituting \eqref{eq:Si_Ri_moment_1} and \eqref{eq:Si_Ri_moment_2} into \eqref{eq:Si_Ri_moment}, we have
	\begin{align*}
	     \E_1\LRm{(R_iS_i - \E_1[R_iS_i])^2 e^{\eta |R_i S_i - \E_1[R_iS_i]|}}\lesssim L_n^{-1}.
	\end{align*}
	Applying Lemma \ref{bern_ineq} again, we can prove \eqref{eq:mat_vec_bound} immediately.
\end{proof}
\subsection{Proof of Lemma \ref{lemma_matrix_bound_2} and \ref{lemma_matrix_bound_3}}\label{proof:lemma_matrix_bound_2}
\begin{proof}
    The proofs of Lemma \ref{lemma_matrix_bound_2} and \ref{lemma_matrix_bound_3} are similar, here we only prove Lemma \ref{lemma_matrix_bound_2}.

    For the first assertion for operator norm of matrix, it suffices to show for any fixed $\bnu \in \mathbb{S}^{pL_n - 1}$ such that
	\begin{equation*}
		\left|\frac{1}{n}\sum_{i \in \mathcal{I}_1} \bnu^{\top}\LRl{[(\bX_i - \bmu_1(U_i))(\bmu_1(U_i) - \bmu_2(U_i))^{\top}]\otimes (\bB_i\bB_i^{\top})} \bnu\right|\lesssim L_n\sqrt{\frac{p \log n}{n}},
	\end{equation*}
	holds with probability at least $1-n^{-\vartheta p L_n}$.
	For $\bnu = (\bnu_{(1)}^{\top},...,\bnu_{(p)}^{\top})^{\top}\in \mathbb{S}^{pL_n-1}$, we write $\tnu(U_i) = (\bnu_{(1)}^{\top}\bB_i^{*},...,\bnu_{(p)}^{\top}\bB_i^{*})^{\top}$. Then we have
	\begin{align*}
		&\bnu^{\top}\LRl{[(\bX_i - \bmu_1(U_i))(\bmu_1(U_i) - \bmu_2(U_i))^{\top}]\otimes (\bB_i\bB_i^{\top})} \bnu\\
		&\qquad = L_n (\tnu(U_i)^{\top}(\bmu_1(U_i) - \bmu_2(U_i)))[\tnu(U_i)^{\top}(\bX_i-\bmu_{1}(U_i))].
	\end{align*}
	Let 
	\begin{align*}
	    T_{i} = (\tnu(U_i)^{\top}(\bmu_1(U_i) - \bmu_2(U_i)))[\tnu(U_i)^{\top}(\bX_i-\bmu_{1}(U_i))],
	\end{align*}
	and $\sigma_i^2 = \tnu(U_i)^{\top}\bSigma(U_i)\tnu(U_i)$, then for $\eta > 0$ it holds that
	\begin{equation*}
		\begin{aligned}
			\E_1[T_i^2e^{\eta|T_i|}]&\lesssim \E_1\LRm{\LRs{\tnu(U_i)^{\top}(\bX_i-\bmu_{1}(U_i))}^2 e^{\eta |\tnu(U_i)^{\top}(\bX_i-\bmu_{1}(U_i))|}}\\
			&\leq 2\E\LRm{e^{\frac{\eta^2\sigma_i^2}{2}} \LRs{\sigma_i^2 + \phi^2 \sigma_i^4}}\\
			&\lesssim \E\LRm{\twonorm{\tnu(U_i)}^2} \lesssim L_n^{-1},
		\end{aligned}
	\end{equation*}
	where the first inequality follows from $|\tnu(u)^{\top}(\bmu_1(u) - \bmu_2(u))|$ is uniformly bounded. By Lemma \ref{bern_ineq}, we are guaranteed that
	\begin{equation*}
		\Prob\LRs{\left|\frac{1}{n}\sum_{i\in \mathcal{I}_1}T_i\right|\lesssim \sqrt{\frac{p\log n}{n}}} \geq 1 - n^{-\vartheta p L_n}.
	\end{equation*}
	
	For the assertion for matrix-vector-product, it suffices to show
	\begin{align*}
	    \left|\frac{1}{n}\sum_{i \in \mathcal{I}_1} \bnu^{\top}\LRl{[(\bX_i - \bmu_1(U_i))(\bmu_1(U_i) - \bmu_2(U_i))^{\top}]\otimes (\bB_i\bB_i^{\top})} \tgamma\right|\lesssim \sqrt{\frac{pL_n \log n}{n}}
	\end{align*}
	holds with probability at least $1-n^{-\vartheta p L_n}$. Notice that
	\begin{align*}
	    &\bnu^{\top}\LRl{[(\bX_i - \bmu_1(U_i))(\bmu_1(U_i) - \bmu_2(U_i))^{\top}]\otimes (\bB_i\bB_i^{\top})} \tgamma\\
	    =& \sqrt{L_n} \LRs{\tnu(U_i)^{\top}(\bX_i-\bmu_{1}(U_i))} \LRs{\tbtheta(U_i)^{\top}(\bmu_1(U_i) - \bmu_2(U_i))},
	\end{align*}
	and
	\begin{align*}
	    \LRabs{\tbtheta(U_i)^{\top}(\bmu_1(U_i) - \bmu_2(U_i))}\leq \delta_p \sup_{u\in[0,1]} \twonorm{\tbtheta(u)} = O(1).
	\end{align*}
	Denote $R_i = \LRs{\tnu(U_i)^{\top}(\bX_i-\bmu_{1}(U_i))} \LRs{\tbtheta(U_i)^{\top}(\bmu_1(U_i) - \bmu_2(U_i))}$. Then we can get
	\begin{align*}
	    \E_1\LRm{R_i^2e^{\eta |R_i|}} &\lesssim \E_1\LRm{\LRs{\tnu(U_i)^{\top}(\bX_i-\bmu_{1}(U_i))}^2 e^{\eta |\tnu(U_i)^{\top}(\bX_i-\bmu_{1}(U_i))|}}\\
	    &\lesssim L_n^{-1}.
	\end{align*}
	Applying Lemma \ref{bern_ineq}, we can prove the desired result.
\end{proof}

\subsection{Proof of Lemma \ref{lemma_matrix_bound_4}}\label{proof:lemma_matrix_bound_4}
\begin{proof}
    Let $\tbM = (\tbM_{1},...,\tbM_{p}) \in \mathbb{R}^{L_n\times p}$ where the $j$-th column $\bM_j = \frac{1}{2}(\talpha_{1j}+\talpha_{2j})$ and
    \begin{align*}
        \talpha_{1j} = \LRs{\E[\bB\bB^{\top}]}^{-1}\E[\bB X_j|Y = 1],\quad \talpha_{1j} = \LRs{\E[\bB\bB^{\top}]}^{-1}\E[\bB X_j|Y = 0].
    \end{align*}
    Similarly, we denote $\hbM = (\hbM_1,...,\hbM_p) \in \mathbb{R}^{L_n\times p}$, where the $j$-th column $\hbM_j = \frac{1}{2}(\halpha_{1j}+\halpha_{2j})$ and
    \begin{align*}
        \halpha_{1j} = \LRs{\frac{1}{n}\sum_{i\in \mathcal{I}_1} \bB_i\bB_i^{\top}}^{-1}\frac{1}{n}\sum_{i\in \mathcal{I}_1}\bB_i X_{ij},\ \halpha_{2j} = \LRs{\frac{1}{n}\sum_{i\in \mathcal{I}_2} \bB_i\bB_i^{\top}}^{-1}\frac{1}{n}\sum_{i\in \mathcal{I}_2}\bB_i X_{ij}.
    \end{align*}
    Recalling the approximation error bound in the proof of Proposition \ref{pro_mean_bound}:
	\begin{align*}
	    \sup_{u \in [0,1]}|\talpha_{1j}^{\top}\bB(u) - \mu_{1j}(u)| \lesssim L_n^{-d},\quad \sup_{u \in [0,1]}|\talpha_{2j}^{\top}\bB(u) - \mu_{2j}(u)| \lesssim L_n^{-d}.
	\end{align*}
    It means that $\twonorm{\tbM^{\top}\bB_i - \bmu(U_i)} \lesssim \sqrt{p}L_n^{-d}$. In addition, we define the good event:
    \begin{align*}
        \mathcal{A} := \LRl{\twonorm{\hbM_j - \tbM_j} \lesssim a_n:\text{ for }1\leq j \leq p},
    \end{align*}
    where $a_n = \sqrt{\frac{L_n\log n}{n}}+L_n^{-d}$. By Proposition \ref{pro_mean_bound}, we know that $\mathbb{P}(\mathcal{A}^c) \leq L_n n^{-\vartheta}$. Next we prove Lemma \ref{lemma_matrix_bound_3} under the good event $\mathcal{A}$.
    
    We first prove the bound for the operator norm of the matrix. Let $\bA_i = [(\hmu(U_i) - \bmu(U_i))(\bX_i-\bmu_{1}(U_i))^{\top}]\otimes (\bB_i\bB_i^{\top})$. According to the proof of Lemma \ref{lemma_matrix_bound_1}, it suffices to show that for any fixed $\bnu, \bxi \in \mathbb{S}^{pL_n - 1}$ such that
	\begin{equation*}
		\left|\frac{1}{n}\sum_{i \in \mathcal{I}_1} \bnu^{\top}\bA_i \bxi\right|\lesssim a_n L_n^{3/2} p\sqrt{\frac{\log n}{n}}
	\end{equation*}
    holds with probability at least $1 - n^{-\vartheta pL_n}$.
	Notice that
	\begin{equation}\label{eq:A_matrix_bound}
		\begin{aligned}
			&\left|\frac{1}{n}\sum_{i \in \mathcal{I}_1} \bnu^{\top}\bA_i \bxi\right| = L_n\LRabs{\frac{1}{n}\sum_{i \in \mathcal{I}_1}[\tnu_{i}^{\top}(\hmu(U_i) -\bmu(U_i))][\txi_{i}^{\top}(\bX_i-\bmu_{1}(U_i))]}\\
			&\qquad\leq L_n\LRabs{\frac{1}{n}\sum_{i \in \mathcal{I}_1}[\txi(U_i)^{\top}(\bX_i-\bmu_{1}(U_i))][\tnu(U_i)^{\top}(\tbM^{\top}\bB_i - \bmu(U_i))]}\\
			&\qquad \qquad + L_n \LRabs{\frac{1}{n}\sum_{i \in \mathcal{I}_1}[\txi(U_i)^{\top}(\bX_i-\bmu_{1}(U_i))][\tnu(U_i)^{\top}(\hbM - \tbM)^{\top}\bB_i]},
		\end{aligned}
	\end{equation}
	where $\tnu(U_i) = (\bnu_{(1)}^{\top}\bB_i^{*},...,\bnu_{(p)}^{\top}\bB_i^{*})^{\top}$ and $\txi(U_i) = (\bxi_{(1)}^{\top}\bB_i^{*},...,\bxi_{(p)}^{\top}\bB_i^{*})^{\top}$ for $\bnu = (\bnu_{(1)}^{\top},...,\bnu_{(p)}^{\top})^{\top}$ and $\bxi = (\bxi_{(1)}^{\top},...,\bxi_{(p)}^{\top})^{\top}$.
	Also, we know that 
	\begin{align*}
	    \txi(U_i)^{\top}(\bX_i-\bmu_{1}(U_i))|U_i \sim \mathcal{N}\LRs{0, \txi(U_i)^{\top}\bSigma(U_i)\txi(U_i)}.
	\end{align*}
	Denote $T_{i} = [\txi(U_i)^{\top}(\bX_i-\bmu_{1}(U_i))][\tnu(U_i)^{\top}(\tbM^{\top}\bB_i - \bmu(U_i))]$. Apply the second assertion in Lemma \ref{lemma:normal_moment_bound} with any constant $\eta > 0$, we get
	\begin{align*}
	    &\E_1\LRm{T_i^2 e^{\eta |T_i|}} \leq \E_1\LRm{T_i^2 e^{\eta |\txi(U_i)^{\top}(\bX_i-\bmu_{1}(U_i))|}}\\
	    &\lesssim pL_n^{-2d}\E_1\LRm{\LRs{\txi(U_i)^{\top}(\bX_i-\bmu_{1}(U_i))}^2e^{\eta|\txi(U_i)^{\top}(\bX_i-\bmu_{1}(U_i))|}}\\
	    &\leq 2pL_n^{-2d} \E\LRm{e^{\frac{\eta^2 \txi(U_i)^{\top}\bSigma(U_i)\txi(U_i)}{2}}\LRs{\txi(U_i)^{\top}\bSigma(U_i)\txi(U_i) + \eta^2\LRs{\txi(U_i)^{\top}\bSigma(U_i)\txi(U_i)}^2}}\\
	    &\lesssim pL_n^{-2d} \E\LRm{\txi(U_i)^{\top}\bSigma(U_i)\txi(U_i)}\lesssim pL_n^{-2d} L_n^{-1},
	\end{align*}
	where we also used the fact
	\begin{align*}
	    \txi(U_i)^{\top}\bSigma(U_i)\txi(U_i) \leq \lambda_1 \twonorm{\txi(U_i)}^2 \leq \lambda_1 \twonorm{\bB_i^*}^2\sum_{j=1}^p \twonorm{\bnu_{(j)}}^2 \leq \lambda_1.
	\end{align*}
	Applying Lemma \ref{bern_ineq}, we can verify that
	\begin{equation}\label{mat_x_1}
		\LRabs{\frac{1}{n}\sum_{i \in \mathcal{I}_1}[\txi(U_i)^{\top}(\bX_i-\bmu_{1}(U_i))][\tnu(U_i)^{\top}(\tbM^{\top}\bB_i - \bmu(U_i))]} \lesssim pL_n^{-d}\sqrt{\frac{L_n\log n}{n}},
	\end{equation}
	holds with probability at least $1 - n^{-\vartheta p L_n}$. Next we proceed to bound the second term in \eqref{eq:A_matrix_bound}. We define a matrix set as
	\begin{equation*}
		\boldsymbol{\Xi} = \LRl{\bM \in \mathbb{R}^{L_n \times p}:\ \twonorm{\bM_{j}} \leq a_n \text{ for }j=1,2,...,p},
	\end{equation*}
	where $a_n = O(\sqrt{L_n\log n/n} + L_n^{-d})$ and $\bM_j$ is the $j$-th column of $\bM$. For each $1\leq j\leq p$, we may find a set $\{\boldsymbol{\zeta}^{\ell} \in \mathbb{R}^{L_n},1\leq \ell \leq n^{ML_n}:\ \twonorm{\boldsymbol{\zeta}^{\ell}}\leq a_n\}$ such that there exists some $1\leq \ell \leq n^{ML_n}$ satisfying that $\twonorm{\bM_j - \boldsymbol{\zeta}^{\ell}} \leq n^{-M}\sqrt{L_n}a_n$\footnote{For each coordinate of $\bM_j$, we can divide the interval $[-a_n, a_n]$ into $n^M$ small intervals with equal length $2a_n/N^M$.}. It means that we can find a subset $\boldsymbol{\Xi}^{'} = \{\bM^{\ell}: 1\leq \ell \leq n^{MpL_n}\} \subseteq \boldsymbol{\Xi}$. And for any $\bM \in \boldsymbol{\Xi}$, there exists some $1\leq \ell \leq n^{MpL_n}$ such that $\twonorm{\bM_{j}^{\ell} - \bM_{j}} \leq a_n \sqrt{L_n} n^{-M}$. It follows that for any $\bM \in \boldsymbol{\Xi}$
	\begin{align*}
	    &\LRabs{\frac{1}{n}\sum_{i \in \mathcal{I}_1}[\txi(U_i)^{\top}(\bX_i-\bmu_{1}(U_i))][\tnu(U_i)^{\top}\bM^{\top}\bB_i]}\\ &\qquad \leq  \LRabs{\frac{1}{n}\sum_{i \in \mathcal{I}_1}[\txi(U_i)^{\top}(\bX_i-\bmu_{1}(U_i))][\tnu(U_i)^{\top}(\bM^{\ell})^{\top}\bB_i]}\\
		&\qquad \qquad+ \LRabs{\frac{1}{n}\sum_{i \in \mathcal{I}_1}[\txi(U_i)^{\top}(\bX_i-\bmu_{1}(U_i))][\tnu(U_i)^{\top}(\bM^{\ell} - \bM)^{\top}\bB_i]}\\
		&\qquad\leq \LRabs{\frac{1}{n}\sum_{i \in \mathcal{I}_1}[\txi(U_i)^{\top}(\bX_i-\bmu_{1}(U_i))][\tnu(U_i)^{\top}(\bM^{\ell})^{\top}\bB_i]}\\
		&\qquad\qquad +a_n \sqrt{p} L_n n^{-M}\max_{i\in \mathcal{I}_1}\LRabs{\txi(U_i)^{\top}(\bX_i-\bmu_{1}(U_i))},
	\end{align*}
	where the last inequality is true since
	\begin{align*}
	    \LRabs{\tnu(U_i)^{\top}(\bM^{\ell})^{\top}\bB_i} \leq \twonorm{\tnu(U_i)}\LRs{\sum_{j=1}^p ((\bM_{j}^{\ell} - \bM_{j})^{\top} \bB_i)^2}^{1/2}\leq \sqrt{p}L_n a_n n^{-M}.
	\end{align*}
	Let $V_{\ell, i} = [\txi(U_i)^{\top}(\bX_i-\bmu_{1}(U_i))][\tnu(U_i)^{\top}(\bM^{\ell})^{\top}\bB_i^{*}]$. Then we have
	\begin{equation}\label{mat_x_2}
		\begin{aligned}
			&\LRabs{\frac{1}{n}\sum_{i \in \mathcal{I}_1}[\txi(U_i)^{\top}(\bX_i-\bmu_{1}(U_i))][\tnu(U_i)^{\top}(\hbM - \tbM)^{\top}\bB_i]}\\
			\leq & \sup_{\bM \in \boldsymbol{\Xi}}\LRabs{\frac{1}{n}\sum_{i \in \mathcal{I}_1}[\txi(U_i)^{\top}(\bX_i-\bmu_{1}(U_i))][\tnu(U_i)^{\top}\bM^{\top}\bB_i]}\\
			\leq & \max_{1\leq \ell\leq n^{MpL_n}} \LRabs{\frac{1}{n}\sum_{i \in \mathcal{I}_1} V_{\ell,i}}+ a_n \sqrt{p} L_n n^{-M}\sqrt{\log n},
		\end{aligned}
	\end{equation}
	where the last inequality follows from the bound for the maximal of $n$ independent Gaussian random variables and the variance of $\txi(U_i)^{\top}(\bX_i-\bmu_{1}(U_i))$ is bounded almost surely. Notice that
	\begin{align*}
	    \txi(U_i)^{\top}(\bX_i-\bmu_{1}(U_i))|U_i \sim \mathcal{N}\LRs{0, \txi(U_i)^{\top}\bSigma(U_i)\txi(U_i)},
	\end{align*}
	where $\E[\txi(U_i)^{\top}\bSigma(U_i)\txi(U_i)] \lesssim L_n^{-1}$ and $\txi(U_i)^{\top}\bSigma(U_i)\txi(U_i)$ is almost surely bounded. Using the second assertion of Lemma \ref{lemma:normal_moment_bound} with any $\eta > 0$, together with $\twonorm{\bM^{\ell} \tnu(U_i)} \leq \|\bM^{\ell}\|_{F} \twonorm{\tnu(U_i)} \leq \sqrt{p}a_n$, we can verify that
 	\begin{align*}
	    \E\LRm{V_i^2 e^{\eta |V_i|}} &\lesssim p a_n^2 L_n^{-1},
	\end{align*}
	Applying Lemma \ref{bern_ineq}, it is easy to show
	\begin{equation}\label{mat_x_3}
		\max_{1\leq \ell\leq n^{MpL_n}} \LRabs{\frac{1}{n}\sum_{i \in \mathcal{I}_1}[\txi(U_i)^{\top}(\bX_i-\bmu_{1}(U_i))][\tnu(U_i)^{\top}(\bM^{\ell})^{\top}\bB_i]} \lesssim pa_n \sqrt{\frac{L_n\log n}{n}}
	\end{equation}
	with probability at least $1 - n^{-\vartheta L_n p}$. By choosing sufficiently large $M$ and 
	substituting (\ref{mat_x_1})-(\ref{mat_x_3}) into \eqref{eq:A_matrix_bound}, we can finish the proof.
	
	For the case in the matrix-vector-product, we notice that
	\begin{align*}
	    \bnu^{\top}\bA_i\tgamma &= \sqrt{L_n}[\tbtheta(U_i)^{\top}(\bX_i-\bmu_{1}(U_i))][\tnu(U_i)^{\top}(\tbM^{\top}\bB_i - \bmu(U_i))]\\
		& + \sqrt{L_n}[\tbtheta(U_i)^{\top}(\bX_i-\bmu_{1}(U_i))][\tnu(U_i)^{\top}(\hbM - \tbM)^{\top}\bB_i].
	\end{align*}
	By utilizing the same chaining technique and applying Lemma \ref{bern_ineq}, we can prove
	\begin{equation*}
	    \left|\frac{1}{n}\sum_{i \in \mathcal{I}_1} \bnu^{\top}\bA_i \tgamma\right|\lesssim a_n L_n p\sqrt{\frac{\log n}{n}},
	\end{equation*}
    holds with probability at least $1 - n^{-\vartheta pL_n}$. Then we can finish the proof.
\end{proof}

\subsection{Proof of Lemma \ref{lemma_matrix_bound_5}}\label{proof:lemma_matrix_bound_5}
\begin{proof}
    We prove Lemma \ref{lemma_matrix_bound_5} under the good event $\mathcal{A}$ defined in Section \ref{proof:lemma_matrix_bound_5}. We prove the bound for the operator norm of the matrix, and $\ell_2$ norm bound for the matrix-vector-product is similar.
	It suffices to show that for any fixed $\bnu, \bxi \in \mathbb{S}^{pL_n - 1}$ such that
	\begin{equation*}
		\LRabs{\frac{1}{n}\sum_{i\in \mathcal{I}_1} \bnu^{\top}\bG(U_i)\bxi - \frac{1}{n}\sum_{i\in \mathcal{I}_2} \bnu^{\top}\bG(U_i)\bxi} \leq CL_n^{3/2}p\sqrt{\frac{\log n}{n}} \LRs{\sqrt{\frac{L_n\log n}{n}}+L_n^{-d}},
	\end{equation*}
	holds with probability at least $1 - n^{-\vartheta pL_n}$. To simplify notations, we denote $\bdelta(U_i) = \bmu_1(U_i)-\bmu_2(U_i)$ and write $\hmu(u) = \hbM^{\top}\bB(u)$, where the $j$-th column of $\hbM$ equals to $(\halpha_{1j}+\halpha_{2j})/2$. Recall the definition of $\bG(U_i)$, we have
	\begin{align*}
	    \bG(U_i) &= \LRs{\bdelta(U_i)\LRs{\widehat{\bmu}(U_i) - \bmu(U_i)}^{\top}}\otimes \LRs{\bB_i\bB_i^{\top}}\\
	    &=\underbrace{\LRs{\bdelta(U_i)\LRs{\hbM^{\top}\bB(U_i) - \tbM^{\top}\bB(U_i)}^{\top}}\otimes \LRs{\bB_i\bB_i^{\top}}}_{\bG_1(U_i)}\\
	    &+\underbrace{\LRs{\bdelta(U_i)\LRs{\tbM^{\top}\bB(U_i) - \bmu(U_i)}^{\top}}\otimes \LRs{\bB_i\bB_i^{\top}}}_{\bG_2(U_i)}.
	\end{align*}
	It yields that
	\begin{equation}\label{mat_proof_1}
		\begin{aligned}
			&\LRabs{\frac{1}{n}\sum_{i\in \mathcal{I}_1} \bnu^{\top}\bG(U_i)\bxi - \frac{1}{n}\sum_{i\in \mathcal{I}_2} \bnu^{\top} \bG(U_i)\bxi}\\
			 \leq & \underbrace{\LRabs{\frac{1}{n}\sum_{i\in \mathcal{I}_1} \bnu^{\top}\bG_1(U_i)\bxi - \frac{1}{n}\sum_{i\in \mathcal{I}_2} \bnu^{\top}\bG_1(U_i)\bxi}}_{\Pi_1}\\
			 + & \underbrace{\LRabs{\frac{1}{n}\sum_{i\in \mathcal{I}_1} \bnu^{\top}\bG_2(U_i)\bxi - \frac{1}{n}\sum_{i\in \mathcal{I}_2} \bnu^{\top}\bG_1(U_i)\bxi}}_{\Pi_2}.
		\end{aligned}
	\end{equation}
    Denote $\tnu(U_i) = (\bnu_{(1)}^{\top}\bB_i^{*},...,\bnu_{(p)}^{\top}\bB_i^{*})^{\top}$ and $\txi(U_i) = (\bxi_{(1)}^{\top}\bB_i^{*},...,\bxi_{(p)}^{\top}\bB_i^{*})^{\top}$. Then we have
    \begin{align*}
        \bnu^{\top}\bG_1(U_i)\bxi = \LRs{\tnu(U_i)^{\top}\bdelta(U_i)}\LRs{\txi(U_i)^{\top}\LRs{\hbM - \tbM}^{\top}\bB_i} =: T_i\LRs{\hbM - \tbM}.
    \end{align*}
    Here we use the same notation $\boldsymbol{\Xi}$ in Section \ref{proof:lemma_matrix_bound_4}. From Proposition \ref{pro_mean_bound}, we know that
	\begin{equation}\label{mat_proof_2}
		\begin{aligned}
			\Pi_1 &= \LRabs{\frac{1}{n}\sum_{i\in \mathcal{I}_1}T_i\LRs{\hbM - \tbM} - \frac{1}{n}\sum_{i\in \mathcal{I}_2}T_i\LRs{\hbM - \tbM}}\\
			&\leq \sup_{\bM \in \boldsymbol{\Xi}} \LRabs{\frac{1}{n}\sum_{i\in \mathcal{I}_1}T_i\LRs{\bM} - \frac{1}{n}\sum_{i\in \mathcal{I}_2}T_i\LRs{\bM}}\\
			& = \sup_{\bM \in \boldsymbol{\Xi}} \LRabs{\frac{1}{n}\sum_{i\in \mathcal{I}_1}\LRs{T_i\LRs{\bM} - \E\LRm{T_i\LRs{\bM}}} - \frac{1}{n}\sum_{i\in \mathcal{I}_2}\LRs{T_i\LRs{\bM} - \E\LRm{T_i\LRs{\bM}}}}\\
			& \leq \underbrace{\sup_{\bM \in \boldsymbol{\Xi}} \LRabs{\frac{1}{n}\sum_{i\in \mathcal{I}_1}T_i(\bM) - \E\LRm{T_i(\bM)}}}_{\Pi_{11}}+ \underbrace{\sup_{\bM \in \boldsymbol{\Xi}} \LRabs{\frac{1}{n}\sum_{i\in \mathcal{I}_2}T_i(\bM) - \E\LRm{T_i(\bM)}}}.
		\end{aligned}
	\end{equation}
	In fact, the second equality holds since the randomness of $T_i(\bM)$ is from $U_i$ and 
	$U_i$ is independent of $Y_i$. Using the facts $\twonorm{\bB_i} \leq \sqrt{L_n}$, $\twonorm{\bdelta(U_i)} \leq \delta_p$, $\twonorm{\txi(U_i)} \leq 1$ and $\twonorm{\tnu(U_i)} \leq 1$, and following the chaining arguments in Section \ref{proof:lemma_matrix_bound_4}, we have
	\begin{equation}\label{mat_proof_3}
		\begin{aligned}
			\Pi_{11}&\leq \max_{1\leq \ell \lesssim n^{MpL_n}}\LRabs{\frac{1}{n}\sum_{i\in \mathcal{I}_1}T_i\LRs{\bM^{\ell}} - \E\LRm{T_i\LRs{\bM^{\ell}}}}+ a_n \sqrt{p} L_n n^{-M},
		\end{aligned}
	\end{equation}
	where $a_n = O(\sqrt{L_n\log n/n} + L_n^{-d})$.	In addition, notice that
	\begin{align*}
	    \LRs{\E\LRm{T_i\LRs{\bM^{\ell}}}}^2 &\leq \E[(\tnu(U_i)^{\top}\bdelta(U_i))^2] \E[(\txi(U_i)^{\top}(\bM^{\ell})^{\top}\bB_i^{*})^2]\\
	    &\leq \E\LRm{\twonorm{\tnu(U_i)}^2\twonorm{\bdelta(U_i)}^2}\E\LRm{\twonorm{\bM^{\ell} \txi(U_i)}^2\twonorm{\bB_i^*}^2}\\
	    &\leq p \delta_p^2 a_n^2 \E\LRm{\twonorm{\tnu(U_i)}^2}\\
	    & \leq p \delta_p^2 a_n^2 \E\LRm{\sum_{j=1}^p\twonorm{\bnu_j}^2 \twonorm{\bB_i^*}^2}\\
	    &\lesssim p \delta_p^2 a_n^2 L_n^{-1},
	\end{align*}
	and
	\begin{align*}
	    \E\LRm{T_i^2\LRs{\bM^{\ell}}} &= \E\LRm{\LRs{\tnu(U_i)^{\top}\bdelta(U_i)}^2\LRs{\txi(U_i)^{\top}(\bM^{\ell})^{\top}\bB_i^{*}}^2}\\
	    &\leq p a_n^2 \E\LRm{\LRs{\tnu(U_i)^{\top}\bdelta(U_i)}^2}\\
	    &\lesssim p\delta_p^2 a_n^2L_n^{-1}.
	\end{align*}
	In fact, we used $\twonorm{\bM^{\ell}\txi(U_i)}^2 = \|\bM^{\ell}\|_{F}^2\twonorm{\txi(U_i)}^2\leq p a_n$ and $\twonorm{\bB_i}^2 \leq L_n$ in the relations above. Thus we have
	\begin{equation*}
		\begin{aligned}
			&\E\LRm{\LRs{T_i\LRs{\bM^{\ell}} - \E \LRm{T_i\LRs{\bM^{\ell}}}}^2 e^{\eta|T_i\LRs{\bM^{\ell}} - \E[T_i\LRs{\bM^{\ell}}]|}}\\
			\leq &  2e^{\eta|\E [T_i\LRs{\bM^{\ell}}]|}\LRl{\E\LRm{T_i^2\LRs{\bM^{\ell}} e^{\eta|T_i\LRs{\bM^{\ell}}|}} + \LRs{\E \LRm{T_i\LRs{\bM^{\ell}}}}^2\E[e^{\eta|\E \LRm{T_i\LRs{\bM^{\ell}}}|}]}\\
			\lesssim & \E\LRm{T_i^2\LRs{\bM^{\ell}}} + \LRs{\E\LRm{T_i\LRs{\bM^{\ell}}}}^2\\
			\leq & 2\E\LRm{\LRs{\tnu(U_i)^{\top}\bdelta(U_i)}^2\LRs{\txi(U_i)^{\top}\LRs{\bM^{\ell}}^{\top}\bB_i}^2}\\
			\lesssim & p L_n a_n^2\E\LRm{\LRs{\tnu(U_i)^{\top}\bdelta(U_i)}^2}\leq  pL_n\delta_p^2a_n^2 \E\LRm{\twonorm{\tnu(U_i)}^2}\\
			\leq & pL_n\delta_p^2a_n^2 \sum_{j=1}^p \E\LRm{\LRs{\bnu_{(j)}^{\top}\bB_i^*}^2}\lesssim  p\delta_p^2a_n^2\sum_{j=1}^p\twonorm{\bnu_{j}}^2 = p\delta_p^2a_n^2.
		\end{aligned}
	\end{equation*}
	where the second inequality holds since $T_i\LRs{\bM^{\ell}}$ is bounded. In accordance with Lemma \ref{bern_ineq}, we obtain that 
	\begin{equation}\label{mat_proof_4}
		\Prob\LRs{\Pi_{11}\lesssim  pa_nL_n^{3/2}\sqrt{\frac{\log n}{n}} + 2a_n \sqrt{p} L_n n^{-M}} \geq 1 - n^{-\vartheta pL_n}.
	\end{equation}
	Similar bound also holds for $\Pi_{12}$, in conjunction with (\ref{mat_proof_2})-(\ref{mat_proof_4}) and proper choice for $M$, we are guaranteed that
	\begin{equation}\label{mat_proof_5}
		\Prob\LRs{\Pi_1 \lesssim p a_n L_n^{3/2} \sqrt{\frac{\log n}{n}}} \geq 1 - n^{-\vartheta pL_n}.
	\end{equation}
	Now let $W_i = \LRs{\tnu(U_i)^{\top}\bdelta(U_i)}\LRs{\txi(U_i)^{\top}[\tbM^{\top}\bB_i-\bmu(U_i)]}$, then it follows from the dependence between $U_i$ and $Y_i$ that
	\begin{align*}
	    \Pi_2 \leq \LRabs{\frac{1}{n}\sum_{i \in \mathcal{I}_1} W_i - \E[W_i]} + \LRabs{\frac{1}{n}\sum_{i \in \mathcal{I}_2} W_i - \E[W_i]}.
	\end{align*}
	Since $W_i$ is bounded and $\sup_{u}|\tbM_{j}^{\top}\bB(u) - \mu_{j}(u)| \lesssim L_n^{-d} \leq a_n$, we have
	\begin{align*}
	    \E\LRm{\LRs{W_i - \E[W_i]}^2 e^{\eta |W_i - \E[W_i]|}}&\lesssim \E[W_i^2] + (\E[W_i])^2\\
	    &\lesssim \delta_p^2  \E[\twonorm{\tbM^{\top}\bB_i-\bmu(U_i)}^2\twonorm{\txi(U_i)}^2]\\
	    &\leq \delta_p^2 \E[\sup_{u}\twonorm{\tbM^{\top}\bB(u)-\bmu(u)}^2\twonorm{\txi(U_i)}^2]\\
	   &\lesssim   \delta_p^2 p a_n^2 L_n^{-1}.
	\end{align*}
	Applying Lemma \ref{bern_ineq}, we get
	\begin{equation}\label{mat_proof_6}
		\Prob\LRs{\Pi_2 \lesssim p a_n \sqrt{\frac{\log n}{n}}} \geq 1 - n^{-\vartheta pL_n}.
	\end{equation}
	Plugging (\ref{mat_proof_6}) and (\ref{mat_proof_5}) into (\ref{mat_proof_1}), we finish the proof Lemma \ref{lemma_matrix_bound_3}.
\end{proof}

\begin{algorithm}[tb]
   \caption{ISTA with backtracking line-search}
   \label{alg:ISTA}
    \begin{algorithmic}
   \STATE {\bfseries Input:} Initial point $\gamma_0 \in \mathbb{R}^{p L_n}$, number of iterations $T$, shrinking rate $\rho \in (0, 1)$, initial step size $\eta_0 \in (0, 1)$.
 
   \FOR{$t = 0,1,...,T-1$}
   \STATE Compute the gradient: $\nabla g(\bgamma^t) = \bD_n \bgamma^t - \bb_n$.
   \STATE Find the smallest nonnegative integer $i_t$ such that with $\eta = \rho^{i_t} \eta_{t-1}$
   \begin{align*}
       g(Q_{\eta}(\bgamma^t)) \leq g(\bgamma^t) + (Q_{\eta}(\bgamma^t) - \bgamma^t)^{\top}\nabla g(\bgamma^t) + \frac{1}{2\eta}\twonorm{Q_{\eta}(\bgamma^t) - \bgamma^t}^2.
   \end{align*}
   \STATE Set $\eta_{t} = \rho^{i_t} \eta_{t-1}$ and update $\bgamma^{t+1} = Q_{\eta_{t}}(\bgamma^t)$.
   \ENDFOR
   \STATE {\bfseries Output:} The final solution $\bgamma^{\top}$.
    \end{algorithmic}
\end{algorithm}
\section{Iterative Shrinkage Thresholding Algorithm}\label{appen:ISTA}
Next we take ISTA as an example to illustrate the optimization procedure to solve \eqref{lasso_2}. Denote $g(\bgamma) = \frac{1}{2}\bgamma^{\top}\bD_n\bgamma - \bb_n^{\top}\bgamma$. Given a point $\bgamma \in \mathbb{R}^{p L_n}$, ISTA approach updates the solution through solving the following subproblem
\begin{align*}
    \bgamma_{+} &= Q_{\eta}(\bgamma)\\
    &= \arg\min_{\boldsymbol{z} \in \mathbb{R}^{p L_n}}\LRl{g(\bgamma) + (\boldsymbol{z}-\bgamma)^{\top}\nabla g(\bgamma) + \frac{1}{2\eta}\twonorm{\boldsymbol{z} - \bgamma}^2 + \lambda_n \sum_{j=1}^p \twonorm{\boldsymbol{z}_{(j)}}}\\
    &=\arg\min_{\boldsymbol{z}\in \mathbb{R}^{p L_n}}\LRl{\frac{1}{2\eta}\sum_{j=1}^p \twonorm{\boldsymbol{z}_{(j)} - (\bgamma - \eta \nabla g(\bgamma))_{(j)}}^2 + \lambda_n \twonorm{\boldsymbol{z}_{(j)}}},
\end{align*}
where $\eta$ is the step size. The solution of the subproblem is given by the soft-thresholding operator, that is
\begin{align*}
    (\bgamma_{+})_{(j)} = \frac{(\bgamma - \eta \nabla g(\bgamma))_{(j)}}{\twonorm{(\bgamma - \eta \nabla g(\bgamma))_{(j)}}}\max\LRl{0, \twonorm{(\bgamma - \eta \nabla g(\bgamma))_{(j)}} - \eta\lambda_n}.
\end{align*}
The step size $\eta$ is determined by a backtracking line-search \citep{beck2009fast} such that
\begin{align*}
    g(\bgamma_{+}) \leq g(\bgamma) + (\bgamma_{+} - \bgamma)^{\top}\nabla g(\bgamma) + \frac{1}{2\eta}\twonorm{\bgamma_{+} - \bgamma}^2.
\end{align*}
Another simple choice for $\eta$ is $1/\twonorm{\bD_n}$, whereas it usually leads a very small step sizes and slow convergence \citep{qin2013efficient}. For the ease of reference, we provide the detailed procedure in Algorithm \ref{alg:ISTA}.

\end{appendices}

\end{document}